\documentclass[aoas,preprint]{imsart}

\RequirePackage[OT1]{fontenc}
\RequirePackage{amsmath,amsthm,amssymb}
\usepackage{graphicx,psfrag,epsf}
\usepackage{enumerate}
\RequirePackage{natbib}
\usepackage{url} 
\usepackage{hyperref}
\usepackage{array}
\usepackage{color}
\usepackage{xcolor}
\usepackage{bm}
\usepackage{enumerate}
\usepackage{float}
\usepackage[section]{placeins}
\usepackage{multirow}
\usepackage{picinpar}
\usepackage{pifont}
\usepackage{algorithm}
\usepackage{algorithmic}
\usepackage[toc,page]{appendix}

\usepackage{xr}

\newenvironment{psmallmatrix}
  {\left(\begin{smallmatrix}}
  {\end{smallmatrix}\right)}

\newtheorem{lemma}{Lemma}[section]
\newtheorem{corollary}[lemma]{Corollary}

\newtheorem{theorem}[lemma]{Theorem}

\newtheorem{assumption}{Assumption}[section]

\newcommand{\D}{\mathcal{D}}
\newcommand{\C}{\mathcal{C}}
\newcommand{\bcov}{\mathbf{Cov}}
\newcommand{\bcor}{\mathbf{Cor}}
\newcommand{\bE}{\mathbf{E}}
\newcommand{\bX}{\mathbf{X}}
\newcommand{\bY}{\mathbf{Y}}
\newcommand{\J}{\mathcal{J}}
\newcommand{\bOmega}{\mathbf{\Omega}}
\newcommand{\bP}{\mathbf{P}}
\newcommand{\bSigma}{\mathbf{\Sigma}}
\newcommand{\bvar}{\mathbf{Var}}
\newcommand{\bzero}{\mathbf{0}}
\newcommand{\bone}{\mathbf{1}}

\makeatletter
\newcommand{\printfnsymbol}[1]{%
  \textsuperscript{\@fnsymbol{#1}}%
}
\makeatother

\newcommand{\tin}{\text{in}}
\newcommand{\tout}{\text{out}}

\startlocaldefs
\numberwithin{equation}{section}
\theoremstyle{plain}
\endlocaldefs

\newcommand{\ignore}[1]{}
\begin{document}

\begin{frontmatter}

\title{Two-sample tests for repeated measurements of  histogram objects with applications to wearable device data}
\runtitle{Two-sample tests for object data with repeated measures}

\begin{aug}
\author[A]{\fnms{Jingru} \snm{Zhang}\ead[label=e1]{jingru.zhang@pennmedicine.upenn.edu}},
\author[B]{\fnms{Kathleen R.} \snm{Merikangas}\ead[label=e2]{merikank@mail.nih.gov}},
\author[A]{\fnms{Hongzhe} \snm{Li}\ead[label=e3]{hongzhe@pennmedicine.upenn.edu}\thanks{equal contribution}}
\and
\author[A]{\fnms{Haochang} \snm{Shou}\ead[label=e4]{hshou@pennmedicine.upenn.edu}\printfnsymbol{1}}
\runauthor{Zhang, Merikangas, Li and Shou}
\affiliation[A]{University of Pennsylvania}
\affiliation[B]{National Institutes of Health}
\address[A]{Division of Biostatistics\\
Department of Biostatistics, Epidemiology and Informatics\\
University of Pennsylvania Perelman School of Medicine\\
Philadelphia, Pennsylvania 19104 \\
USA \\
EMAIL: jingru.zhang@pennmedicine.upenn.edu\\
EMAIL: hongzhe@pennmedicine.upenn.edu\\
EMAIL: hshou@pennmedicine.upenn.edu
}
\address[B]{Genetic Epidemiology Research Branch \\
National Institute of Mental Health \\
National Institutes of Health \\
Bethesda, MD 20892\\
USA \\
EMAIL: merikank@mail.nih.gov
}
\end{aug}


\begin{abstract}
Repeated observations have become increasingly common in biomedical research and longitudinal studies. For instance, wearable sensor devices are deployed to continuously track physiological and biological signals from each individual over multiple days. It remains of great interest to appropriately evaluate how the daily distribution of biosignals might differ across disease groups and demographics. Hence these data  could be formulated as multivariate complex object data such as probability densities, histograms, and observations on a tree. Traditional statistical methods would often fail to apply as they are sampled from an arbitrary non-Euclidean metric space. In this paper, we propose novel non-parametric graph-based two-sample tests for object data with repeated measures. A set of test statistics are proposed to capture various possible alternatives. We derive their asymptotic null distributions under the permutation null.  These tests exhibit substantial power improvements over the existing methods while controlling the type I errors under finite samples as shown through simulation studies. The proposed tests are demonstrated to provide additional insights on the location, inter- and intra-individual variability of the daily physical activity distributions in a sample of studies for mood disorders. 

\end{abstract}


\end{frontmatter}

\section{Introduction}b
Repeated measures are frequently obtained to capture the within-individual variation and enhance the data reproducibility. 
For example, studies that examine physical activities (PA) using accelerometers often observe individuals' 24hr activity profiles repeatedly over several days or weeks \citep{burton2013activity,krane2014actigraphic,de2017actigraphic}. Within each day, the physical accelerations during movement are recorded with a high frequency and processed into a time series of activity intensity metrics such as activity counts, vector of magnitude (VM) or Euclidean norm minus one (ENMO) over certain epoch lengths (e.g. 5-, 15- or 60-seconds). Commonly extracted markers from accelerometry data include the total amount of PA such as total log activity intensities and step counts \citep{varma_total_2018}, and time spent in different activity intensity levels. In particular, proportion of time spent in sedentary behavior (SB), light (LPA) and moderate-to-vigorous physical activity (MVPA) have been reported to meaningfully correlated with physical and mental functioning and health \citep{faurholt2012differences,de2017actigraphic,murray_measuring_2020}.  However, there remain several known limitations in these traditional PA endpoints. First, metric such as time spent in SB, LPA and MVPA reduces the continuous activity profiles into a composition of only three discrete categories, resulting in great loss of the rich information captured by the densely measured raw accelerometry data. In fact, MVPA might be relatively sparse in a largely sedentary population and are less sensitive to meaningful clinical differences within the population. Secondly, these variables are determined with a priori selected cutpoints. Yet there is a lack of consensus of cutpoints for data collected across study populations (e.g., children vs. adults, diseased vs. healthy individuals), type of devices, wearing positions (e.g., hip vs. wrist)\citep{schrack_assessing_2016, leeger2019}. It has also been reported that the recording frequency, the choice of epoch length and wear time algorithms during processing steps could significantly vary the endpoints and potentially lead to inconsistent conclusions \citep{Bandaplos2016}. Hence it remains challenging to compare findings across studies with these traditionally derived metrics.

Instead of relying upon a few discrete categories defined by relatively arbitrary cutpoints, recently increasing attentions have been paid to model the continuous distribution of the raw daily activity intensities \citep{keadle_impact_2014, schrack_assessing_2016, yang2020quantlet}. In this paper, we also take the daily activity histogram for each individual as the observed outcome and develop statistical methods that compare density objects between groups. As an illustration, we plotted the observed activity data from one individual over four days (two weekdays and two weekends) in Figure \ref{fig:realden} from the National Institute of Mental Health (NIMH) Family Study of Spectrum Disorders \citep{merikangas_independence_2014, merikangas2019jama}. Their time-specific activity intensities at one-minute intervals are shown on the top row and the corresponding histograms of activity intensities in log-transformed scale are shown on the bottom.  As Figure \ref{fig:realden} shows, despite the overall similarity in the time-specific activity patterns across days, the evident shifts in schedules from weekdays to weekends might not be of biological interests. Hence a second advantage of directly modeling the daily activity distributions is that it avoids the need for registering time stamps across days \citep{wrobel_2019}. 

\begin{figure}[H]
    \includegraphics[width=\textwidth]{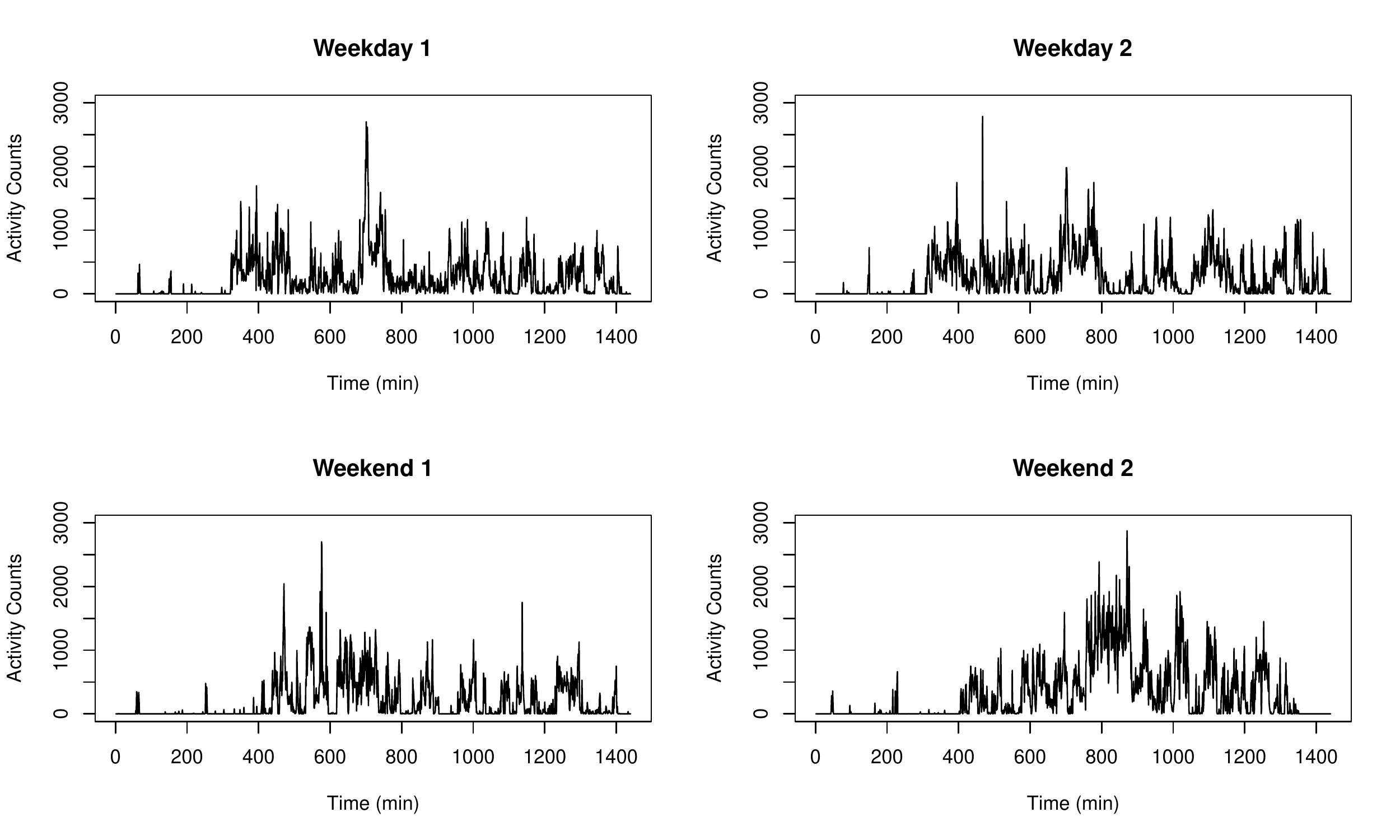}
    \includegraphics[width=\textwidth]{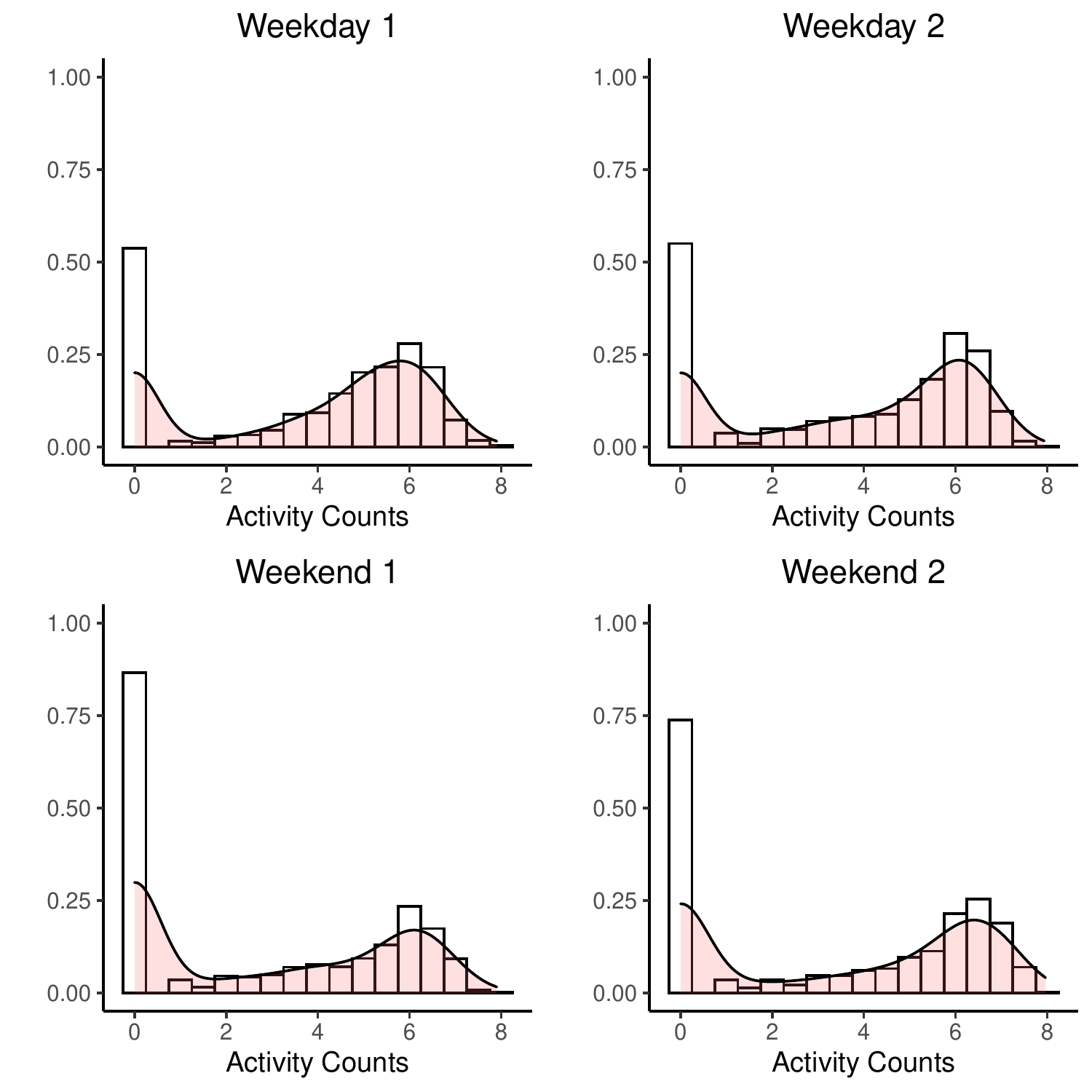}
    \caption{activity intensities for a randomly chosen individual over 4 days. Top: trends of activity intensities; bottom: histograms and densities  of activity intensities.}\label{fig:realden}
\end{figure}

We consider the problem of testing whether the activity density functions or distributions are significantly different between individuals from various clinical groups. As previously noted, the conventional representations of time spent in different levels of PA are derived from discretized distributions using pre-determined cutpoints. To minimize the loss of information, we will be working with the entire probability densities of the continuous daily activity intensities. Our challenges are two folds. First, probability densities, as characterized by the empirical histogram of the daily physical activity intensities, are non-Euclidean and hence many traditional two-sample test statistics are no longer applicable. Second, physical activity tracking over multiple days result in repeated probability densities. As far as we know, there are few existing method that could handle within-individual dependency in the complex object data. 

While two-sample testing for mutivariate objects in Euclidean space or even infinite dimensional space have been studied extensively in the statistics literature, fewer tools are available for two-sample testing when the data are samples of density or distributional functions.  To deal with a wide range of data types, nonparametric tests are preferable. \cite{friedman1979} proposed the first practical test as an extension of the Wald-Wolfowitz runs test to multivariate data. This framework has been extended to other graph-based testing methods. For example, \cite{rosenbaum2005} used the minimum distance pairing (MDP);  \cite{schilling1986} and \cite{henze1988} adopted the nearest neighbor graph (NNG);  \cite{chen2017new} and \cite{chen2018weighted} proposed a generalized edge-count test and a weighted edge-count test to address the problems under scale alternatives and unequal sample sizes, respectively. Recently, an extension of analysis of variance for metric space valued data objects was proposed by \cite{dubey2019frechet}, where Fr{\'e}chet mean and variance are used to construct the test statistic. \cite{yang2020quantlet} proposed quantlets as basis functions to approximate the quantile function objects in a regression setting.

However, most of these existing tests for object data assume that the observations are independent, which cannot be directly applied to repeated measures of object data where within-individual observations are correlated.  One simple way to deal with this issue is to apply these tests to the average of the within-individual measures and convert the problem into a standard two-sample test for independent object observations \citep{dawson1993size}. However, it is not trivial to define the average of non-Euclidean object data. In addition, taking averages oversimplifies the true complexity of data and ignores the within-individual variability that could also be clinically relevant when studying individuals' behaviors and mood \citep{murray_measuring_2020}.

We propose a new nonparametric testing framework for density data with repeated measures.  This framework builds upon graph-based two-sample testing methods that  are flexible and require few assumptions \citep{chen2017new, chen2018weighted}. In particular, to take into account the repeated nature of the data,  we consider the between-individual and within-individual similarity graphs defined via the Wasserstein distances between two density functions. Based on the constructed graph, we define several test statistics that are powerful for various possible alternatives, including difference in population Fr{\'e}chet means, Fr{\'e}chet variances and Fr{\'e}chet covariance.  A new permutation null distribution is considered using the between-individual and within-individual similarity graphs. We also derive the asymptotic null distributions of these statistics under the permutation null, facilitating their applications to large data sets.  

We evaluate the proposed test statistics using simulations and compare the power with several competing tests developed for density data. Our approaches are used in an extensive analysis to evaluate the effects of age, body mass index and types of mood disorders on daily activities in the NIMH family study population.

\section{Non-parametric tests for density  functions with repeated measures based on a similarity graph}\label{sec:newtests}
\subsection{A permutation null distribution for density data with repeated measures}\label{sec:newperm}

To analyze the repeated measurements of activity data, we treat the observed activity densities over $l$ days from each individual as a vector of outcome.  We assume that individuals are divided into two groups with $\mathcal{X}_1,\cdots,\mathcal{X}_{n_1}$ representing density objects for $n_1$ individuals from group 1  and  $\mathcal{Y}_1,\cdots,\mathcal{Y}_{n_2}$ representing densities for $n_2$ individuals from group 2. For a given individual $u$ from group 1, we have 
 $\mathcal{X}_u=(X_{u1},\cdots,X_{ul})$ representing each of the $l$ days' activity densities. Similarly, for individual $v$ from group 2,    $\mathcal{Y}_v=(Y_{v1},\cdots,Y_{vl})$.
 
We assume that each individual density $X_{ui}$ and  $Y_{vj}$ $(u=1, 2,\cdots, n_1; v=1,2,\cdots, n_2; i,j=1,2,\cdots,l)$ belongs to space $\D$, where $\D$ represents a class of one-dimensional densities such that $\int_R u^2 f(u)\mathrm{d}u <\infty$ for $f \in \D$. For any two random densities $\bf{X}, \bf{Y} \in \D$, we define $d_W$ to be the Wasserstein metric as 
 $$d_W^2(\bX, {\bf Y})=\int_R \{T(u)-u\}^2 \bX(u) \mathrm{d}u$$
where $T=F_{\bf{Y}}^{-1} \circ F_{\bf{X}}$ is the optimal transport, and $F_{\bf{X}}$ and  $F_{\bf{Y}}$ are the distribution functions of $\bf{X}$ and $\bf{Y}$, respectively.

We further assume that $X_{ui}$ and $X_{uj}$ have identical distribution function $F_1$, but might be correlated; similarly $Y_{vi}$ and $Y_{vj}$ have the same distribution $F_2$. The vector of $l$-day densities $\mathcal{X}_u$, $u=1,\cdots, n_1$, however, are independently and identically distributed across individuals according to $\tilde P_1$. $\mathcal{Y}_v$, $v=1,\cdots, n_2$ are i.i.d according to $\tilde P_2$.

For a random density $\bf{X}\in \D$ from group 1,  
we define the corresponding group-level Fr{\'e}chet mean  $\mu_{F1}$ and Fr{\'e}chet variance $V_{F1}$ as 
 \[
 \mu_{F1}={\arg\min}_{x\in \D}E\{d_W^2(x,\bX)\}, ~V_{F1}=E\{d_W^2(\mu_{F1},\bX)\}.
 \]
Similarly $\mu_{F2}$ and $V_{F2}$ represent the  Fr{\'e}chet mean  and Fr{\'e}chet variance for a random density $\bf{Y}$ from group 2. Given a vector of random densities whose elements are dependent with each other, we define their Wasserstein covariance following the framework in \cite{petersen2019wasserstein}. Specifically, for two random densities $\bX_i$ and $\bX_j$ from group 1, the Wasserstein covariance is defined as 
 \begin{eqnarray*}
	Cov_{F1,ij} = E[\int_R \{F^{-1}_{X_i}(u)-\mu_{F1}(u)\}\{F^{-1}_{X_j}(u)-\mu_{F1}(u)\} \mathrm{d} u],
\end{eqnarray*}
where the expectation was taken over the joint distribution of $\bX_i$ and $\bX_j$.
	Similarly, $Cov_{F2,ij}$ denotes the  Wasserstein covariance between $\bY_i$ and $\bY_j$ from group 2.

We are interested in testing the null hypothesis $H_0: \tilde P_1=\tilde P_2$, which implies that the $N=n_1+n_2$ samples are from the same distribution.  Based on our motivating examples, group differences in physical activity distributions could occur in mean $\mu_{F1}\ne \mu_{F2}$, between-individual variability $V_{F1}\ne V_{F2}$ or within-individual variability among repeated observations $Cov_{F1,ij}\ne 	Cov_{F2,ij}$ for at least one $(i,j), i\ne j$ pair. For a given test, any of such alternatives should lead to rejection of the null when the sample sizes are large enough.  Instead of imposing any parametric assumptions, we propose a set of non-parametric test statistics based on a similarity graph constructed using pairwise Wasserstein distance as detailed in Section \ref{sec:stat} and a permutation procedure to capture these various possible alternatives. The permutation procedure treats the repeated measures from the same individual as the permutation unit. Specifically, the permutation is done by randomly assigning $n_1$ individuals out of the total $N$ individuals to group 1 and the rest to group 2. If an individual is assigned to group 1, then the repeated measures of the individual are labeled as observations from group 1. Note that we do not require equal correlation or exchangeability within individual among repeated observations since our data are observed sequentially over time. However, to ensure the exchangeability across individuals under the null $H_0$, we do require that the number of repeated observations is the same for all the individuals.  
In the following, when there is no further specification, we use $\bP,~\bE,~\bvar$ and $\bcov$ to denote probability, expectation, variance and covariance, respectively under this permutation null distribution. An illustration of the data structure and distribution assumptions are presented in Figure \ref{fig:alt}.

\subsection{Graph-based statistics for data with repeated measures}\label{sec:stat}
Our proposed test statistics are constructed  from similarity graph \citep{chen2017new}  that includes  both a within-individual graph   and a between-individual graph in order to take into account repeated measures. 
To construct the graph based on the Wasserstein distance  $d_W$,  we pool all repeated measures from a total of $N=n_1+n_2$ individuals, and construct a similarity graph $G$ as a minimum spanning tree (MST). MST  is a spanning tree that connects all observations that minimizes the sum of the total distances of the edges in the tree. In particular, a $m$-MST is the union of the $1$st MST, $\cdots$, $m$th MST, where the $1$st MST is the MST and the $j$th ($j>1$) MST is a spanning tree connecting all observations that minimizes the sum of distances across edges, but individual to the constraint that this spanning tree does not contain any edge from  the previous $1$st MST, $\cdots$, and $(j-1)$th MST. Since the graph-based statistics are usually more powerful under a slightly denser graph \citep{friedman1979}, we choose $9$-MST for our similarity graph $G$ in our simulation studies and real application, following the recommendation by \cite{chen2018weighted}. Based on the similarity graph $G$ of  all the observations, we further divide its edges into two parts. If an edge connects two observations from the same individual,  it belongs to the within-individual similarity graph $G_\tin$, otherwise, it belongs to the between-individual similarity graph $G_\tout$ (see Figure \ref{fig:exG} for an illustration). 

\begin{figure}[!h]
    \includegraphics[height=0.78\textheight]{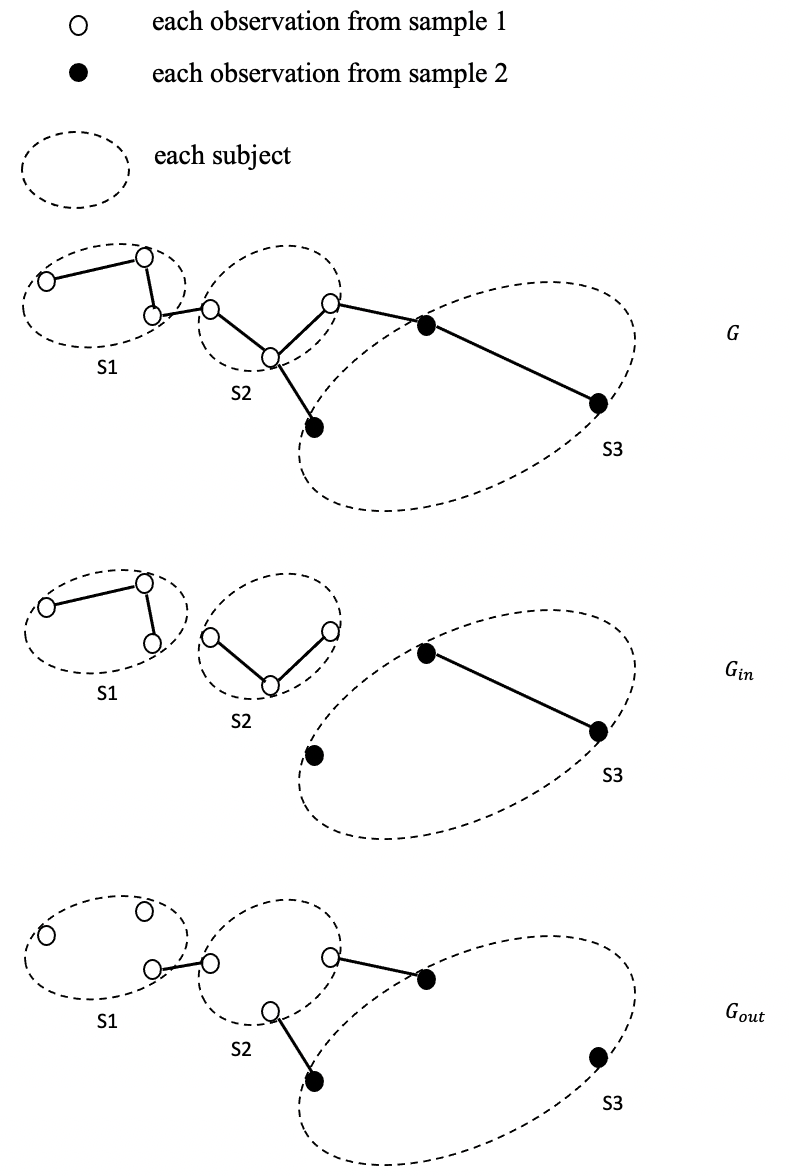}
    \caption{An example of similarity graph $G$, within-individual gragh $G_\tin$ and between-individual graph $G_\tout$ for three individuals with repeated measures. }\label{fig:exG}
\end{figure}

Given a constructed graph $G$, we let $D=(D_{uv})_{N\times N}$ be a symmetric matrix, where $D_{uv}$ denotes the number of edges between individuals $u$ and $v$ in $G$, and let $D_u=\sum_{v\neq u}D_{uv}$ be the total number of edges connecting individual $u$ and others. The total number of edges in $G$ is denoted by $|G|$. Furthermore, let $g_i$ be an indicator function that takes value 1 when node $i$ belongs to an individual from group 1, and 2 otherwise. We denote an edge in $G$ by the indices of the nodes that are connected by the edge such as $e=(u,v)$. Define
\begin{align*}
&R_{\tout,k} = \sum_{(i,j)\in G_{\tout}}I(g_i=g_j=k), 
&R_{\tin,1} = \sum_{(i,j)\in G_{\tin}}I(g_i=g_j=1).
\end{align*}
Here 
$R_{\tout,k}$  is the number of between-individual edges in $G_{\tout}$ that connect observations belonging to the same group $k$, $k=1,2$.  $R_{\tin,1}$ is  the number of within-individual edges in $G_{\tin}$ from group 1. 

To accommodate various alternatives to the null hypothesis, we consider six different test statistics presented in Table \ref{test}. The six test statistics are defined based on different functions of $R_{\tout,1}$, $R_{\tout,2}$  and $R_{\tin,1}$ and their expectations and variances calculated under the permutation null. These different test statistics are developed for testing the same null $H_0: \tilde P_1=\tilde P_2$, but their statistical power depends on specific alternatives, as summarized in Figure \ref{fig:alt}. For each of the test statistics, under $H_0$ and a fixed graph $G$, one could randomly shuffle the group assignments for all individuals to estimate their corresponding null distributions. 

\begin{figure}[!h]
	\includegraphics[width=\textwidth]{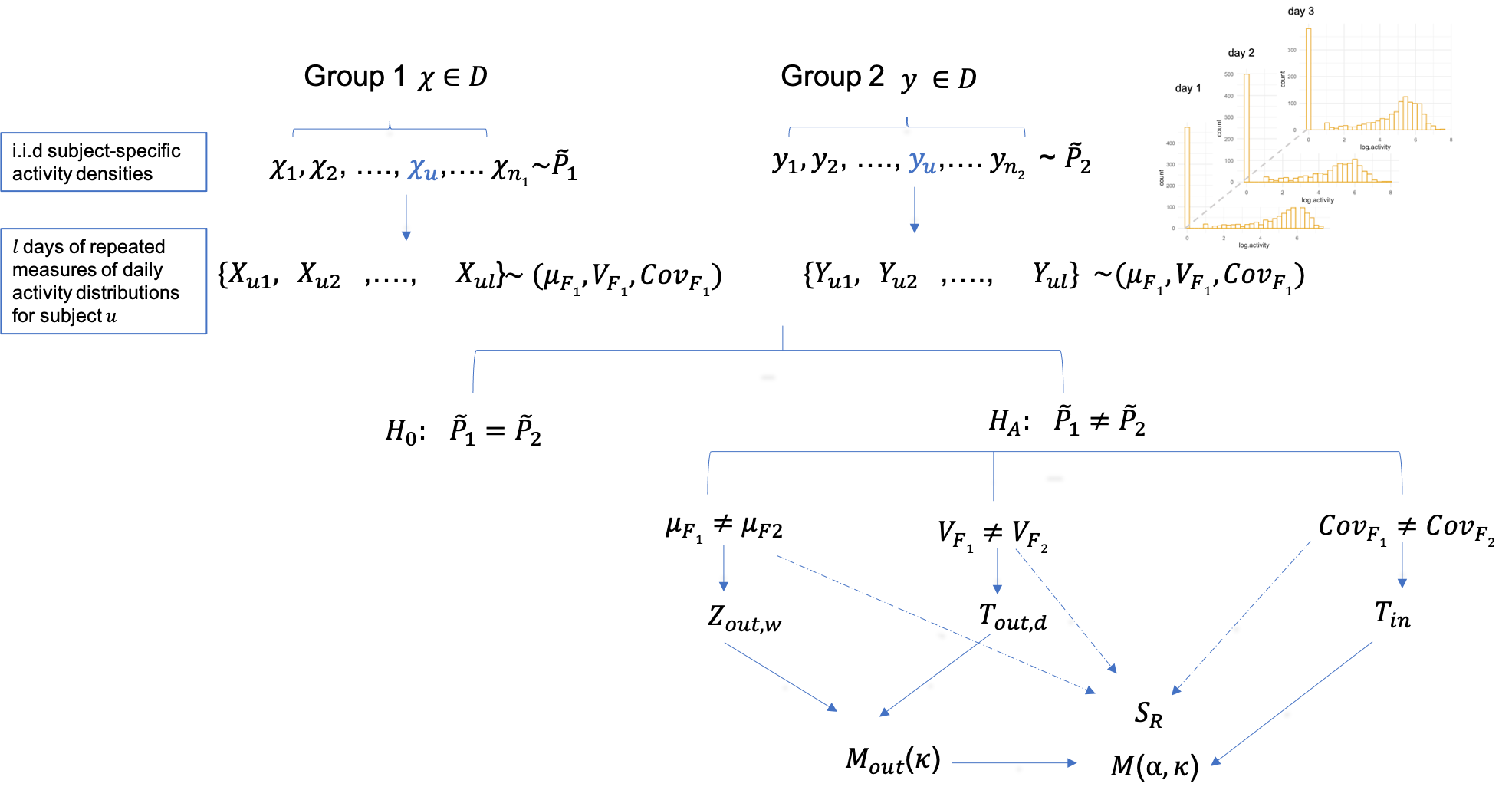}
	\caption{An overview of the repeated data structure, the null hypothesis and various alternatives that each of the proposed test statistics are most suitable for. }\label{fig:alt}
\end{figure}

Specifically, $T_{\tin}$ builds upon the contrast of within-individual edge counts between group 1 and group 2, holding the total number of edge counts to be constant. Hence it captures the covariance among the repeatedly observed densities. Rejecting $H_0$ based on $T_{\tin}$ implies that $Cov_{F1}\ne Cov_{F2}$, suggesting that the group difference occurs in the amount of day-to-day variability in daily activity distributions.  

The next three test statistics $Z_{\tout,w}$, $T_{\tout,d}$ and $M_{\tout}(\kappa)$ are developed to capture the group difference in the marginal distribution of individual activity densities $F_1$ and $F_2$. In particular, $Z_{\tout,w}$ evaluates the mean difference between the two groups, and rejecting $H_0$ implies that $\mu_{F1}\ne \mu_{F2}$. Similarly, $T_{\tout,d}$ examines the group difference in between-individual variances and rejects $H_0$ when $V_{F1} \ne V_{F2}$.  Finally, $M_{\tout}(\kappa)$ combines the comparison in both mean and variance by taking the maximum of the two. Note that these statistics are adapted from the existing formulations from \cite{zhang2017graph}. However this is not a direct application from the previous work due to the existence of repeated observations per individual. Our novelty lies in expanding the similarity graph to include both $G_{out}$ and $G_{in}$ that allow more than one edges connecting between any pair of nodes. As a consequence, we construct the statistics based on the weighted edge counts and new derivations are needed to obtain the asymptomatic distributions. The two final statistics $S_R$ and $M(\alpha,\kappa)$ combine the previously defined statistics in a weighted fashion and flexibly capture differences occurred in both the between-individual distributions $F_1$ and $F_2$ as well as the within-individual covariance.

\begin{table}[!h]\linespread{1.3} 
	\caption{Proposed test statistics for the difference of two population distributions of the density functions. }\label{test}
	\begin{tabular}{c}
		\hline
		\multicolumn{1}{c}{\underline{Within-individual test}}\\
\multicolumn{1}{l}{Wasserstein covariance difference}\\
$
T_{\tin} = |Z_{\tin}|,
 Z_{\tin} = \frac{R_{\text{in},1} - \bE(R_{\text{in},1})}{\sqrt{\bvar(R_{\text{in},1})}}.
$\\
	\multicolumn{1}{c}{\underline{Between-individual test}}\\
 \multicolumn{1}{l}{Mean difference}\\
$
Z_{\tout,w} = \frac{(n_2-1)R_{\tout,1}+(n_1-1)R_{\tout,2}-\bE((n_2-1)R_{\tout,1}+(n_1-1)R_{\tout,2})}{\sqrt{\bvar((n_2-1)R_{\tout,1}+(n_1-1)R_{\tout,2})}}.
$\\
 \multicolumn{1}{l}{Variance difference}\\
$
T_{\tout,d}=|Z_{\tout,d}|,
Z_{\tout,d} = \frac{R_{\tout,1}-R_{\tout,2}-\bE(R_{\tout,1}-R_{\tout,2})}{\sqrt{\bvar(R_{\tout,1}-R_{\tout,2})}}.
$\\
\multicolumn{1}{l}{Overall difference}\\
$
M_{\tout}(\kappa) = \max\{T_{\tout,d},\kappa Z_{\tout,w}\}.
$\\
	\multicolumn{1}{c}{\underline{Joint between and within-individual test}}\\
	
\multicolumn{1}{l}{Sum-type test}\\
$
S_R=\left( \begin{array}{c}
	R_{\text{out},1} - \bE(R_{\text{out},1}) \\ R_{\text{out},2} - \bE(R_{\text{out},2}) \\ R_{\text{in},1}-\bE(R_{\text{in},1})
\end{array} \right)^T \bSigma^{-1}\left(\begin{array}{c}
	R_{\text{out},1} - \bE(R_{\text{out},1}) \\ R_{\text{out},2} - \bE(R_{\text{out},2}) \\ R_{\text{in},1}-\bE(R_{\text{in},1})
\end{array} \right),
$\\
 where $\bSigma=\bvar((R_{\text{out},1},R_{\text{out},2},R_{\text{in},1})^T)$.\\
\multicolumn{1}{l}{Max-type test}\\
$
M(\alpha,\kappa) =  \max\{T_\tin,\alpha M_\tout(\kappa)\}.
$\\
\hline
 \end{tabular}
\end{table}
 
 \subsection{Analytic expressions of the new statistics}\label{sec:expression}
In the following, we first derive the exact analytic expressions for the expectation and variance of $(R_{\tout,1},$ $R_{\tout,2},R_{\tin,1})$, so that the proposed test statistics in Section \ref{sec:stat} can be computed efficiently.
The analytic expressions are provided in the following theorem. The detailed proof could be found in the Supplementary Material. 

\begin{theorem}\label{th:expression} Under the permutation null, 
the analytic expressions of the expectation of 
 $(R_{\tout,k},R_{\tin,1})^T, k=1,2$ are 
\begin{align*}
 \bE(R_{\tout,k}) = |G_{\tout}|\frac{n_k(n_k-1)}{N(N-1)}, 
 \bE(R_{\tin,1}) = |G_{\tin}|\frac{n_1}{N}, 
\end{align*}
the analytic expressions of the variances 
\begin{align*}
& \bvar(R_{\tout,k}) = \frac{n_1n_2(n_1-1)(n_2-1)}{N(N-1)(N-2)(N-3)}\times \\
&\hspace{2.5cm}\Bigg\{\frac{1}{2}\sum_{u\neq v}D_{uv}^2+\frac{n_k-2}{N-n_k-1}\left(\sum_uD_u^2-\frac{4|G_{\tout}|^2}{N}\right) 
  -\frac{2}{N(N-1)}|G_{\tout}|^2\Bigg\}, \\
&\bvar(R_{\tin,1}) = \frac{n_1n_2}{N(N-1)}\left(\sum_uD_{uu}^2-\frac{|G_{\tin}|^2}{N}\right), 
\end{align*}
and the analytic expressions of covariance
\begin{align*}
& \bcov(R_{\tout,1}R_{\tout,2}) = \frac{n_1n_2(n_1-1)(n_2-1)}{N(N-1)(N-2)(N-3)}\times \\
&\hspace{3.5cm}\Bigg\{\frac{1}{2}\sum_{u\neq v}D_{uv}^2-\left(\sum_uD_u^2-\frac{4}{N}|G_{\tout}|^2\right) 
 -\frac{2}{N(N-1)}|G_{\tout}|^2\Bigg\}, \\
& \bcov(R_{\tout,k}R_{\tin,1}) = (-1)^{k+1} \frac{n_1n_2(n_k-1)}{N(N-1)(N-2)}\left(\sum_{u=1}^ND_{uu}D_u-\frac{2}{N}|G_{\tin}||G_{\tout}|\right).
\end{align*}
\end{theorem}

Using the results of Theorem \ref{th:expression}, we can check that  under the permutation null, 
$\bE(Z_{\tout,w}) = \bE(Z_{\tout,d}) = \bE(Z_{\tin}) = 0,$
$ \bvar(Z_{\tout,w}) = \bvar(Z_{\tout,d}) = \bvar(Z_{\tin}) = 1, $
and $\bcov(Z_{\tout,w},Z_{\tout,d}) = 0,$
$\bcov(Z_{\tout,w},Z_{\tin}) = 0$. In addition, 
$$\bcov(Z_{\tout,d},Z_{\tin}) = \frac{\sum_{u=1}^ND_{uu}D_u-\frac{2}{N}|G_{\tin}||G_{\tout}|}{\sqrt{(\sum_{u=1}^ND_u^2-\frac{4|G_{\tout}|^2}{N})(\sum_{u=1}^ND_{uu}^2-\frac{|G_{\tin}|^2}{N})}}. 
$$

It is straightforward to verify that the statistic $S_R$ can be rewritten in the following form
$$S_R=(Z_{\text{out},w}, Z_{\text{out},d}, Z_{\text{in}}) \bOmega^{-1}(Z_{\text{out},w}, Z_{\text{out},d}, Z_{\text{in}})^T,$$ 
 where $\bOmega=\bvar\{(Z_{\text{out},w}, Z_{\text{out},d}, Z_{\text{in}})^T\}$. The detailed proof is provided in the Supplementary Material. 

\section{Asymptotic distribution under the permutation null}\label{sec:asym}
The critical values of the test statistics can be determined by performing permutations of individual nodes as stated in Section \ref{sec:newperm}. However, such a permutation procedure is often time consuming. To make the tests computationally more efficient, we have  derived the asymptotic null distributions of the test statistics. In Section \ref{sec:simu}, we examine how the critical values obtained from asymptotic results agree with those obtained through permutation directly in finite sample settings. 

Before stating the theorem, we need to define a few additional notations for the similarity graph $G$. Denote by $\C_u$ the set of repeated measures belonging to individual $u$.
For an edge $e=(i,j)\in G_\tout$,  $i\in\C_u,~j\in\C_v~(u\neq v)$, let $A_{\tout,e}$ be the subset of edges that share nodes with $e$ as
\begin{align*}
A_{\tout,e} = &
\{e\}\cup\{e'=(k,l)\in G_\tout: k\in\C_u\cup\C_v \text{ or }l\in\C_u\cup\C_v\}\cup\{e''\in G_\tin:e'' \\
& \text{ in individuals } u \text{ or }v\}.
\end{align*}
For an edge $e=(i,j)\in G_\tin,~i,j\in\C_u$, let
\begin{align*}
A_{\tin,e}
= \{e'\in G_\tout: \text{ one endpoint of }e' \in\C_u\}\cup\{e''\in G_\tin: e''\text{ in individual } u \}.
\end{align*}
Define
\begin{align*}
& A_e = A_{\tout,e}I(e\in G_\tout)+A_{\tin,e}I(e\in G_\tin),\\
& B_{\tout,e} = \cup_{\tilde e\in A_{\tout,e}}A_{\tilde e},~ B_{\tin,e}
= \cup_{\tilde e\in A_{\tin,e}}A_{\tilde e}, \\
& B_e = B_{\tout,e}I(e\in G_\tout)+B_{\tin,e}I(e\in G_\tin).
\end{align*}

To derive the asymptotic null distribution  of the proposed  test statistics,  we assume $n_1=O(N)$, $n_2=O(N)$ and $l=O(1)$. In addition, the following conditions are needed:\\
\indent {Condition 1 (C1)}: $|G_\tout|,~|G_\tin|=O(N). $\\
\indent {Condition 2 (C2)}: $
\sum_uD_u^2-4|G_\tout|^2/N,~\sum_uD_{uu}^2-|G_\tin|^2/N = O(N)
$\\
\indent {Condition 3 (C3)}: $\sum_{e\in G_\tout}|A_{\tout,e}||B_{\tout,e}| = o(N^{1.5}).
	$\\
Here, we  use $a=O(b)$ to denote that $a$ and $b$ are of the same order and $a=o(b)$ to denote that $a$ is of a smaller order than $b$.

Condition C1 requires that the number of the edges  in $G_{\tout}$ and $G_{\tin}$ is in the same order as $N$. 
Condition C2 guarantees  that $(R_{\tout,1},R_{\tout,2},R_{\tin,1})^T$ does not degenerate asymptotically. Since 
\[
\sum_uD_u^2-4|G_\tout|^2/N=\sum_u\left(D_u-\frac{2|G_\tout|}{N}\right)^2,
\]
\[
\sum_uD_{uu}^2-|G_\tin|^2/N = \sum_u\left(D_{uu}-\frac{|G_\tin|}{N}\right)^2,
\]
if $D_u-2|G_\tout|/N=O(1)$ and $D_{uu}-|G_\tin|/N=O(1)$, then Condition C2 is  satisfied.
Condition C3 requires  the number of edges  from an individual in the graph $G$ such  being not too large. A similar condition was needed for graph-based statistics for independent observations  \citep{chen2017new, chen2018weighted}.
Conditions C1 and C2  imply that 
$
\sum_uD_u^2,~\sum_uD_{uu}^2=O(N).
$
In addition, we note that 
$
2|G_\tout|=\sum_{u\neq v}D_{uv}\leq\sum_{u\neq v}D_{uv}^2\leq\sum_uD_u^2 
$
and
$
|G_\tin|=\sum_uD_{uu}\leq\sum_uD_{uu}D_u\leq\sqrt{\sum_uD_{uu}^2}\sqrt{\sum_uD_u^2},
$
which leads to 
\[
\sum_{u\neq v}D_{uv}^2=O(N) ~\text{ and }~\sum_uD_{uu}D_u=O(N).
\]

We assume the following limits exist, 
\[
\lim_{N\rightarrow\infty}\frac{|G_\tout|}{N}=b_1,~\lim_{N\rightarrow\infty}\frac{\sum_uD_u^2}{N}-\frac{4|G_\tout|^2}{N^2}=b_2,~\lim_{N\rightarrow\infty}\frac{\sum_{u\neq v}D_{uv}^2}{N}=b_3,
\]
\[
\lim_{N\rightarrow\infty}\frac{|G_\tin|}{N}=b_4,~\lim_{N\rightarrow\infty}\frac{\sum_uD_{uu}^2}{N}-\frac{|G_\tin|^2}{N^2}=b_5,~\lim_{N\rightarrow\infty}\frac{\sum_{u}D_{uu}D_u}{N}=b_6.
\]
The following theorem  presents  the asymptotic distribution of $(Z_{\tout,w},Z_{\tout,d},Z_{in})^T$
under the permutation null when $N\rightarrow \infty$. 
\begin{theorem}\label{th:asym}
	Under Conditions 1-3 and under the new permutation null distribution, as $N\rightarrow \infty$, $(Z_{\tout,w},Z_{\tout,d},Z_{in})^T$ converges to a multivariate Gaussian distribution with mean $\bzero$ and covariance matrix 
	$$
	\begin{pmatrix}
	1 & 0 & 0 \\
	0 & 1 & \rho_Z\\
	0 & \rho_Z & 1
	\end{pmatrix},
	$$
	where $$\rho_Z=\frac{b_6-2b_1b_4}{\sqrt{b_2b_5}}.$$
\end{theorem}

 Based on Theorem \ref{th:asym}, it is easy to obtain the asymptotic cumulative distribution functions (CDF) of $T_\tin$, $Z_{\tout,w}$, $T_{\tout,d}$, $M_\tout(\kappa)$, $S_R$ and $M(\alpha,\kappa)$ under the permutation null. They are given  in the following  Corollary \ref{corollary:Tin}.

\begin{corollary}\label{corollary:Tin}
	Under Conditions C1-C3,  and under the permutation null distribution, as $N\rightarrow\infty$ the asymptotic CDFs for each of the test statistic are 
{\small	\begin{enumerate}
		\item $P(T_\tin\leq x)\longrightarrow 2\Phi(x)-1$. 
\item  $P(Z_{\tout,w}\leq x)\longrightarrow\Phi(x)$.
\item $P(T_{\tout,d}\leq x)\longrightarrow 2\Phi(x)-1$.
\item $P(M_\tout(\kappa)\leq x)\longrightarrow (1-2\Phi(-x))\Phi(x/\kappa)$ 
	\item  $P(S_R\leq x)\longrightarrow\chi^2(3)$
	\item  $P(M(\alpha,\kappa)\leq x)\longrightarrow\Phi(x/(\alpha\kappa))\bP(-x/\alpha\leq Z_{\tout,d}\leq x/\alpha,-x\leq Z_{\tin}\leq x)$ 
	\end{enumerate}
	}
where  $\Phi(\cdot)$ denotes the CDF of a standard normal distribution.
\end{corollary}
The term $P(-x/\alpha\leq Z_{\tout,d}\leq x/\alpha,-x\leq Z_{\tin}\leq x)$ can be calculated from function \texttt{pmvnorm()} in the \texttt{R} package \texttt{mvtnorm}, where the correlation between $Z_{\tout,d}$ and $Z_{\tin}$ can be estimated using  finite sample estimate, 
\[
\rho_{Z,N} = \frac{\sum_{u=1}^ND_{uu}D_u-\frac{2}{N}|G_{\tin}||G_{\tout}|}{\sqrt{(\sum_{u=1}^ND_u^2-\frac{4|G_{\tout}|^2}{N})(\sum_{u=1}^ND_{uu}^2-\frac{|G_{\tin}|^2}{N})}}.
\]
It is easy to see that $\lim_{N\rightarrow\infty}\rho_{Z,N}=\rho_Z$.

\section{Simulation studies}
\label{sec:simu}
We evaluate the performance of the proposed test statistics $T_{\tin}$, $Z_{\tout,w}$, $T_{\tout,d}$, $M_{\tout}(\kappa)$, $S_R$ and $M(\alpha,\kappa)$ under various simulation settings. 
Under each setting, we compare the results  with the generalized edge-count test ($S$) of \cite{chen2017new} and the Fr\'{e}chet test (Fretest) of \cite{dubey2019frechet}. As far as we know, neither method allows for data with repeated measures and would rely on between-individual distance metrics from a single observation. Simply applying those tests on the individual observations without accounting for within-subject correlation leads to an inflated type 1 error (results omitted). To ensure a fair comparison, we apply these two tests on the subject-level graph constructed based on two definitions of distance metrics that respect the hierarchical structure among the repeated observations. The first distance is chosen to be the Wasserstein distance calculated from each subject's barycenter (average distance).  Alternatively, we use the integrated distance by taking the square root of the total sum square distances across all the $l$ observations for any pair of individuals. We denote the generalized edge-count test and the Fr\'{e}chet test calculated under the first distance metric by $S1$ and Fretest1, and those calculated under the second definition as $S2$ and Fretest2, respectively. More specifically, let $Z_u=(Z_{u1},\cdots,Z_{ul})$ and $Z_v=(Z_{v1},\cdots,Z_{vl})$, where $Z_{uj},~Z_{vj}~(j=1,\cdots,l)$ represent the repeated measures for individuals $u$ and $v$, the two distances are defined as:
\begin{enumerate}
\item Average distance: $d(Z_u,Z_v) = d_W(\tilde Z_u,\tilde Z_v)$, where $\tilde Z_u$ and $\tilde Z_v$ are the barycenters of $Z_u$ and $Z_v$, respectively, i.e.,
\[
\tilde Z_u = \arg\min_{x\in\Omega}\sum_{i=1}^ld_W(x,Z_{ui}),\quad \tilde Z_v = \arg\min_{x\in\Omega}\sum_{i=1}^ld_W(x,Z_{vi}).
\]
\item Integrated distance: $d(Z_u,Z_v)=\sqrt{\sum_{i=1}^ld_W^2(Z_{ui},Z_{vi})}$.
\end{enumerate}

\ignore{For two individuals $u$ and $v$, let $Z_u=(Z_{u1},\cdots,Z_{ul})$ and $Z_v=(Z_{v1},\cdots,Z_{vl})$, where $Z_{uj},~Z_{vj}~(j=1,\cdots,l)$ are the repeated measures for individuals $u$ and $v$, respectively. Denote by $W$ the 2-Wasserstein distance. 
\begin{enumerate}[(1)]
\item Averaged: $d(Z_u,Z_v) = W(\tilde Z_u,\tilde Z_v)$, where $\tilde Z_u$ and $\tilde Z_v$ are the barycenters of $Z_u$ and $Z_v$, respectively, i.e.,
\[
\tilde Z_u = \arg\min_{x\in\Omega}\sum_{i=1}^lW(x,Z_{ui}),\quad \tilde Z_v = \arg\min_{x\in\Omega}\sum_{i=1}^lW(x,Z_{vi}).
\]
We denote $S$ and Fretest by $S1$ and Fretest1, respectively, when using the definition (1).
\item Integrated: $d(Z_u,Z_v)=\sqrt{\sum_{i=1}^lW^2(Z_{ui},Z_{vi})}$.
We denote $S$ and Fretest by $S2$ and Fretest2, respectively, when using the definition (2).
\end{enumerate}
}

Following the recommendations from \cite{zhang2017graph},  when there is no prior knowledge about the type of between-individual difference (i.e., location difference or scale difference), we choose $\kappa=1.14$ for the statistic $M_{\tout}(\kappa)$ and denote it by $M_{\tout}$ for simplicity. For the statistic $M(\alpha,\kappa)$, the parameter $\alpha$ weights the between-individual difference. Here, we let $\kappa=1.14$ and $\alpha=1$, and denote the statistic by $M$ for simplicity.

The general setup for the simulation settings is as following. We generate the observed physical activity density for individual $u$ on day $j$ to be equal to the density function of a $p$-dimensional multivariate normal distribution with mean $\theta_{uj}$ and variance $\omega_u^2I_p$. That is, $X_{uj}=\psi_p(\theta_{uj},\omega_u^2I_p), j=1,\cdots,l$.  We further assume that $\omega_u$ is independent and identically distributed from a uniform distribution $U(\nu_{k1},\nu_{k2})$, with $k=1, 2$ corresponding to group label. $\theta_{uj}~(j=1,\cdots,l)$ are sampled from another multivariate normal distribution $N_p(a_{ku},\sigma^2I_p)$ with individual-specific mean $a_{ku}$. We further assume an exchangeable correlation between $\theta_{uj}'s$, which leads to 
\[
(\theta_{u1}^T,\cdots,\theta_{ul}^T)^T|\mu_{ku}\sim N_{pl}(\mu_{ku},\sigma^2\varrho_k\otimes I_p),
\]
where $\otimes$ denotes the Kronecker product, $\mu_{ku}=(a_{ku}^T,\cdots,a_{ku}^T)^T$ and $\varrho_k=\rho_k\bone_l\bone_l^T+(1-\rho_k)I_l$.
Here, $a_{ku}~(u = 1,\cdots, n) \stackrel{i.i.d}{\sim} N_p(\beta_k,\epsilon_k^2I_p)$. We also consider an exponentially decayed correlation between $\theta_{uj}'s$ with the $(s,t)$-element of $\varrho_k$ being $\rho_k^{|s-t|},~s,t=1,\cdots,l$. The results are similar and are given in the Supplementary Material.

We simulate unbalanced data with $n_1=50,~n_2=80$ individuals for each group and $l=5$ days for each individual.  When applying the proposed statistics, we use the Wasserstein distance to measure the dissimilarity between any two density functions, which can be explicitly calculated. The similarity graph $G$ is constructed by the procedure outlined in Section \ref{sec:stat} with $9$-MST.

\subsection{Simulations for one-dimensional density, $p=1$}\label{sec:density}
We  consider five different parameter settings as listed on the top rows of Table \ref{simu1}. 
All the test statistics are compared in terms of type 1 error and power.  Here, Model (A1) is the null model when there is no difference between the two groups, Models (A2)-(A4) represent the cases where the two groups differ in within-individual covariance, between-individual mean and between-individual variability, respectively. Model (A5) represents the case that differences exist in mean, variance and also in the within-individual covariance.

\begin{table}[!htbp]
	\caption{Parameter values for 5 different simulation settings for comparisons. (A) one-dimensional density functions; (B) 30-dimensional density functions.} \label{simu1}
	\begin{tabular}{l}
		\hline
	\multicolumn{1}{c}{(A) - one-dimensional density functions}\\
 A1:  null model.\\
  $\qquad\rho_1=0.6,~\beta_1=0,~\epsilon_1=1,~\nu_{11}=1,~\nu_{12}=2;$ \\
 $\qquad\rho_2=0.6,~\beta_2=0,~\epsilon_2=1,~\nu_{21}=1,~\nu_{22}=2;~\sigma=1$.\\
 A2: within-individual variability difference in $\rho$.
\\
   $\qquad\rho_1=0,~\beta_1=0,~\epsilon_1=1,~\nu_{11}=1,~\nu_{12}=1.2;$ \\
 $\qquad\rho_2=0.8,~\beta_2=0,~\epsilon_2=1,~\nu_{21}=1,~\nu_{22}=1.2;~\sigma=1$.
\\
A3: between-individual mean difference in $\beta$ and $\nu_{\cdot 1}+\nu_{\cdot 2}$.
\\
 $\qquad\rho_1=0,~\beta_1=0,~\epsilon_1=1,~\nu_{11}=1,~\nu_{12}=1.2;$ \\
 $\qquad\rho_2=0,~\beta_2=0.7,~\epsilon_2=1,~\nu_{21}=0.96,~\nu_{22}=1.16;~\sigma=1$.
\\
A4: between-individual variability difference in $\epsilon$ and $\nu_{\cdot 2}-\nu_{\cdot 1}$.
\\
  $\qquad\rho_1=0,~\beta_1=0,~\epsilon_1=1,~\nu_{11}=1,~\nu_{12}=1.3;$ \\
 $\qquad\rho_2=0,~\beta_2=0,~\epsilon_2=1.1,~\nu_{21}=0.97,~\nu_{22}=1.33;~\sigma=1$.
\\ 
A5: within-individual variability difference in $\rho$, between-individual mean difference in $\beta$ \\ and $\nu_{\cdot 1}+\nu_{\cdot 2}$, 
variance difference in $\epsilon$ and $\nu_{\cdot 2}-\nu_{\cdot 1}$.
\\
  $\qquad\rho_1=0,~\beta_1=0,~\epsilon_1=1,~\nu_{11}=1,~\nu_{12}=1.3;$ \\
 $\qquad\rho_2=0.35,~\beta_2=0.5,~\epsilon_2=1.1,~\nu_{21}=0.97,~\nu_{22}=1.36;~\sigma=1$.
 \\

 \hline
	\multicolumn{1}{c}{(B) - 30-dimensional density functions}\\
 B1:  null model.\\
 $\qquad\rho_1=0.3,~\beta_1=\bzero_p,~\epsilon_1=1,~\nu_{11}=1,~\nu_{12}=2;$ \\
 $\qquad\rho_2=0.3,~\beta_2=\bzero_p,~\epsilon_2=1,~\nu_{21}=1,~\nu_{22}=2;~\sigma=1$.\\
 B2: within-individual variability difference in $\rho$.
 \\
 $\qquad\rho_1=0,~\beta_1=\bzero_p,~\epsilon_1=1,~\nu_{11}=1,~\nu_{12}=1.3;$ \\
 $\qquad\rho_2=0.1,~\beta_2=\bzero_p,~\epsilon_2=1,~\nu_{21}=1,~\nu_{22}=1.3;~\sigma=1$.
 \\
 B3: between-individual mean difference in $\beta$ and $\nu_{\cdot 1}+\nu_{\cdot 2}$.
 \\
 $\qquad\rho_1=0,~\beta_1=\bzero_p,~\epsilon_1=1,~\nu_{11}=1,~\nu_{12}=1.3;$ \\
 $\qquad\rho_2=0,~\beta_2=0.1\bone_p,~\epsilon_2=1,~\nu_{21}=1.2,~\nu_{22}=1.5;~\sigma=1$.
 \\
 B4: between-individual variability difference in $\epsilon$ and $\nu_{\cdot 2}-\nu_{\cdot 1}$.
 \\
 $\qquad\rho_1=0,~\beta_1=\bzero_p,~\epsilon_1=1,~\nu_{11}=1,~\nu_{12}=1.3;$ \\
 $\qquad\rho_2=0,~\beta_2=\bzero_p,~\epsilon_2=1.1,~\nu_{21}=0.8,~\nu_{22}=1.5;~\sigma=1$.
 \\ 
 B5: within-individual variability difference in $\rho$, between-individual mean difference in $\beta$ \\ and $\nu_{\cdot 1}+\nu_{\cdot 2}$, 
 variance difference in $\epsilon$ and $\nu_{\cdot 2}-\nu_{\cdot 1}$.
 \\
 $\qquad\rho_1=0,~\beta_1=\bzero_p,~\epsilon_1=1,~\nu_{11}=1,~\nu_{12}=1.3;$ \\
 $\qquad\rho_2=0.09,~\beta_2=0.1\bone_p,~\epsilon_2=1.03,~\nu_{21}=1,~\nu_{22}=1.5;~\sigma=1$.
 \\
 \hline
\end{tabular}
\end{table}

\begin{table}[!htp]
  \centering
  \caption{Empirical power of the proposed test statistics in the first 6 columns, generalized edge-count test ($S1$, $S2$) and Fr\'{e}chet test (Fretest$1$, Fretest$2$) at 0.05 significance level under the five scenarios denoted by A1--A5. The bold fonts indicate for test with the best power and those with  power is  over 95\% of the best power for each of the models. }\label{tab:den1} 
   \renewcommand\arraystretch{1}
 \begin{tabular*}{1.0\textwidth}{@{\extracolsep{\fill}} lcccccccccc}
    \hline
 & $T_{\tin}$ & $Z_{\tout,w}$ & $T_{\tout,d}$ & $M_{\tout}$ & $S_R$ & $M$ & $S1$ & $S2$ & Fretest1 & Fretest2 \\
 \hline
 & \multicolumn{10}{c}{(A) One-dimensional density}\\
  & \multicolumn{10}{c}{Null model}\\
   A1 & 0.044 & 0.061 & 0.047 & 0.057 & 0.051 & 0.052 & 0.051 & 0.053 & 0.057 & 0.057 \\

 & \multicolumn{10}{c}{Alternative model}\\ 
  A2 & \textbf{0.911} & 0.038 & 0.100 & 0.066 & 0.719 & 0.786 & 0.133 & \textbf{0.950} & 0.423 &  0.048 \\
A3 & 0.048 & \textbf{0.973} &  0.064 & \textbf{0.962} &  \textbf{0.939} & \textbf{0.954} & 0.645 & 0.575 & 0.287 & 0.276 \\
A4 & 0.038 & 0.190 & \textbf{0.911} & \textbf{0.867} & 0.802 & 0.830 & 0.142 & 0.104 & 0.321 & 0.324 \\
A5 & 0.245 &  0.664 & \textbf{0.994} & \textbf{0.995}  & \textbf{0.992} & \textbf{0.994} & 0.422 & 0.616 & 0.583 & 0.298\\
\hline   
 & \multicolumn{10}{c}{(B) 30-dimensional density}\\
    & \multicolumn{10}{c}{Null  model}\\
    B1 & 0.048 & 0.045 & 0.049 & 0.041 & 0.042 & 0.045 & 0.035 & 0.039 & 0.078 & 0.088  \\ 
    & \multicolumn{10}{c}{Alternative model}\\
    B2 & \textbf{0.926} & 0.055 &  0.046 & 0.054 & 0.840 & 0.865 & 0.164 & 0.051 & 0.371 & 0.087 \\
    B3 & 0.054 & \textbf{0.969} & 0.058 & \textbf{0.939} & 0.836 & 0.916 & 0.766 & 0.713 & 0.136 & 0.201  \\
    B4 & 0.143 & 0.273 & \textbf{0.893} & 0.847 & 0.787 & 0.809 & 0.757 & 0.827 & \textbf{0.864} & \textbf{0.883} \\
    B5 & \textbf{0.865} & 0.387 & 0.192 & 0.355 & \textbf{0.853} & \textbf{0.897} & 0.513 & 0.223 & 0.754 & 0.425  \\
    \hline
  \end{tabular*} 
 \end{table}
 
Table \ref{tab:den1} shows  the empirical power of the proposed statistics at $\alpha=0.05$ level based on 1,000 replications. 
Under the null model (A1), all the statistics are able to control the type 1 errors at the nominal level.

As for detecting the group differences in the alternative Models (A2-A5), the power of $S$ and Fr\'{e}chet test is uniformly lower than our proposed statistics except for $S2$ under Model (A2). As expected,  the power of different test statistics depends on the alternative hypothesis.  
In Model (A2), when $\rho_1$ is different from $\rho_2$,  $T_\tin$ shows its superior performance of detecting group differences in covariance among repetaed measures within individuals. 
For Model (A3),  since the difference only happens in the group mean parameters, all the proposed test statistics except $T_{\tin}$ and $T_{\tout,d}$ yield high power. Model (A4)  is designed to examine the power of the tests when the between-individual variability is different  between the two groups. We observe that indeed all the proposed tests except $T_{\tin}$ and $Z_{\tout,w}$ have high power. The results for Model (A5) suggest that $T_{\tout,d}$ works well for detecting group difference in between-individual variability and $Z_{\tout,w}$ is suitable for detecting differences in the between-individual mean. Since there is a smaller difference in $\rho'$s than that under Model (A2), $T_{\tin}$ did not yield high power in this scenario.

\subsection{Simulations for moderate-dimensional density, $p=30$}\label{sec:densitymult}
Although we have mostly been concerned with two-sample testing for one-dimensional probability densities based on a single morality of measurements such as physical activity intensity, it is worth noting that our proposed tests are directly applicable to density objects from multimodal measurements as long as there is a well defined distance metrics. In fact, many wearable devices simultaneously collect multiple markers such as heart rate, respiratory rate in addition to the physical movement, and there is need to compare the joint density distributions of multivariate measures in mobile health research. To illustrate their utility for multivariable density objects with repeated measures, we conduct another set of simulation studies for $p=30$. Our simulation setups are similar to $p=1$ case, where we simulate an unbalanced sample with $n_1=50,~n_2=80$ individuals in each group and $l=5$ repeated measures per individual. All of the statistics are assessed and compared under five different scenarios as listed in Table \ref{simu1}. These 5 models parallel the Models (A1)-(A5), except that we consider density measures for $30$-dimensional variables. 

Table \ref{tab:den1} shows the estimated empirical powers of the proposed statistics at 0.05 significance level based on 1000 simulations. 
Again we observe that all the statistics control the type 1 errors at  the approximate level.  However, the type 1 errors of the Fr\'{e}chet tests are slightly inflated. 


For Models (B2)-(B4), the power of the proposed tests remain similar to those under the one-dimensional setting in Section \ref{sec:density}. As a comparison,  although tests $S1$ and $S2$ can detect the between-individual mean and variance differences (B3, B4), they are not effective for detecting the within-individual variability difference (B2). Fretest1 and Fretest2, on the other hand, work well when only between-individual variance differ as in Model (B4). The results for Model (B5) indicate that the proposed tests $S_R$ and $M$ perform well for the overall difference and is much better than the competing tests $S1$, $S2$, Fretest1 and Fretest2.

Finally, we also perform simulations to examine whether the asymptotic $p$-values could approximate the  $p$-values obtained from 10,000 permutations. The results show that the two $p$-values are very close and the power obtained by the asymptotic $p$-value is similar to that based on the permutation $p$-value for all the proposed test statistics. As sample size increases, the results are almost identical as expected. We omit the details here and present the results in the Supplementary Material, Section C.


\section{Comparisons of physical activity distributions in mood disorder samples}\label{sec:real}
We apply each of the six test statistics to the continuous physical activity measures collected from a subset of the participants from the National Institute of Mental Health (NIMH) Family Study of Spectrum Disorders \citep{merikangas_independence_2014, shou_dysregulation_2017, merikangas2019jama}. In this study,  384 individuals were instructed to wear the Philips Actiwatch devices for about two weeks. 
The daily activity data were processed into 1440 minute-level intensity values each day. Meanwhile, the 384 individuals were interviewed and assessed into four clinical groups based on DSM-IV criteria as: healthy control (HC), major depressive disorders (MDD), type-I bipolar disorders (BPI) and type-II bipolar disorders (BPII). Previous research studies have consistently reported a lower average daytime motor activities among bipolar patients based on summary statistics from physical activity measures \citep{scottJAMA2017, murray_measuring_2020}. Age and body mass index (BMI) are among the other factors are known to be associated with the mean activity levels \citep{schott2007, varma2017}. However, although there were a few papers that suggested potential links between bipolar disorder and interindividual and intraindividual variability in activity patterns, the evidence was much less robust and the extracted markers for quantifying variability was quite heterogeneous\citep{indic2011, pagani2016, RobillardJPN2015}, making it even more challenging to understand the complex disease manifestations. Hence we focus on comparing the continuous physical activity profiles and testing whether mean and variability of the daily physical activity differ across disease groups or by demographic characteristics. To apply our proposed methods, we first estimate the empirical daily probability densities using the observed minute-by-minute activity intensities. Let $Z_{uj}=(z_{uj}^{(1)},\cdots,z_{uj}^{(1440)})^T$ be the vector of ordered 1440 activity intensities for individual $u$ on day $j$. Here $z_{uj}^{(q)}$ represents the empirical $q$th quantile of the probability distribution of activity intensities per day. The Wasserstein distance metric is calculated to quantify the distance between two empirical distributions based on any pairs of $Z_{ui}$ and $Z_{vj}$. 
Since densities are empirically estimated from the ordered values, the Wasserstein distance between densities is equivalent to the Euclidean distance between the two empirical quantiles, i.e., 
\[
d(Z_{ui},Z_{vj}) = \left(\sum_{q=1}^{1440}\left(z_{ui}^{(q)}-z_{vj}^{(q)}\right)^2\right)^{1/2}
\]
We further construct the similarity graph $G$ following the procedure that is introduced in Section \ref{sec:stat} with $9$-MST. As a sensitivity analysis, we also apply the tests using 5-MST, 15-MST and under the maximum mean discrepancy \citep{gretton2012kernel}. The results are similar and are provided in the Supplementary Material, Section E.

\ignore{
Given the richness of this dataset, the following interesting problems are studied.
\begin{itemize}
\item (P1) 
\item (P2) Whether the probability distributions of activity intensities are significantly different among young people with Age $\leq 30$, middle-aged people with $30<\text{Age}\leq 60$ and old people with Age $> 60$.
\item (P3) Whether the probability distributions of activity intensities are significantly different among lean people with BMI $\leq 18.5$, normal people with $18.5<\text{BMI}\leq 30$ and obese people with BMI $> 30$.
\end{itemize}
}

Considering the potential difference in daily routines and movement between weekdays and weekends, we apply the test statistics separately to observations collected on weekdays and weekends with $l=7$ and $l=3$ days, respectively.  For each analysis, the individuals with fewer than the given number of days $l$ are excluded from the analysis. For those with more than $l$ days of observations, a random subset of $l$ days are included in generating the test statistics. Sensitivity analysis was conducted by repeating the random subsetting 1000 times in order to assess the variability in the test results due to choice of days. To summarize the results, we take the $p$-value $p_j, j=1,2,\cdots, 1000$ from each of the 1000 trials, and estimate an overall $p$-value as
\[
\hat p = 1-\frac{2}{1+e^{2\theta}},
\]
where
\[
\theta = \frac{1}{1000}\sum_{j=1}^{1000}\frac{1}{2}\log\left(\frac{1+p_j}{1-p_j}\right).
\]

\subsection{Comparison of  activity densities between healthy individuals and those  with mood disorders} 
We first compare the activity densities among healthy individuals and those with histories of mood disorders in the free-living conditions during the weekdays and weekends. 
Figure  \ref{tab:real1} shows the $p$-values of the pairwise comparisons using the proposed test statistics, the generalized edge-count tests ($S1$, $S2$) and Fr\'{e}chet tests (Fretest1, Fretest2) (detailed $p$-values and the sample sizes for different groups are given in Table 5 of the Supplementary Material). 
We first observe that the differences between diagnostic groups are mostly driven by activity patterns on weekends and no significant difference is observed during the weekdays. In particular, we observe that the healthy individuals have significantly different activity distributions from those with BPI. Among the proposed statistics, $Z_{\tout,w}$ achieves the most significant results when comparing healthy  with BPI and BPII vs BPI. While $T_{\tin}$ and $T_{\tout,d}$ result in nonsignificant large $p$-values and cannot reject the null hypothesis. These results suggest  that  there exist  significant differences in  the population-level mean activity  density between healthy and BPI or between BPI and BPII. This is consistent with findings from the existing literature where BPI patients were found to have lower average activity levels especially in the later of the day \citep{scottJAMA2017, shou_dysregulation_2017} and less time spent in MVPA \citep{chapmanSP2017}. But no significant difference is observed in the variance of activity densities or  in day-to-day variability of the activity density.  
Since all of $M_{\tout}(\kappa)$, $S_R$ and $M(\alpha,\kappa)$ include a $Z_{\tout,w}$ in their definitions,  they are also effective to capture the  mean  difference of activity densities  when $Z_{\tout,w}$ yields a small $p$-value.

\begin{figure}[!h]
	\centering
	\includegraphics[width=0.98\textwidth, height=0.25\textheight]{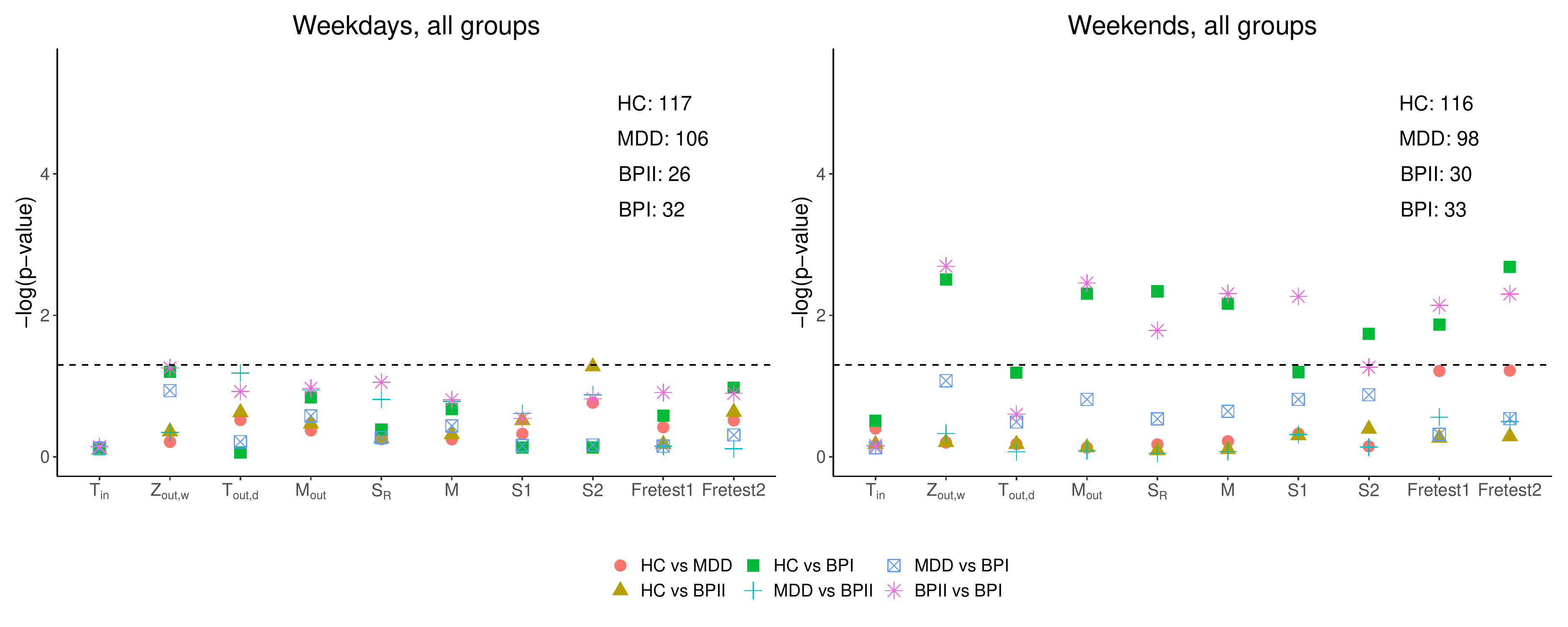}

	\caption{Comparisons of activity distributions among the healthy controls (HC), MDD, BPI and BPII individuals for activities over weekdays and on weekends. For each individual, 7 weekdays and 3 weekends of data are used.  The $-\log$($p$-values) are plotted for each of the proposed test statistics, the generalized edge-count tests ($S1$, $S2$) and Fr\'{e}chet tests (Fretest1, Fretest2). The corresponding sample sizes for each group is presented on the upper right corner. }\label{tab:real1}
\end{figure}

\subsection{Comparison of activity distributions among different age groups}
It is well known that age is associated with the amount of physical activity. For example, \cite{schrack_assessing_2014} found `a 1.3\% decrease per year' in cumulative physical activity counts from mid-to-late life among a elderly population. Similar results have been reported in several other large cohort studies and age groups including NHANES and UK Biobank \citep{varma2017, Viciana2016, dohertyplos2017}. However, few studies has examined how inter- and intra-individual variability in physical activity differ by age. We ask whether the proposed test statistics are able to detect differences in the daily activity densities over different age categories and inform us where the difference lies.  To ensure a proper power with an adequate sample size, we take the two diagnostic groups with the largest sample sizes, the HC and major depressive disorder (MDD), and stratify them into three age groups,   young (age$\le $30), middle age (30$<$age$\le$60) and older age (age$>60$) groups.  We also separately test activity densities from weekdays and weekends. 

\begin{figure}[!h]
	\centering
	\includegraphics[width=0.98\textwidth, height=0.5\textheight]{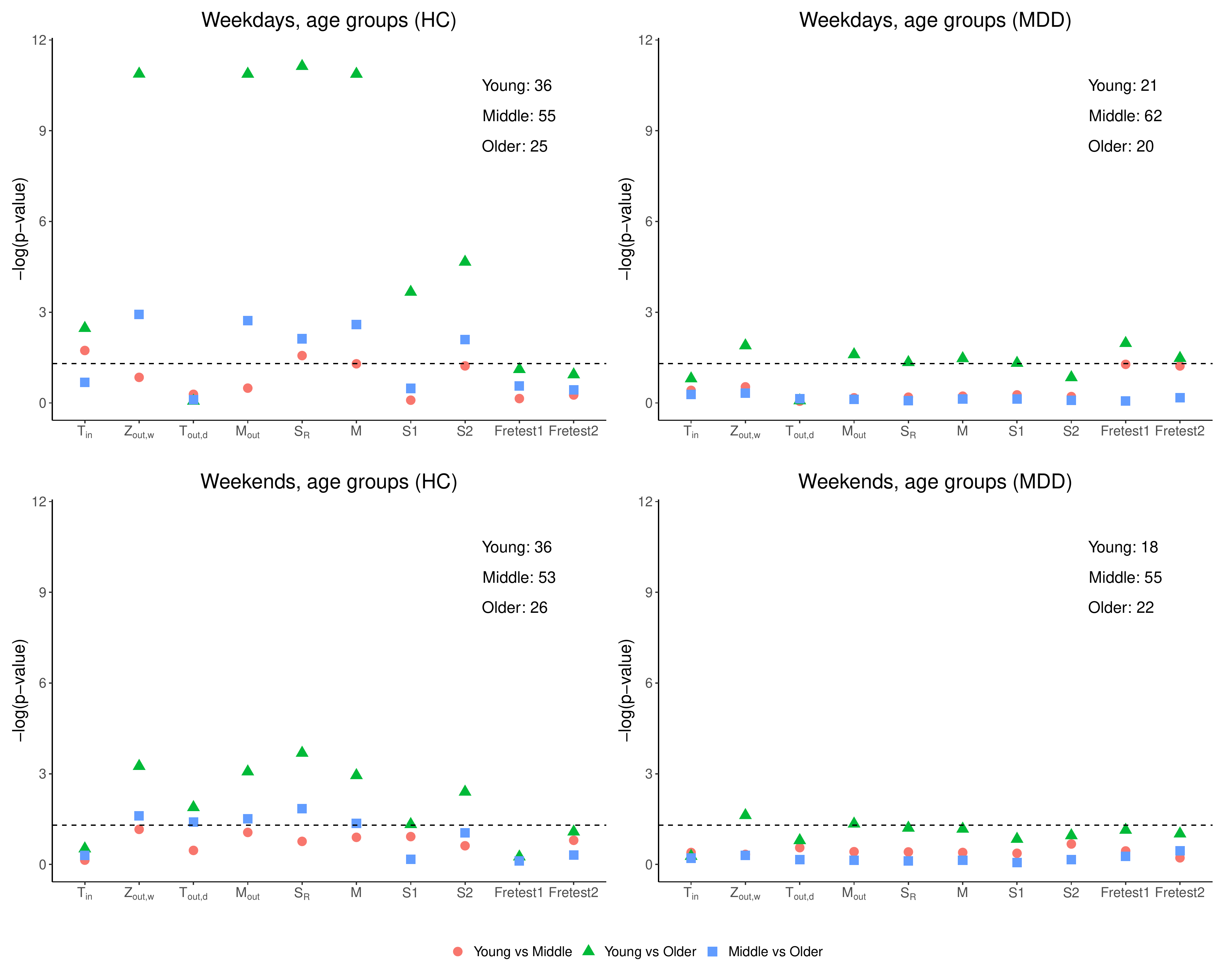}

	\caption{Comparisons of activity distributions in different age groups for young ($\le $30), middle (30$<$age$\le$60) and older age ($>60$) groups. The $-\log$($p$-values) of the proposed test statistics, the generalized edge-count tests ($S1$, $S2$) and Fr\'{e}chet tests (Fretest1, Fretest2) are presented for different comparisons. The corresponding sample sizes for each group is presented on the upper right corner.}\label{tab:real-age1}
\end{figure}

The $p$-values of the proposed test statistics, the generalized edge-count test ($S1$, $S2$) and Fr\'{e}chet test (Fretest1, Fretest2) are shown  in Figure  \ref{tab:real-age1} with detailed $p$-values given  in Table 6 of the Supplementary Material. Overall, among the healthy individuals,  the proposed statistical tests find large differences in the distributions of activity intensities among the three age groups for both weekdays and weekends. Such differences are especially prominent when comparing the young age group or middle age group with the older group during the weekdays.   In contrast, Fr\'{e}chet test fails to detect such differences in most of the comparisons and is only able to capture marginally significant results when comparing the young and older individuals among MDD patients.  The tests $S1$ and $S2$ also show fewer significant results than our proposed tests.  

To further demonstrate the possible gain of power,  we note that among the patients with MDD, only the proposed $Z_{\tout,w}$ test shows statistically significant difference between young and older groups for both weekend and weekday activities. Fr\'{e}chet test shows some difference in activity distributions between young and older groups but only for the weekdays. To confirm the detected differences in the original data, we visualize the density data in Figure \ref{tab:real-age1.1} by projecting them onto a lower-dimensional plots using multidimensional scaling (MDS) based on the Wasserstein distances. The figure clearly shows difference in activity densities between young and older groups for both weekdays and weekends among MDD patients. 

\begin{figure}[!h]
	\centering
	\includegraphics[width=0.95\textwidth]{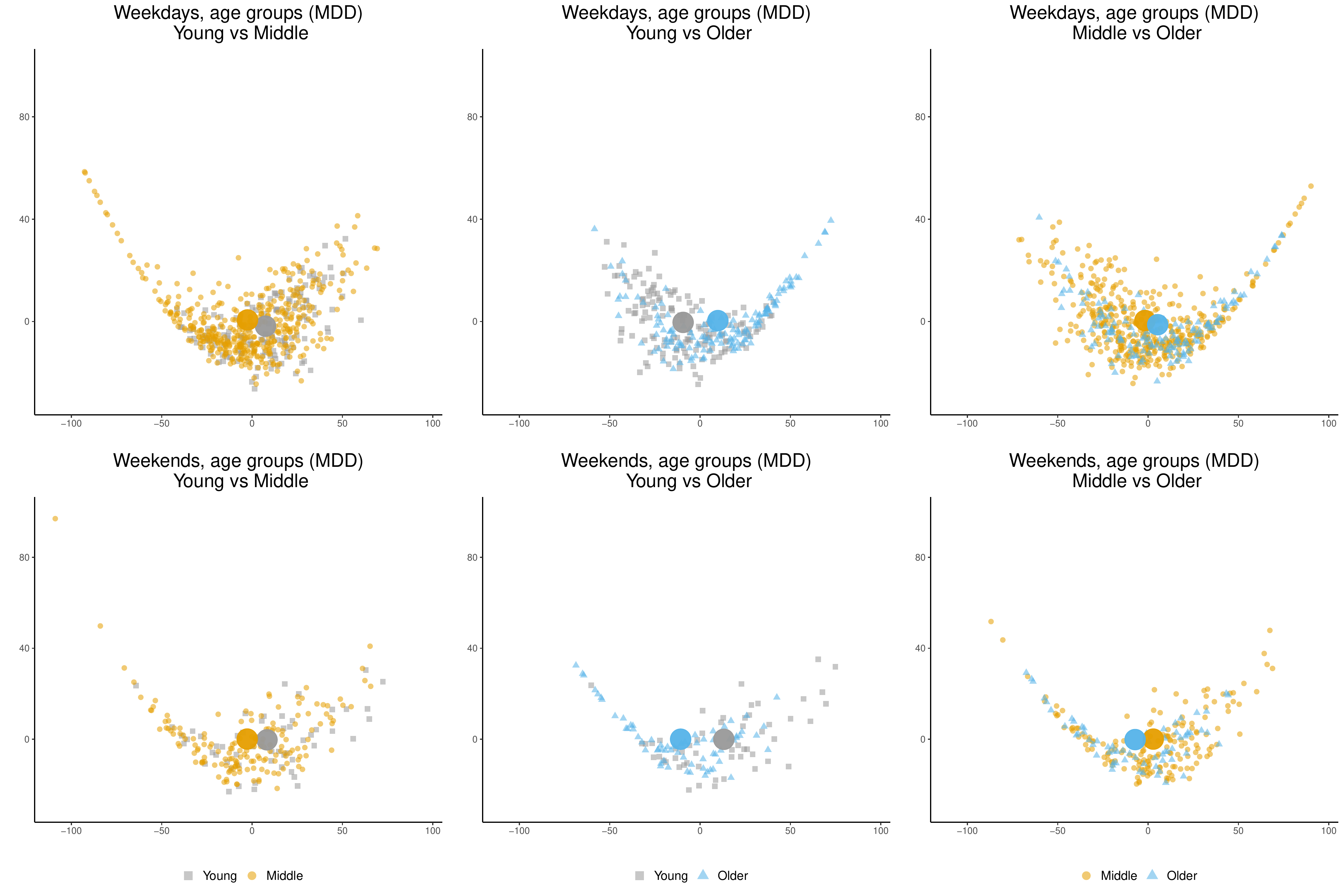}\\
		\caption{Multidimensional scaling (MDS) plots based on the Wasserstain distances to visualize the distribution of activity densities among MDD patients, across three pairwise comparisons by age groups (left, middle, right) and on weekdays (top) and weekends (bottom). 
	}\label{tab:real-age1.1}
\end{figure}

Finally, it is also interesting that $T_{\tin}$ detects significant difference in day-to-day variability  between the healthy young group and older groups. In fact, we obtained a negative value for $Z_{\tin}$ which implies that the younger subjects have larger day-to-day variability than the two groups with older ages. Lastly,  $T_{\tout,d}$  does not yield any significant results, indicating that there is large subject heterogeneity within each age group, yet their scales are comparable.

\subsection{Comparison of activity distributions among different BMI groups}
A third factor that could potentially affect differential physical activity patterns is the body mass index (BMI). We apply our proposed tests to examine difference in daily activity density among individuals  who are lean with BMI$\leq$25 and obese with BMI$>$25 among healthy individuals and those with mood disorders (OTHER).  To control for the age effects, we only consider those individuals with age of 30 years or older.

\begin{figure}[!h]
	\centering
	\includegraphics[width=0.9\textwidth]{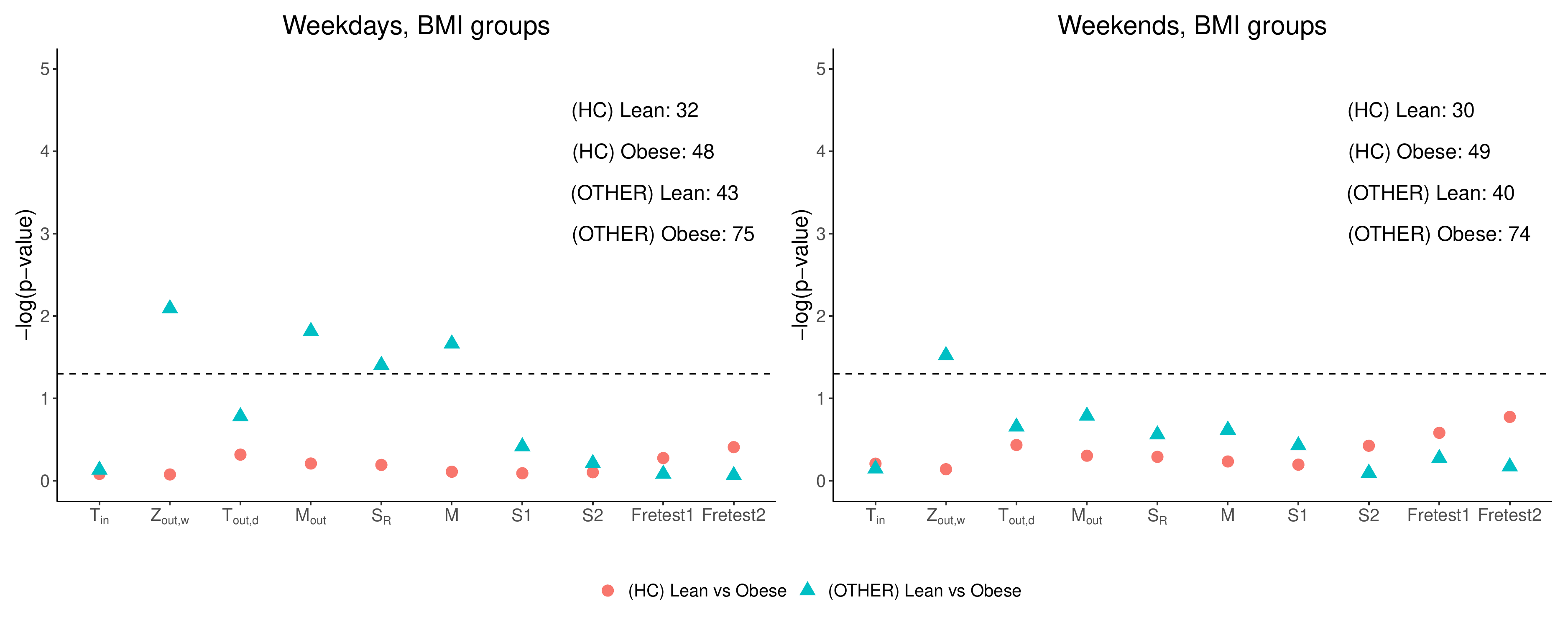}

	\caption{Comparisons of activity distributions by BMI (lean and obese groups).  The $-\log$($p$-values)  of the proposed test statistics, the generalized edge-count tests ($S1$, $S2$) and Fr\'{e}chet tests (Fretest1, Fretest2) are presented for different comparisons. The corresponding sample sizes for each group is presented on the upper right corner. }\label{tab:real-bmi1}
\end{figure}

The results are provided in Figure  \ref{tab:real-bmi1}. Among the healthy individuals,  little difference is  observed in their activity distribution patterns between lean and obese individuals during the weekdays and weekends. When assessing among patients in the OTHER group, we observe some  differences in the mean of the activity distributions  both during  weekdays and weekends. We do not see group differences in the within-individual or between-individual variability. The generalized edge-count test and Fr\'{e}chet test achieve nonsignificant large $p$-values and fail to reject the null hypothesis for all the comparisons. This further demonstrates that our proposed test statistics can detect difference in activities that could be missed by other methods.

\section{Discussion}\label{sec:conclude}
In this paper, we have extended the graph-based two-sample tests for density  data and proposed several  test statistics to account for repeated measure data by considering both the within-individual similarity graph $G_\tin$ and between-individual similarity graph $G_\tout$.  The graph allows for more than one edge between any two individuals, which  extends the existing graph-based testing methods where only one edge between any two individuals is allowed. 
 We have proposed a list of six test statistics that capture  different alternatives that are associated with distributions of density functions,  including differences in mean, inter- and intra-individual variances.    These statistics are constructed based on the similarity graph $G$, which is the union of $G_\tin$ and $G_\tout$. Furthermore, we have developed  the asymptotic null distributions that  can be used to obtain  $p$-values under the permutation null.  The test statistics are easy to calculate and the testing procedures are computationally efficient. Our simulation studies have shown that the proposed test statistics control  the desired type 1 errors and are more powerful than existing distance-based tests that ignore the repeated observations. 

In our analysis of the physical activity measures with repeated observations, we have observed a substantial differences in the day-to-day variability within subject across disease groups and age categories. Such findings have rarely been reported previously. Our proposed tests are able to take into account such within-individual dependency and variability. Compared to the two versions of Fr\'{e}chet tests, we observed increased power in detecting the differences in activity densities. In addition, by comparing results utilizing various proposed test statistics, we are able to further understand the complex data structures and decompose the source of differences between various groups. 

Our proposed permutation procedure treats the entire vector of repeated observations from an individual as the independent unit, which requires that we have the same number of observations for each individual. Otherwise, the within-individual Wasserstein covariance is not well defined. In our analysis of the NIMH physical activity data, we noticed that the results are robust to different subsets of the observations used in our analysis and reported an average through Fisher's transformation. An interesting future research topic is to develop statistical tests that allow for different numbers of repeated observations.

\ignore{
\begin{supplement}
\sname{Supplement to ``Two-sample tests for repeated measurements of histogram objects with applications to wearable device data"} 
\slink[url]{http://www.e-publications.org/ims/support/dowload/imsart-ims.zip}
\sdescription{The supplementary material contains the following: \\
Supplement A provides detailed proof of Theorems \ref{th:expression} and \ref{th:asym} that derive the analytic expressions and the asymptotic distributions of the proposed test statistics. \\
Supplement B provides simulation results for data with an exponentially decayed within-individual correlation.\\
Supplement C includes additional simulation results comparing the asymptotic and permutation $p$-value over 100 simulation runs. \\
Supplement D provides results of real data with 9-MST.\\
Supplement E provides results of real data comparisons when adopting 5-MST, 15-MST as the similarity graph and adopting maximum mean discrepancy as a metric. }
\end{supplement}

\section*{Acknowledgments}
This study was supported in part by the Inter-Agency Agreement from National Institute of Mental Health for Dr. Shou and NIH grants GM129781 and GM123056.
}

\bibliographystyle{imsart-nameyear}
\bibliography{gtests}

\begin{thebibliography}{37}

\bibitem[\protect\citeauthoryear{Banda et~al.}{2016}]{Bandaplos2016}
\begin{barticle}[author]
\bauthor{\bsnm{Banda},~\bfnm{J.~A.}\binits{J.~A.}},
  \bauthor{\bsnm{Haydel},~\bfnm{K.~F.}\binits{K.~F.}},
  \bauthor{\bsnm{Davila},~\bfnm{T.}\binits{T.}},
  \bauthor{\bsnm{Desai},~\bfnm{M.}\binits{M.}},
  \bauthor{\bsnm{Bryson},~\bfnm{S.}\binits{S.}},
  \bauthor{\bsnm{Haskell},~\bfnm{W.~L.}\binits{W.~L.}},
  \bauthor{\bsnm{Matheson},~\bfnm{D.}\binits{D.}} \AND
  \bauthor{\bsnm{Robinson},~\bfnm{T.~N.}\binits{T.~N.}}
(\byear{2016}).
\btitle{{{E}ffects of {V}arying {E}poch {L}engths, {W}ear {T}ime {A}lgorithms,
  and {A}ctivity {C}ut-{P}oints on {E}stimates of {C}hild {S}edentary
  {B}ehavior and {P}hysical {A}ctivity from {A}ccelerometer {D}ata}}.
\bjournal{PLoS One}
\bvolume{11}
\bpages{e0150534}.
\end{barticle}
\endbibitem

\bibitem[\protect\citeauthoryear{Burton et~al.}{2013}]{burton2013activity}
\begin{barticle}[author]
\bauthor{\bsnm{Burton},~\bfnm{Christopher}\binits{C.}},
  \bauthor{\bsnm{McKinstry},~\bfnm{Brian}\binits{B.}},
  \bauthor{\bsnm{T{\u{a}}tar},~\bfnm{Aurora~Szentagotai}\binits{A.~S.}},
  \bauthor{\bsnm{Serrano-Blanco},~\bfnm{Antoni}\binits{A.}},
  \bauthor{\bsnm{Pagliari},~\bfnm{Claudia}\binits{C.}} \AND
  \bauthor{\bsnm{Wolters},~\bfnm{Maria}\binits{M.}}
(\byear{2013}).
\btitle{Activity monitoring in patients with depression: a systematic review}.
\bjournal{Journal of Affective Disorders}
\bvolume{145}
\bpages{21--28}.
\end{barticle}
\endbibitem

\bibitem[\protect\citeauthoryear{Chapman et~al.}{2017}]{chapmanSP2017}
\begin{barticle}[author]
\bauthor{\bsnm{Chapman},~\bfnm{J.~J.}\binits{J.~J.}},
  \bauthor{\bsnm{Roberts},~\bfnm{J.~A.}\binits{J.~A.}},
  \bauthor{\bsnm{Nguyen},~\bfnm{V.~T.}\binits{V.~T.}} \AND
  \bauthor{\bsnm{Breakspear},~\bfnm{M.}\binits{M.}}
(\byear{2017}).
\btitle{{{Q}uantification of free-living activity patterns using accelerometry
  in adults with mental illness}}.
\bjournal{Sci Rep}
\bvolume{7}
\bpages{43174}.
\end{barticle}
\endbibitem

\bibitem[\protect\citeauthoryear{Chen, Chen and Su}{2018}]{chen2018weighted}
\begin{barticle}[author]
\bauthor{\bsnm{Chen},~\bfnm{Hao}\binits{H.}},
  \bauthor{\bsnm{Chen},~\bfnm{Xu}\binits{X.}} \AND
  \bauthor{\bsnm{Su},~\bfnm{Yi}\binits{Y.}}
(\byear{2018}).
\btitle{A weighted edge-count two-sample test for multivariate and object
  data}.
\bjournal{Journal of the American Statistical Association}
\bvolume{113}
\bpages{1146--1155}.
\end{barticle}
\endbibitem

\bibitem[\protect\citeauthoryear{Chen and Friedman}{2017}]{chen2017new}
\begin{barticle}[author]
\bauthor{\bsnm{Chen},~\bfnm{Hao}\binits{H.}} \AND
  \bauthor{\bsnm{Friedman},~\bfnm{Jerome~H}\binits{J.~H.}}
(\byear{2017}).
\btitle{A new graph-based two-sample test for multivariate and object data}.
\bjournal{Journal of the American Statistical Association}
\bvolume{112}
\bpages{397--409}.
\end{barticle}
\endbibitem

\bibitem[\protect\citeauthoryear{Chen and Shao}{2005}]{chen2005stein}
\begin{barticle}[author]
\bauthor{\bsnm{Chen},~\bfnm{Louis~HY}\binits{L.~H.}} \AND
  \bauthor{\bsnm{Shao},~\bfnm{Qi-Man}\binits{Q.-M.}}
(\byear{2005}).
\btitle{Stein’s method for normal approximation}.
\bjournal{An introduction to Stein’s method}
\bvolume{4}
\bpages{1--59}.
\end{barticle}
\endbibitem

\bibitem[\protect\citeauthoryear{Dawson and Lagakos}{1993}]{dawson1993size}
\begin{barticle}[author]
\bauthor{\bsnm{Dawson},~\bfnm{Jeffrey~D}\binits{J.~D.}} \AND
  \bauthor{\bsnm{Lagakos},~\bfnm{Stephen~W}\binits{S.~W.}}
(\byear{1993}).
\btitle{Size and power of two-sample tests of repeated measures data}.
\bjournal{Biometrics}
\bpages{1022--1032}.
\end{barticle}
\endbibitem

\bibitem[\protect\citeauthoryear{De~Crescenzo et~al.}{2017}]{de2017actigraphic}
\begin{barticle}[author]
\bauthor{\bsnm{De~Crescenzo},~\bfnm{Franco}\binits{F.}},
  \bauthor{\bsnm{Economou},~\bfnm{Alexis}\binits{A.}},
  \bauthor{\bsnm{Sharpley},~\bfnm{Ann~L}\binits{A.~L.}},
  \bauthor{\bsnm{Gormez},~\bfnm{Aynur}\binits{A.}} \AND
  \bauthor{\bsnm{Quested},~\bfnm{Digby~J}\binits{D.~J.}}
(\byear{2017}).
\btitle{Actigraphic features of bipolar disorder: a systematic review and
  meta-analysis}.
\bjournal{Sleep Medicine Reviews}
\bvolume{33}
\bpages{58--69}.
\end{barticle}
\endbibitem

\bibitem[\protect\citeauthoryear{Doherty et~al.}{2017}]{dohertyplos2017}
\begin{barticle}[author]
\bauthor{\bsnm{Doherty},~\bfnm{A.}\binits{A.}},
  \bauthor{\bsnm{Jackson},~\bfnm{D.}\binits{D.}},
  \bauthor{\bsnm{Hammerla},~\bfnm{N.}\binits{N.}},
  \bauthor{\bsnm{Plötz},~\bfnm{T.}\binits{T.}},
  \bauthor{\bsnm{Olivier},~\bfnm{P.}\binits{P.}},
  \bauthor{\bsnm{Granat},~\bfnm{M.~H.}\binits{M.~H.}},
  \bauthor{\bsnm{White},~\bfnm{T.}\binits{T.}}, \bauthor{\bparticle{van}
  \bsnm{Hees},~\bfnm{V.~T.}\binits{V.~T.}},
  \bauthor{\bsnm{Trenell},~\bfnm{M.~I.}\binits{M.~I.}},
  \bauthor{\bsnm{Owen},~\bfnm{C.~G.}\binits{C.~G.}},
  \bauthor{\bsnm{Preece},~\bfnm{S.~J.}\binits{S.~J.}},
  \bauthor{\bsnm{Gillions},~\bfnm{R.}\binits{R.}},
  \bauthor{\bsnm{Sheard},~\bfnm{S.}\binits{S.}},
  \bauthor{\bsnm{Peakman},~\bfnm{T.}\binits{T.}},
  \bauthor{\bsnm{Brage},~\bfnm{S.}\binits{S.}} \AND
  \bauthor{\bsnm{Wareham},~\bfnm{N.~J.}\binits{N.~J.}}
(\byear{2017}).
\btitle{{{L}arge {S}cale {P}opulation {A}ssessment of {P}hysical {A}ctivity
  {U}sing {W}rist {W}orn {A}ccelerometers: {T}he {U}{K} {B}iobank {S}tudy}}.
\bjournal{PLoS One}
\bvolume{12}
\bpages{e0169649}.
\end{barticle}
\endbibitem

\bibitem[\protect\citeauthoryear{Dubey and M{\"u}ller}{2019}]{dubey2019frechet}
\begin{barticle}[author]
\bauthor{\bsnm{Dubey},~\bfnm{Paromita}\binits{P.}} \AND
  \bauthor{\bsnm{M{\"u}ller},~\bfnm{Hans-Georg}\binits{H.-G.}}
(\byear{2019}).
\btitle{Fr{\'e}chet analysis of variance for random objects}.
\bjournal{Biometrika}
\bvolume{106}
\bpages{803--821}.
\end{barticle}
\endbibitem

\bibitem[\protect\citeauthoryear{Faurholt-Jepsen
  et~al.}{2012}]{faurholt2012differences}
\begin{barticle}[author]
\bauthor{\bsnm{Faurholt-Jepsen},~\bfnm{Maria}\binits{M.}},
  \bauthor{\bsnm{Brage},~\bfnm{S{\o}ren}\binits{S.}},
  \bauthor{\bsnm{Vinberg},~\bfnm{Maj}\binits{M.}},
  \bauthor{\bsnm{Christensen},~\bfnm{Ellen~Margrethe}\binits{E.~M.}},
  \bauthor{\bsnm{Knorr},~\bfnm{Ulla}\binits{U.}},
  \bauthor{\bsnm{Jensen},~\bfnm{Hans~M{\o}rch}\binits{H.~M.}} \AND
  \bauthor{\bsnm{Kessing},~\bfnm{Lars~Vedel}\binits{L.~V.}}
(\byear{2012}).
\btitle{Differences in psychomotor activity in patients suffering from unipolar
  and bipolar affective disorder in the remitted or mild/moderate depressive
  state}.
\bjournal{Journal of Affective Disorders}
\bvolume{141}
\bpages{457--463}.
\end{barticle}
\endbibitem

\bibitem[\protect\citeauthoryear{Friedman and Rafsky}{1979}]{friedman1979}
\begin{barticle}[author]
\bauthor{\bsnm{Friedman},~\bfnm{Jerome~H}\binits{J.~H.}} \AND
  \bauthor{\bsnm{Rafsky},~\bfnm{Lawrence~C}\binits{L.~C.}}
(\byear{1979}).
\btitle{Multivariate generalizations of the Wald-Wolfowitz and Smirnov
  two-sample tests}.
\bjournal{The Annals of Statistics}
\bpages{697--717}.
\end{barticle}
\endbibitem

\bibitem[\protect\citeauthoryear{Gretton et~al.}{2012}]{gretton2012kernel}
\begin{barticle}[author]
\bauthor{\bsnm{Gretton},~\bfnm{Arthur}\binits{A.}},
  \bauthor{\bsnm{Borgwardt},~\bfnm{Karsten~M}\binits{K.~M.}},
  \bauthor{\bsnm{Rasch},~\bfnm{Malte~J}\binits{M.~J.}},
  \bauthor{\bsnm{Sch{\"o}lkopf},~\bfnm{Bernhard}\binits{B.}} \AND
  \bauthor{\bsnm{Smola},~\bfnm{Alexander}\binits{A.}}
(\byear{2012}).
\btitle{A kernel two-sample test}.
\bjournal{The Journal of Machine Learning Research}
\bvolume{13}
\bpages{723--773}.
\end{barticle}
\endbibitem

\bibitem[\protect\citeauthoryear{Henze}{1988}]{henze1988}
\begin{barticle}[author]
\bauthor{\bsnm{Henze},~\bfnm{Norbert}\binits{N.}}
(\byear{1988}).
\btitle{A multivariate two-sample test based on the number of nearest neighbor
  type coincidences}.
\bjournal{The Annals of Statistics}
\bpages{772--783}.
\end{barticle}
\endbibitem

\bibitem[\protect\citeauthoryear{Indic et~al.}{2011}]{indic2011}
\begin{barticle}[author]
\bauthor{\bsnm{Indic},~\bfnm{P.}\binits{P.}},
  \bauthor{\bsnm{Salvatore},~\bfnm{P.}\binits{P.}},
  \bauthor{\bsnm{Maggini},~\bfnm{C.}\binits{C.}},
  \bauthor{\bsnm{Ghidini},~\bfnm{S.}\binits{S.}},
  \bauthor{\bsnm{Ferraro},~\bfnm{G.}\binits{G.}},
  \bauthor{\bsnm{Baldessarini},~\bfnm{R.~J.}\binits{R.~J.}} \AND
  \bauthor{\bsnm{Murray},~\bfnm{G.}\binits{G.}}
(\byear{2011}).
\btitle{{{S}caling behavior of human locomotor activity amplitude: association
  with bipolar disorder}}.
\bjournal{PLoS One}
\bvolume{6}
\bpages{e20650}.
\end{barticle}
\endbibitem

\bibitem[\protect\citeauthoryear{Keadle et~al.}{2014}]{keadle_impact_2014}
\begin{barticle}[author]
\bauthor{\bsnm{Keadle},~\bfnm{Sarah~Kozey}\binits{S.~K.}},
  \bauthor{\bsnm{Shiroma},~\bfnm{Eric~J.}\binits{E.~J.}},
  \bauthor{\bsnm{Freedson},~\bfnm{Patty~S.}\binits{P.~S.}} \AND
  \bauthor{\bsnm{Lee},~\bfnm{I.~Min}\binits{I.~M.}}
(\byear{2014}).
\btitle{Impact of accelerometer data processing decisions on the sample size,
  wear time and physical activity level of a large cohort study}.
\bjournal{BMC Public Health}
\bvolume{14}
\bpages{1210}.
\bdoi{10.1186/1471-2458-14-1210}
\end{barticle}
\endbibitem

\bibitem[\protect\citeauthoryear{Krane-Gartiser
  et~al.}{2014}]{krane2014actigraphic}
\begin{barticle}[author]
\bauthor{\bsnm{Krane-Gartiser},~\bfnm{Karoline}\binits{K.}},
  \bauthor{\bsnm{Henriksen},~\bfnm{Tone Elise~Gjotterud}\binits{T.~E.~G.}},
  \bauthor{\bsnm{Morken},~\bfnm{Gunnar}\binits{G.}},
  \bauthor{\bsnm{Vaaler},~\bfnm{Arne}\binits{A.}} \AND
  \bauthor{\bsnm{Fasmer},~\bfnm{Ole~Bernt}\binits{O.~B.}}
(\byear{2014}).
\btitle{Actigraphic assessment of motor activity in acutely admitted inpatients
  with bipolar disorder}.
\bjournal{PloS One}
\bvolume{9}
\bpages{e89574}.
\end{barticle}
\endbibitem

\bibitem[\protect\citeauthoryear{Leeger-Aschmann et~al.}{2019}]{leeger2019}
\begin{barticle}[author]
\bauthor{\bsnm{Leeger-Aschmann},~\bfnm{C.~S.}\binits{C.~S.}},
  \bauthor{\bsnm{Schmutz},~\bfnm{E.~A.}\binits{E.~A.}},
  \bauthor{\bsnm{Zysset},~\bfnm{A.~E.}\binits{A.~E.}},
  \bauthor{\bsnm{Kakebeeke},~\bfnm{T.~H.}\binits{T.~H.}},
  \bauthor{\bsnm{Messerli-Bürgy},~\bfnm{N.}\binits{N.}},
  \bauthor{\bsnm{Stülb},~\bfnm{K.}\binits{K.}},
  \bauthor{\bsnm{Arhab},~\bfnm{A.}\binits{A.}},
  \bauthor{\bsnm{Meyer},~\bfnm{A.~H.}\binits{A.~H.}},
  \bauthor{\bsnm{Munsch},~\bfnm{S.}\binits{S.}},
  \bauthor{\bsnm{Jenni},~\bfnm{O.~G.}\binits{O.~G.}},
  \bauthor{\bsnm{Puder},~\bfnm{J.~J.}\binits{J.~J.}} \AND
  \bauthor{\bsnm{Kriemler},~\bfnm{S.}\binits{S.}}
(\byear{2019}).
\btitle{{{A}ccelerometer-derived physical activity estimation in preschoolers -
  comparison of cut-point sets incorporating the vector magnitude vs the
  vertical axis}}.
\bjournal{BMC Public Health}
\bvolume{19}
\bpages{513}.
\end{barticle}
\endbibitem

\bibitem[\protect\citeauthoryear{Merikangas
  et~al.}{2014}]{merikangas_independence_2014}
\begin{barticle}[author]
\bauthor{\bsnm{Merikangas},~\bfnm{K.~R.}\binits{K.~R.}},
  \bauthor{\bsnm{Cui},~\bfnm{L.}\binits{L.}},
  \bauthor{\bsnm{Heaton},~\bfnm{L.}\binits{L.}},
  \bauthor{\bsnm{Nakamura},~\bfnm{E.}\binits{E.}},
  \bauthor{\bsnm{Roca},~\bfnm{C.}\binits{C.}},
  \bauthor{\bsnm{Ding},~\bfnm{J.}\binits{J.}},
  \bauthor{\bsnm{Qin},~\bfnm{H.}\binits{H.}},
  \bauthor{\bsnm{Guo},~\bfnm{W.}\binits{W.}},
  \bauthor{\bsnm{Shugart},~\bfnm{Y.~Y.}\binits{Y.~Y.}},
  \bauthor{\bsnm{Zarate},~\bfnm{C.}\binits{C.}} \AND
  \bauthor{\bsnm{Angst},~\bfnm{J.}\binits{J.}}
(\byear{2014}).
\btitle{Independence of familial transmission of mania and depression: results
  of the {NIMH} family study of affective spectrum disorders}.
\bjournal{Mol Psychiatry}
\bvolume{19}
\bpages{214--9}.
\bdoi{10.1038/mp.2013.116}
\end{barticle}
\endbibitem

\bibitem[\protect\citeauthoryear{Merikangas et~al.}{2019}]{merikangas2019jama}
\begin{barticle}[author]
\bauthor{\bsnm{Merikangas},~\bfnm{K.~R.}\binits{K.~R.}},
  \bauthor{\bsnm{Swendsen},~\bfnm{J.}\binits{J.}},
  \bauthor{\bsnm{Hickie},~\bfnm{I.~B.}\binits{I.~B.}},
  \bauthor{\bsnm{Cui},~\bfnm{L.}\binits{L.}},
  \bauthor{\bsnm{Shou},~\bfnm{H.}\binits{H.}},
  \bauthor{\bsnm{Merikangas},~\bfnm{A.~K.}\binits{A.~K.}},
  \bauthor{\bsnm{Zhang},~\bfnm{J.}\binits{J.}},
  \bauthor{\bsnm{Lamers},~\bfnm{F.}\binits{F.}},
  \bauthor{\bsnm{Crainiceanu},~\bfnm{C.}\binits{C.}},
  \bauthor{\bsnm{Volkow},~\bfnm{N.~D.}\binits{N.~D.}} \AND
  \bauthor{\bsnm{Zipunnikov},~\bfnm{V.}\binits{V.}}
(\byear{2019}).
\btitle{{{R}eal-time {m}obile {m}onitoring of the {d}ynamic {a}ssociations
  {a}mong {m}otor {a}ctivity, {e}nergy, {m}ood, and {s}leep in {a}dults {w}ith
  {b}ipolar {d}isorder}}.
\bjournal{JAMA Psychiatry}
\bvolume{76}
\bpages{190--198}.
\end{barticle}
\endbibitem

\bibitem[\protect\citeauthoryear{Murray et~al.}{2020}]{murray_measuring_2020}
\begin{barticle}[author]
\bauthor{\bsnm{Murray},~\bfnm{Greg}\binits{G.}},
  \bauthor{\bsnm{Gottlieb},~\bfnm{John}\binits{J.}},
  \bauthor{\bsnm{Hidalgo},~\bfnm{Maria~Paz}\binits{M.~P.}},
  \bauthor{\bsnm{Etain},~\bfnm{Bruno}\binits{B.}},
  \bauthor{\bsnm{Ritter},~\bfnm{Philipp}\binits{P.}},
  \bauthor{\bsnm{Skene},~\bfnm{Debra~J.}\binits{D.~J.}},
  \bauthor{\bsnm{Garbazza},~\bfnm{Corrado}\binits{C.}},
  \bauthor{\bsnm{Bullock},~\bfnm{Ben}\binits{B.}},
  \bauthor{\bsnm{Merikangas},~\bfnm{Kathleen}\binits{K.}},
  \bauthor{\bsnm{Zipunnikov},~\bfnm{Vadim}\binits{V.}},
  \bauthor{\bsnm{Shou},~\bfnm{Haochang}\binits{H.}},
  \bauthor{\bsnm{Gonzalez},~\bfnm{Robert}\binits{R.}},
  \bauthor{\bsnm{Scott},~\bfnm{Jan}\binits{J.}},
  \bauthor{\bsnm{Geoffroy},~\bfnm{Pierre~A.}\binits{P.~A.}} \AND
  \bauthor{\bsnm{Frey},~\bfnm{Benicio~N.}\binits{B.~N.}}
(\byear{2020}).
\btitle{Measuring circadian function in bipolar disorders: {Empirical} and
  conceptual review of physiological, actigraphic, and self-report approaches}.
\bjournal{Bipolar Disorders}.
\bdoi{10.1111/bdi.12963}
\end{barticle}
\endbibitem

\bibitem[\protect\citeauthoryear{Pagani et~al.}{2016}]{pagani2016}
\begin{barticle}[author]
\bauthor{\bsnm{Pagani},~\bfnm{L.}\binits{L.}},
  \bauthor{\bsnm{St~Clair},~\bfnm{P.~A.}\binits{P.~A.}},
  \bauthor{\bsnm{Teshiba},~\bfnm{T.~M.}\binits{T.~M.}},
  \bauthor{\bsnm{Service},~\bfnm{S.~K.}\binits{S.~K.}},
  \bauthor{\bsnm{Fears},~\bfnm{S.~C.}\binits{S.~C.}},
  \bauthor{\bsnm{Araya},~\bfnm{C.}\binits{C.}},
  \bauthor{\bsnm{Araya},~\bfnm{X.}\binits{X.}},
  \bauthor{\bsnm{Bejarano},~\bfnm{J.}\binits{J.}},
  \bauthor{\bsnm{Ramirez},~\bfnm{M.}\binits{M.}},
  \bauthor{\bsnm{Castrillón},~\bfnm{G.}\binits{G.}},
  \bauthor{\bsnm{Gomez-Makhinson},~\bfnm{J.}\binits{J.}},
  \bauthor{\bsnm{Lopez},~\bfnm{M.~C.}\binits{M.~C.}},
  \bauthor{\bsnm{Montoya},~\bfnm{G.}\binits{G.}},
  \bauthor{\bsnm{Montoya},~\bfnm{C.~P.}\binits{C.~P.}},
  \bauthor{\bsnm{Aldana},~\bfnm{I.}\binits{I.}},
  \bauthor{\bsnm{Navarro},~\bfnm{L.}\binits{L.}},
  \bauthor{\bsnm{Freimer},~\bfnm{D.~G.}\binits{D.~G.}},
  \bauthor{\bsnm{Safaie},~\bfnm{B.}\binits{B.}},
  \bauthor{\bsnm{Keung},~\bfnm{L.~W.}\binits{L.~W.}},
  \bauthor{\bsnm{Greenspan},~\bfnm{K.}\binits{K.}},
  \bauthor{\bsnm{Chou},~\bfnm{K.}\binits{K.}},
  \bauthor{\bsnm{Escobar},~\bfnm{J.~I.}\binits{J.~I.}},
  \bauthor{\bsnm{Ospina-Duque},~\bfnm{J.}\binits{J.}},
  \bauthor{\bsnm{Kremeyer},~\bfnm{B.}\binits{B.}},
  \bauthor{\bsnm{Ruiz-Linares},~\bfnm{A.}\binits{A.}},
  \bauthor{\bsnm{Cantor},~\bfnm{R.~M.}\binits{R.~M.}},
  \bauthor{\bsnm{Lopez-Jaramillo},~\bfnm{C.}\binits{C.}},
  \bauthor{\bsnm{Macaya},~\bfnm{G.}\binits{G.}},
  \bauthor{\bsnm{Molina},~\bfnm{J.}\binits{J.}},
  \bauthor{\bsnm{Reus},~\bfnm{V.~I.}\binits{V.~I.}},
  \bauthor{\bsnm{Sabatti},~\bfnm{C.}\binits{C.}},
  \bauthor{\bsnm{Bearden},~\bfnm{C.~E.}\binits{C.~E.}},
  \bauthor{\bsnm{Takahashi},~\bfnm{J.~S.}\binits{J.~S.}} \AND
  \bauthor{\bsnm{Freimer},~\bfnm{N.~B.}\binits{N.~B.}}
(\byear{2016}).
\btitle{{{G}enetic contributions to circadian activity rhythm and sleep pattern
  phenotypes in pedigrees segregating for severe bipolar disorder}}.
\bjournal{Proc Natl Acad Sci U S A}
\bvolume{113}
\bpages{E754--761}.
\end{barticle}
\endbibitem

\bibitem[\protect\citeauthoryear{Petersen and
  M{\"u}ller}{2019}]{petersen2019wasserstein}
\begin{barticle}[author]
\bauthor{\bsnm{Petersen},~\bfnm{Alexander}\binits{A.}} \AND
  \bauthor{\bsnm{M{\"u}ller},~\bfnm{Hans-Georg}\binits{H.-G.}}
(\byear{2019}).
\btitle{Wasserstein covariance for multiple random densities}.
\bjournal{Biometrika}
\bvolume{106}
\bpages{339--351}.
\end{barticle}
\endbibitem

\bibitem[\protect\citeauthoryear{Robillard et~al.}{2015}]{RobillardJPN2015}
\begin{barticle}[author]
\bauthor{\bsnm{Robillard},~\bfnm{R.}\binits{R.}},
  \bauthor{\bsnm{Hermens},~\bfnm{D.~F.}\binits{D.~F.}},
  \bauthor{\bsnm{Naismith},~\bfnm{S.~L.}\binits{S.~L.}},
  \bauthor{\bsnm{White},~\bfnm{D.}\binits{D.}},
  \bauthor{\bsnm{Rogers},~\bfnm{N.~L.}\binits{N.~L.}},
  \bauthor{\bsnm{Ip},~\bfnm{T.~K.}\binits{T.~K.}},
  \bauthor{\bsnm{Mullin},~\bfnm{S.~J.}\binits{S.~J.}},
  \bauthor{\bsnm{Alvares},~\bfnm{G.~A.}\binits{G.~A.}},
  \bauthor{\bsnm{Guastella},~\bfnm{A.~J.}\binits{A.~J.}},
  \bauthor{\bsnm{Smith},~\bfnm{K.~L.}\binits{K.~L.}},
  \bauthor{\bsnm{Rong},~\bfnm{Y.}\binits{Y.}},
  \bauthor{\bsnm{Whitwell},~\bfnm{B.}\binits{B.}},
  \bauthor{\bsnm{Southan},~\bfnm{J.}\binits{J.}},
  \bauthor{\bsnm{Glozier},~\bfnm{N.}\binits{N.}},
  \bauthor{\bsnm{Scott},~\bfnm{E.~M.}\binits{E.~M.}} \AND
  \bauthor{\bsnm{Hickie},~\bfnm{I.~B.}\binits{I.~B.}}
(\byear{2015}).
\btitle{{{A}mbulatory sleep-wake patterns and variability in young people with
  emerging mental disorders}}.
\bjournal{J Psychiatry Neurosci}
\bvolume{40}
\bpages{28--37}.
\end{barticle}
\endbibitem

\bibitem[\protect\citeauthoryear{Rosenbaum}{2005}]{rosenbaum2005}
\begin{barticle}[author]
\bauthor{\bsnm{Rosenbaum},~\bfnm{Paul~R}\binits{P.~R.}}
(\byear{2005}).
\btitle{An exact distribution-free test comparing two multivariate
  distributions based on adjacency}.
\bjournal{Journal of the Royal Statistical Society: Series B (Statistical
  Methodology)}
\bvolume{67}
\bpages{515--530}.
\end{barticle}
\endbibitem

\bibitem[\protect\citeauthoryear{Schilling}{1986}]{schilling1986}
\begin{barticle}[author]
\bauthor{\bsnm{Schilling},~\bfnm{Mark~F}\binits{M.~F.}}
(\byear{1986}).
\btitle{Multivariate two-sample tests based on nearest neighbors}.
\bjournal{Journal of the American Statistical Association}
\bvolume{81}
\bpages{799--806}.
\end{barticle}
\endbibitem

\bibitem[\protect\citeauthoryear{Schott}{2007}]{schott2007}
\begin{barticle}[author]
\bauthor{\bsnm{Schott},~\bfnm{James~R}\binits{J.~R.}}
(\byear{2007}).
\btitle{A test for the equality of covariance matrices when the dimension is
  large relative to the sample sizes}.
\bjournal{Computational Statistics \& Data Analysis}
\bvolume{51}
\bpages{6535--6542}.
\end{barticle}
\endbibitem

\bibitem[\protect\citeauthoryear{Schrack et~al.}{2014}]{schrack_assessing_2014}
\begin{barticle}[author]
\bauthor{\bsnm{Schrack},~\bfnm{J.~A.}\binits{J.~A.}},
  \bauthor{\bsnm{Zipunnikov},~\bfnm{V.}\binits{V.}},
  \bauthor{\bsnm{Goldsmith},~\bfnm{J.}\binits{J.}},
  \bauthor{\bsnm{Bai},~\bfnm{J.}\binits{J.}},
  \bauthor{\bsnm{Simonsick},~\bfnm{E.~M.}\binits{E.~M.}},
  \bauthor{\bsnm{Crainiceanu},~\bfnm{C.}\binits{C.}} \AND
  \bauthor{\bsnm{Ferrucci},~\bfnm{L.}\binits{L.}}
(\byear{2014}).
\btitle{Assessing the "physical cliff": detailed quantification of age-related
  differences in daily patterns of physical activity}.
\bjournal{J Gerontol A Biol Sci Med Sci}
\bvolume{69}
\bpages{973--9}.
\bdoi{10.1093/gerona/glt199}
\end{barticle}
\endbibitem

\bibitem[\protect\citeauthoryear{Schrack et~al.}{2016}]{schrack_assessing_2016}
\begin{barticle}[author]
\bauthor{\bsnm{Schrack},~\bfnm{Jennifer~A.}\binits{J.~A.}},
  \bauthor{\bsnm{Cooper},~\bfnm{Rachel}\binits{R.}},
  \bauthor{\bsnm{Koster},~\bfnm{Annemarie}\binits{A.}},
  \bauthor{\bsnm{Shiroma},~\bfnm{Eric~J.}\binits{E.~J.}},
  \bauthor{\bsnm{Murabito},~\bfnm{Joanne~M.}\binits{J.~M.}},
  \bauthor{\bsnm{Rejeski},~\bfnm{W.~Jack}\binits{W.~J.}},
  \bauthor{\bsnm{Ferrucci},~\bfnm{Luigi}\binits{L.}} \AND
  \bauthor{\bsnm{Harris},~\bfnm{Tamara~B.}\binits{T.~B.}}
(\byear{2016}).
\btitle{Assessing {Daily} {Physical} {Activity} in {Older} {Adults}:
  {Unraveling} the {Complexity} of {Monitors}, {Measures}, and {Methods}}.
\bjournal{The Journals of Gerontology Series A: Biological Sciences and Medical
  Sciences}
\bvolume{71}
\bpages{1039--1048}.
\bdoi{10.1093/gerona/glw026}
\end{barticle}
\endbibitem

\bibitem[\protect\citeauthoryear{Scott et~al.}{2017}]{scottJAMA2017}
\begin{barticle}[author]
\bauthor{\bsnm{Scott},~\bfnm{J.}\binits{J.}},
  \bauthor{\bsnm{Murray},~\bfnm{G.}\binits{G.}},
  \bauthor{\bsnm{Henry},~\bfnm{C.}\binits{C.}},
  \bauthor{\bsnm{Morken},~\bfnm{G.}\binits{G.}},
  \bauthor{\bsnm{Scott},~\bfnm{E.}\binits{E.}},
  \bauthor{\bsnm{Angst},~\bfnm{J.}\binits{J.}},
  \bauthor{\bsnm{Merikangas},~\bfnm{K.~R.}\binits{K.~R.}} \AND
  \bauthor{\bsnm{Hickie},~\bfnm{I.~B.}\binits{I.~B.}}
(\byear{2017}).
\btitle{{{A}ctivation in {B}ipolar {D}isorders: {A} {S}ystematic {R}eview}}.
\bjournal{JAMA Psychiatry}
\bvolume{74}
\bpages{189--196}.
\end{barticle}
\endbibitem

\bibitem[\protect\citeauthoryear{Shou et~al.}{2017}]{shou_dysregulation_2017}
\begin{barticle}[author]
\bauthor{\bsnm{Shou},~\bfnm{H.}\binits{H.}},
  \bauthor{\bsnm{Cui},~\bfnm{L.}\binits{L.}},
  \bauthor{\bsnm{Hickie},~\bfnm{I.}\binits{I.}},
  \bauthor{\bsnm{Lameira},~\bfnm{D.}\binits{D.}},
  \bauthor{\bsnm{Lamers},~\bfnm{F.}\binits{F.}},
  \bauthor{\bsnm{Zhang},~\bfnm{J.}\binits{J.}},
  \bauthor{\bsnm{Crainiceanu},~\bfnm{C.}\binits{C.}},
  \bauthor{\bsnm{Zipunnikov},~\bfnm{V.}\binits{V.}} \AND
  \bauthor{\bsnm{Merikangas},~\bfnm{K.~R.}\binits{K.~R.}}
(\byear{2017}).
\btitle{Dysregulation of objectively assessed 24-hour motor activity patterns
  as a potential marker for bipolar {I} disorder: results of a community-based
  family study}.
\bjournal{Translational Psychiatry}
\bvolume{7}
\bpages{e1211}.
\bdoi{10.1038/tp.2017.136}
\end{barticle}
\endbibitem

\bibitem[\protect\citeauthoryear{Varma et~al.}{2017}]{varma2017}
\begin{barticle}[author]
\bauthor{\bsnm{Varma},~\bfnm{V.~R.}\binits{V.~R.}},
  \bauthor{\bsnm{Dey},~\bfnm{D.}\binits{D.}},
  \bauthor{\bsnm{Leroux},~\bfnm{A.}\binits{A.}},
  \bauthor{\bsnm{Di},~\bfnm{J.}\binits{J.}},
  \bauthor{\bsnm{Urbanek},~\bfnm{J.}\binits{J.}},
  \bauthor{\bsnm{Xiao},~\bfnm{L.}\binits{L.}} \AND
  \bauthor{\bsnm{Zipunnikov},~\bfnm{V.}\binits{V.}}
(\byear{2017}).
\btitle{{{R}e-evaluating the effect of age on physical activity over the
  lifespan}}.
\bjournal{Prev Med}
\bvolume{101}
\bpages{102--108}.
\end{barticle}
\endbibitem

\bibitem[\protect\citeauthoryear{Varma et~al.}{2018}]{varma_total_2018}
\begin{barticle}[author]
\bauthor{\bsnm{Varma},~\bfnm{Vijay~R.}\binits{V.~R.}},
  \bauthor{\bsnm{Dey},~\bfnm{Debangan}\binits{D.}},
  \bauthor{\bsnm{Leroux},~\bfnm{Andrew}\binits{A.}},
  \bauthor{\bsnm{Di},~\bfnm{Junrui}\binits{J.}},
  \bauthor{\bsnm{Urbanek},~\bfnm{Jacek}\binits{J.}},
  \bauthor{\bsnm{Xiao},~\bfnm{Luo}\binits{L.}} \AND
  \bauthor{\bsnm{Zipunnikov},~\bfnm{Vadim}\binits{V.}}
(\byear{2018}).
\btitle{Total volume of physical activity: {TAC}, {TLAC} or {TAC}($\lambda$)}.
\bjournal{Preventive Medicine}
\bvolume{106}
\bpages{233--235}.
\bdoi{10.1016/j.ypmed.2017.10.028}
\end{barticle}
\endbibitem

\bibitem[\protect\citeauthoryear{Viciana, Mayorga-Vega and
  Martínez-Baena}{2016}]{Viciana2016}
\begin{barticle}[author]
\bauthor{\bsnm{Viciana},~\bfnm{J.}\binits{J.}},
  \bauthor{\bsnm{Mayorga-Vega},~\bfnm{D.}\binits{D.}} \AND
  \bauthor{\bsnm{Martínez-Baena},~\bfnm{A.}\binits{A.}}
(\byear{2016}).
\btitle{{{M}oderate-to-{V}igorous {P}hysical {A}ctivity {L}evels in {P}hysical
  {E}ducation, {S}chool {R}ecess, and {A}fter-{S}chool {T}ime: {I}nfluence of
  {G}ender, {A}ge, and {W}eight {S}tatus}}.
\bjournal{J Phys Act Health}
\bvolume{13}
\bpages{1117--1123}.
\end{barticle}
\endbibitem

\bibitem[\protect\citeauthoryear{Wrobel et~al.}{2019}]{wrobel_2019}
\begin{barticle}[author]
\bauthor{\bsnm{Wrobel},~\bfnm{J.}\binits{J.}},
  \bauthor{\bsnm{Zipunnikov},~\bfnm{V.}\binits{V.}},
  \bauthor{\bsnm{Schrack},~\bfnm{J.}\binits{J.}} \AND
  \bauthor{\bsnm{Goldsmith},~\bfnm{J.}\binits{J.}}
(\byear{2019}).
\btitle{{{R}egistration for exponential family functional data}}.
\bjournal{Biometrics}
\bvolume{75}
\bpages{48--57}.
\end{barticle}
\endbibitem

\bibitem[\protect\citeauthoryear{Yang et~al.}{2020}]{yang2020quantlet}
\begin{barticle}[author]
\bauthor{\bsnm{Yang},~\bfnm{H.}\binits{H.}},
  \bauthor{\bsnm{Baladandayuthapani},~\bfnm{V.}\binits{V.}},
  \bauthor{\bsnm{Rao},~\bfnm{A.~U.~K.}\binits{A.~U.~K.}} \AND
  \bauthor{\bsnm{Morris},~\bfnm{J.~S.}\binits{J.~S.}}
(\byear{2020}).
\btitle{{{Q}uantile {f}unction on {s}calar {r}egression {a}nalysis for
  {d}istributional {d}ata}}.
\bjournal{J Am Stat Assoc}
\bvolume{115}
\bpages{90--106}.
\end{barticle}
\endbibitem

\bibitem[\protect\citeauthoryear{Zhang and Chen}{2020}]{zhang2017graph}
\begin{barticle}[author]
\bauthor{\bsnm{Zhang},~\bfnm{Jingru}\binits{J.}} \AND
  \bauthor{\bsnm{Chen},~\bfnm{Hao}\binits{H.}}
(\byear{2020}).
\btitle{Graph-based two-sample tests for data with repeated observations}.
\bjournal{Statistica Sinica}.
\end{barticle}
\endbibitem

\end{thebibliography}

\newpage
\renewcommand{\appendixname}{Supplement}
\appendix 
\section{Proofs of theorems}
\subsection{Proof of Theorem \ref{th:expression}}\label{pf:expression}
\begin{proof}
According to the following definitions
\begin{align*}
&R_{\tout,1} = \sum_{(u,v)\in G_{\tout}}I(g_u=g_v=1), \\
&R_{\tout,2} = \sum_{(u,v)\in G_{\tout}}I(g_u=g_v=2), \\
&R_{\tin,1} = \sum_{(i,j)\in G_{\tin}}I(g_i=g_j=1),
\end{align*}
we have
\begin{align*}
& \bE(R_{\tout,1}) = \sum_{(u,v)\in G_{\tout}}\bP(g_u=g_v=1) = |G_\tout|\frac{n_1(n_1-1)}{N(N-1)}, \\
& \bE(R_{\tout,2}) = \sum_{(u,v)\in G_{\tout}}\bP(g_u=g_v=2) = |G_\tout|\frac{n_2(n_2-1)}{N(N-1)}, \\
& \bE(R_{\tin,1}) = \sum_{(i,j)\in G_{\tin}}\bP(g_i=g_j=1) = |G_\tin|\frac{n_1}{N}.
\end{align*}

For $\bvar(R_{\tout,1})$, we only need to figure out $\bE(R_{\tout,1}^2)$. We have 
\begin{align*}
& \bE(R_{\text{out},1}^2) \\
= & \sum_{(u,v)\in G_{\text{out}}}P(g_u=g_v=1) + \sum_{\substack{(u,v),(u,s)\in G_{\text{out}} \\ v\neq s}}P(g_u=g_v=g_s=1) \\
& +  \sum_{\substack{(u,v),(s,t)\in G_{\text{out}} \\ u,v,s,t \text{ are all different}}}P(g_u=g_v=g_s=g_t=1) \\
= & |G_{\text{out}}|\frac{n_1(n_1-1)}{N(N-1)} + \sum_{u\neq v}\frac{D_{uv}(D_{uv}-1)}{2}\frac{n_1(n_1-1)}{N(N-1)} \\
& + \sum_{u\neq v}D_{uv}(D_u-D_{uv})\frac{n_1(n_1-1)(n_1-2)}{N(N-1)(N-2)} \\
& + \left(|G_{\text{out}}|^2+\sum_{u\neq v}\frac{D_{uv}^2}{2}-\sum_uD_u^2\right)\frac{n_1(n_1-1)(n_1-2)(n_1-3)}{N(N-1)(N-2)(N-3)} \\
 =& \sum_{u\neq v}\frac{D_{uv}^2}{2}\frac{n_1(n_1-1)}{N(N-1)} + \sum_{u\neq v}D_{uv}(D_u-D_{uv})\frac{n_1(n_1-1)(n_1-2)}{N(N-1)(N-2)} \\
 & + \left(|G_{\text{out}}|^2+\sum_{u\neq v}\frac{D_{uv}^2}{2}-\sum_uD_u^2\right)\frac{n_1(n_1-1)(n_1-2)(n_1-3)}{N(N-1)(N-2)(N-3)}. 
\end{align*}
The analytic expression of $\bvar(R_{\tout,2})$ can be derived in a similar way as that of $\bvar(R_{\tout,1})$.

For $\bcov(R_{\tout,1},R_{\tout,2})$, we only need to figure out $\bE(R_{\tout,1},R_{\tout,2})$. We have 
\begin{align*}
\bE(R_{\text{out},1}R_{\text{out},2})=&\sum_{(u,v)(s,t)\in G_{\text{out}}}P(g_u=g_v=1,g_s=g_t=2) \\
= & \left(|G_{\text{out}}|^2+\sum_{u\neq v}\frac{D_{uv}^2}{2}-\sum_uD_u^2\right)\frac{n_1(n_1-1)n_2(n_2-1)}{N(N-1)(N-2)(N-3)}. 
\end{align*}

Then we compute $\bE(R_{\tin,1}^2)$, $\bE(R_{\tout,1},R_{\tin,1})$ and $\bE(R_{\tout,2},R_{\tin,1})$ so that the analytic expressions of $\bvar(R_{\tin,1})$, $\bcov(R_{\tout,1},R_{\tin,1})$ and $\bcov(R_{\tout,2},R_{\tin,1})$ can be derived immediately. We have
\begin{align*}
&\bE(R_{\text{in},1}^2) \\
=& \sum_{(i,j)\in G_{\text{in}}}P(g_i=g_j=1) + \sum_{\substack{(i,j),(i,k)\in G_{\text{in}} \\ j\neq k}}P(g_i=g_j=g_k=1) \\
& +  \sum_{\substack{(i,j),(k,l)\in G_{\text{in}} \\ i,j,k,l \text{ are all different}}}P(g_i=g_j=g_k=g_l=1) \\
=& \sum_{u=1}^ND_{uu}^2\frac{n_1n_2}{N(N-1)} + |G_{\text{in}}|^2\frac{n_1(n_1-1)}{N(N-1)},\\
& \bE(R_{\text{in},1}R_{\text{out},1}) \\
=& \sum_{\substack{(i,j)\in G_{\text{in}},i\in \C_s \\(s,t)\in G_{\text{out}}}}P(g_i=g_j=g_t=1) + \sum_{\substack{(i,j)\in G_{\text{in}},i\notin\C_s,j\notin\C_t \\(s,t)\in G_{\text{out}}}}P(g_i=g_j=g_s=g_t=1) \\
=& \sum_{u=1}^ND_{uu}D_u\frac{n_1(n_1-1)}{N(N-1)} + \left(|G_{\text{in}}||G_{\text{out}}|-\sum_{u=1}^ND_{uu}D_u\right)\frac{n_1(n_1-1)(n_1-2)}{N(N-1)(N-2)},\\
& \bE(R_{\text{in},1}R_{\text{out},2}) \\
=& \sum_{\substack{(i,j)\in G_{\text{in}},i\notin\C_s,j\notin\C_t \\(s,t)\in G_{\text{out}}}}P(g_i=g_j=1,~g_s=g_t=2) \\
=& \left(|G_{\text{in}}||G_{\text{out}}|-\sum_{u=1}^ND_{uu}D_u\right)\frac{n_1n_2(n_2-1)}{N(N-1)(N-2)}.
\end{align*}
\end{proof}

\subsection{Rewrite $S_R$}\label{pf:SR}
\begin{lemma}
The statistic $S_R$ can be rewritten in the following form
$$S_R=(Z_{\text{out},w}, Z_{\text{out},d}, Z_{\text{in}}) \bOmega^{-1}(Z_{\text{out},w}, Z_{\text{out},d}, Z_{\text{in}})^T,$$ 
 where $\bOmega=\bvar((Z_{\text{out},w}, Z_{\text{out},d}, Z_{\text{in}})^T)$. 
\end{lemma}
\begin{proof}
Let
 \[
    \mathbf{R}=\begin{pmatrix}
      R_{\tout,1}-\bE(R_{\tout,1}) \\
      R_{\tout,2}-\bE(R_{\tout,2}) \\
      R_{\tin,1}-\bE(R_{\tin,1})
    \end{pmatrix}, 
    \]
     \begin{align*}
    \mathbf{Z} 
    = & \begin{pmatrix}
                   Z_{\tout,m} \\
                   Z_{\tout,d}\\
                   Z_{\tin}
                 \end{pmatrix}\\
     = & \begin{pmatrix}
                                 \frac{n_2-1}{\sqrt{\bvar((n_2-1)R_{\tout,1}+(n_1-1)R_{\tout,2})}} & \frac{n_1-1}{\sqrt{\bvar((n_2-1)R_{\tout,1}+(n_1-1)R_{\tout,2})}} & 0 \\
                                 \frac{1}{\sqrt{\bvar(R_{\tout,1}-R_{\tout,2})}} & -\frac{1}{\sqrt{\bvar(R_{\tout,1}-R_{\tout,2})}} & 0 \\
                                 0 & 0 & \frac{1}{\sqrt{\bvar(R_{\tin,1})}}
                               \end{pmatrix}\mathbf{R}\\
    \triangleq &\mathbf B\mathbf{R}.
  \end{align*}
  It is easy to see that $\mathbf B$ is invertible. From the definition of $S$, it can be written as
  \[
  S=\mathbf{R}^T\mathbf\Sigma^{-1}\mathbf{R}=(\mathbf B^{-1}\mathbf{Z})^T\mathbf\Sigma^{-1}(\mathbf B^{-1}\mathbf{Z})=\mathbf{Z}^T(\mathbf B\mathbf\Sigma \mathbf B^T)^{-1}\mathbf{Z}.
  \]
  Here, $\mathbf B\mathbf\Sigma \mathbf B^T$ is the variance of $(Z_{\text{out},w}, Z_{\text{out},d}, Z_{\text{in}})^T$.
 \end{proof}

\subsection{Proof of Theorem \ref{th:asym}}\label{pf:asym}
Applying Stein's method, we prove $(R_{\tout,1},$ $R_{\tout,2},R_{\tin,1})^T$ converges in distribution to a trivariate Gaussian distribution as $N\rightarrow\infty$ first. Consider sums of the form $W=\sum_{i\in\J}\xi_i,$ where $\J$ is an index set and $\xi_i$ are random variables with $\bE(\xi_i)=0,$ and $\bE(W^2)=1.$ The following assumption restricts the dependence between $\{\xi_i:i\in\J\}$.

  \begin{assumption}\label{a}
    [\cite{chen2005stein}, p. 17] For each $i\in\J$ there exists $S_i\subset T_i\subset \J$ such that $\xi_i$ is independent of $\xi_{S_i^c}$ and $\xi_{S_i}$ is independent of $\xi_{T_i^c}$.
  \end{assumption}

  We will use the following theorem.

  \begin{theorem}\label{th:app}
    [\cite{chen2005stein}, Theorem 3.4] Under Assumption \ref{a}, we have
    \[
    \sup_{h\in Lip(1)}|\bE h(W)-\bE h(Z)|\leq\delta
    \]
    where $Lip(1)=\{h:\mathbb{R}\rightarrow\mathbb{R},~\parallel h'\parallel\leq 1\}$, $Z$ has $\mathcal{N}(0,1)$ distribution and
    \[
    \delta=2\sum_{i\in\J}(\bE|\xi_i\eta_i\theta_i|+|\bE(\xi_i\eta_i)|\bE|\theta_i|)+\sum_{i\in\J}\bE|\xi_i\eta_i^2|
    \]
    with $\eta_i=\sum_{j\in S_i}\xi_j$ and $\theta_i=\sum_{j\in T_i}\xi_j,$ where $S_i$ and $T_i$ are defined in Assumption \ref{a}.
  \end{theorem}

  To prove Theorem \ref{th:asym}, we take one step back to study the statistic under the bootstrap null distribution, which is defined as follows: For each subject, we assign it to be from sample 1 with probability $n_1/N$, and from sample 2 with probability $n_2/N$, independently of other subjects. Let $n_1^B$ be the number of subjects that are assigned to be from sample 1. Then, conditioning on $n_1^B=n_1,$ the bootstrap null distribution becomes the permutation null distribution. We use $\bP_B,\bE_B,\bvar_B$ to denote the probability, expectation, and variance under the bootstrap null distribution, respectively.
  
 Let
  \[
  \bE(R_{\tout,1})\triangleq \mu_1,\quad \bE(R_{\tout,2})\triangleq \mu_2,\quad \bE(R_{\tin,1})\triangleq \mu_3,
  \]
  \[
  \bvar(R_{\tout,1})\triangleq (\sigma_1)^2,\quad \bvar(R_{\tout,2})\triangleq (\sigma_2)^2,\quad \bvar(R_{\tin,1})\triangleq (\sigma_3)^2,
  \]
  \[
\bcov(R_{\tout,1},R_{\tout,2})\triangleq \sigma_{12},\quad \bcov(R_{\tout,1},R_{\tin,1})\triangleq \sigma_{13},\quad \bcov(R_{\tout,2},R_{\tin,1})\triangleq \sigma_{23}.   
  \]
  Let $p_n={n_1}/{N},q_n={n_2}/{N}$, then $\lim_{n\rightarrow\infty}p_n=p,\lim_{n\rightarrow\infty}q_n=q$.
  Using the similar steps as in the Proof \ref{pf:expression}, we have
 \begin{align*}
& \bE_B(R_{\tout,1}) = |G_{\tout}|p_n^2\triangleq \mu_1^B, \\
& \bE_B(R_{\tout,2}) = |G_{\tout}|q_n^2\triangleq \mu_2^B, \\
& \bE_B(R_{\tin,1}) = |G_{\tin}|p_n\triangleq \mu_3^B, \\
& \bvar_B(R_{\tout,1}) = p_n^2q_n^2\left\{\frac{1}{2}\sum_{u\neq v}D_{uv}^2+\frac{p_n}{q_n}\sum_uD_u^2\right\}\triangleq (\sigma_1^B)^2, \\
& \bvar_B(R_{\tout,1}) = p_n^2q_n^2\left\{\frac{1}{2}\sum_{u\neq v}D_{uv}^2+\frac{q_n}{p_n}\sum_uD_u^2\right\}\triangleq (\sigma_2^B)^2, \\
&\bvar_B(R_{\tin,1}) = p_nq_n\sum_uD_{uu}^2\triangleq (\sigma_3^B)^2.
\end{align*}

Let
  \begin{align*}
    & W_1^B=\frac{R_{\tout,1}-\mu_1^B}{\sigma_1^B},\quad W_1=\frac{R_{\tout,1}-\mu_1}{\sigma_1}, \\
    & W_2^B=\frac{R_{\tout,2}-\mu_2^B}{\sigma_2^B},\quad W_2=\frac{R_{\tout,2}-\mu_2}{\sigma_2},  \\
    & W_3^B=\frac{R_{\tin,1}-\mu_3^B}{\sigma_3^B},\quad W_3=\frac{R_{\tin,1}-\mu_3}{\sigma_3}, \\
    & W_4^B=\frac{n_1^B-Np_n}{\sigma_0}, \text{ where }\sigma_0^2=Np_n(1-p_n).
  \end{align*}
  Under the conditions of Theorem \ref{th:asym}, as $N\rightarrow\infty$, we can prove the following results:
  \begin{enumerate}[(1)]
    \item $(W_1^B,W_2^B,W_3^B,W_4^B)$ becomes multivariate Gaussian distributed under the bootstrap null.
    \item \begin{align*}
   & \frac{\sigma_1^B}{\sigma_1}\rightarrow c_1,\quad\frac{\mu_1^B-\mu_1}{\sigma_1^B}\rightarrow 0;\quad\frac{\sigma_2^B}{\sigma_2}\rightarrow c_2,\quad\frac{\mu_2^B-\mu_2}{\sigma_2^B}\rightarrow 0;\\
   & \frac{\sigma_3^B}{\sigma_3}\rightarrow c_3,\quad\frac{\mu_3^B-\mu_3}{\sigma_3^B}\rightarrow 0,
    \end{align*}
    where $c_1$, $c_2$ and $c_3$ are constants.
    \item \begin{align*}
    & |\lim_{N\rightarrow\infty}\bcor(W_1,W_2)|<1, \quad|\lim_{N\rightarrow\infty}\bcor(W_1,W_3)|<1, \\
    & |\lim_{N\rightarrow\infty}\bcor(W_2,W_3)|<1.
    \end{align*}
  \end{enumerate}

  From (1) and given that $\bvar_B(W_4^B)=1,$ the conditional distribution of $(W_1^B,W_2^B,W_3^B)^T$ given $W_4^B$ is a trivariate Gaussian distribution under the bootstrap null distribution as $N\rightarrow\infty.$ Since the permutation null distribution is equivalent to the bootstrap null distribution given $W_4^B=0,$ $(W_1^B,W_2^B,W_3^B)^T$ follows a trivariate Gaussian distribution under the permutation null distribution as $N\rightarrow\infty.$ Since
  \begin{align*}
  & W_1=\frac{\sigma_1^B}{\sigma_1}\left(W_1^B+\frac{\mu_1^B-\mu_1}{\sigma_1^B}\right),\quad W_2=\frac{\sigma_2^B}{\sigma_2}\left(W_2^B+\frac{\mu_2^B-\mu_2}{\sigma_2^B}\right),\\
  & W_3=\frac{\sigma_3^B}{\sigma_3}\left(W_3^B+\frac{\mu_3^B-\mu_3}{\sigma_3^B}\right),
  \end{align*}
  given (2), we have $(W_1,W_2,W_3)^T$ follows a trivariate Gaussian distribution under the permutation null distribution as $N\rightarrow\infty.$ Together with (3), we can conclude that $(R_{\tout,1},R_{\tout,2},R_{\tin,1})^T$ converges in distribution to a trivariate Gaussian distribution as $N\rightarrow\infty$. In the following, we prove the results (1)---(3).

To prove (1), by Cram{\'e}r-Wold device, we only need to show that $W=a_1W_1^B+a_2W_2^B+a_3W_3^B+a_4W_4^B$ is asymptotically Gaussian distributed for any combination of $a_1,a_2,a_3,a_4$ such that $\bvar_B(W)>0$.

We first define more notations. For an edge $(u,v)$ of $G_\tout$, i.e., $uv\in\J_1=\{uv:u<v,(u,v)\in G_\tout\}$, let
\[
\xi_{uv}=a_1\frac{I(g_u=g_v=1)-p_n^2}{\sigma_1^B} + a_2\frac{I(g_u=g_v=2)-q_n^2}{\sigma_2^B}.
\]
For an edge $(k,l)$ of $G_\tin$, i.e., $kl\in\J_2=\{kl:k<l,(k,l)\in G_\tin\}$, let
\[
\xi_{kl}=a_3\frac{I(g_k=g_l=1)-p_n}{\sigma_3^B}.
\]
And for a subject $s\in\J_3=\{1,\dots,N\}$, let
  \[
  \xi_s=a_4\frac{I(g_{s}=1)-p_n}{\sigma_0}.
  \]
Thus,
\[
W = a_1W_1^B+a_2W_2^B+a_3W_3^B+a_4W_4^B= \sum_{i\in\J}\xi_i,
\]
where $\J=\J_1\cup\J_2\cup\J_3$.

For an edge $e=(u,v)\in G_\tout$, let
\[
C_{\tout,e} = \{u,v\}\cup\{s\in\{1,\dots,N\}: s \text{ is connected to }u \text{ or }v \text{ in }G_\tout\}.
\]
For an edge $e=(i,j)\in G_\tin,~i,j\in\C_u$, let
\[
C_{\tin,e} = \{u\}\cup\{s\in\{1,\dots,N\}: s \text{ is connected to }u \text{ in }G_\tout\}.
\]
We define
\[
C_e=C_{\tout,e}I(e\in G_\tout)+C_{\tin,e}I(e\in G_\tin)
\]
and introduce following index sets to satisfy Assumption \ref{a}. 
For $i\in\J_1$ (i.e., $i$ is an edge $(u,v)\in G_\tout$), let
\begin{align*}
& S_i = A_i\cup\{u,v\}, \\
& T_i = B_i\cup C_i.
\end{align*}
For $i\in\J_2\cup\J_3$ (i.e., $i$ is an edge $(k,l)\in G_\tin,~k,l\in \C_s$ or $i$ is a subject $s\in\J_3$),
\begin{align*}
& S_i = A_{i}\cup\{s\}, \\
& T_i = B_i\cup C_i.
\end{align*}
Then $S_i$ and $T_i$ satisfy Assumption \ref{a}.

Let $a=\max\{|a_1|,|a_2|,|a_3|,|a_4|\}$ and $\sigma=\min\{\sigma_1^B,\sigma_2^B,\sigma_3^B,\sigma_0\}$. Since
\[|\xi_i| \leq
\begin{cases}
{2a}/{\sigma}& \text{ if }i\in\J_1, \\
a/\sigma& \text{ if }i\in\J_2\cup\J_3
\end{cases},
\]
we have
\[
\sum_{j\in S_i}|\xi_j|\leq (|A_i|+2)2a/\sigma,\quad \sum_{j\in T_i}|\xi_j|\leq (|B_i|+|C_i|)2a/\sigma,\quad i\in\J,
\]
and the terms $\bE|\xi_i\eta_i\theta_i|$, $|\bE(\xi_i\eta_i)|\bE|\theta_i|$ and $\bE|\xi_i\eta_i^2|$ are all bounded by
\[
\frac{32a^3}{\sigma^3}|A_i||B_i|.
\]
So we have
  \begin{align*}
    \delta = & \frac{1}{\sqrt{(\bvar_B(W))^3}}\Bigg\{2\sum_{i\in\J}(\bE_B|\xi_i\eta_i\theta_i|+|\bE_B(\xi_i\eta_i)|\bE_B|\theta_i|)+\sum_{i\in\J}\bE_B|\xi_i\eta_i^2|\Bigg\} \\
    \leq & \frac{1}{\sqrt{(\bvar_B(W))^3}}\frac{160a^3}{\sigma^3}\sum_{i\in\J}|A_i||B_i| \\
    \leq & \frac{480a^3}{\sigma^3\sqrt{(\bvar_B(W))^3}}\sum_{e\in G_\tout}|A_{\tout,e}||B_{\tout,e}|
  \end{align*}
Since $480a^3/\sqrt{(\bvar_B(W))^3}$ is of constant order and $\sigma=O(N^{0.5})$, we have $\delta\rightarrow 0$ when $\sum_{e\in G_\tout}|A_{\tout,e}||B_{\tout,e}|=o(N^{1.5})$.

Next we prove result (2). Since $|G_\tout|=O(N)$ and $\sum_uD_u^2-4|G_\tout|^2/N=O(N)$, let $\lim_{N\rightarrow\infty}|G_\tout|/N=b_1$ and $\lim_{N\rightarrow\infty}\sum_uD_u^2/N-4|G_\tout|^2/N^2=b_2$, $b_1,b_2\in(0,\infty)$. Then $\lim_{N\rightarrow\infty}\sum_uD_u^2/N=4b_1^2+b_2$, i.e., $\sum_uD_u^2=O(N)$. Since 
\[
2|G_\tout|=\sum_{u\neq v}D_{uv}\leq\sum_{u\neq v}D_{uv}^2\leq\sum_uD_u^2,
\]
we have $\sum_{u\neq v}D_{uv}^2=O(N)$ and let $\lim_{N\rightarrow\infty}\sum_{u\neq v}D_{uv}^2/N=b_3\in(0,\infty)$.
Hence, as $N\rightarrow\infty$
\begin{align*}
& \frac{(\sigma_1^B)^2}{N}\longrightarrow p^2q^2\left\{\frac{1}{2}b_3+\frac{p}{q}(4b_1^2+b_2)\right\}, \\
& \frac{(\sigma_1)^2}{N}\longrightarrow p^2q^2\left\{\frac{1}{2}b_3+\frac{p}{q}b_2\right\}, \\
& \frac{\sigma_1^B}{\sigma_1}\longrightarrow\sqrt{1+\frac{8pb_1^2}{qb_3+2pb_2}}.
\end{align*}
Similarly, we have
\[
\frac{\sigma_2^B}{\sigma_2}\longrightarrow\sqrt{1+\frac{8qb_1^2}{pb_3+2qb_2}}.
\]
Since $|G_\tin|=O(N)$ and $\sum_uD_{uu}^2-|G_\tin|^2/N=O(N)$, let $\lim_{N\rightarrow\infty}|G_\tin|/N=b_4$ and $\lim_{N\rightarrow\infty}\sum_uD_{uu}^2/N-|G_\tin|^2/N^2=b_5$, $b_4,b_5\in(0,\infty)$. Then $\lim_{N\rightarrow\infty}\sum_uD_{uu}^2/N=b_4^2+b_5$, and we have
\begin{align*}
& \frac{(\sigma_3^B)^2}{N}\longrightarrow pq(b_4^2+b_5), \\
& \frac{(\sigma_3)^2}{N}\longrightarrow pqb_5, \\
& \frac{\sigma_3^B}{\sigma_3}\longrightarrow\sqrt{1+\frac{b_4^2}{b_5}}.
\end{align*}

Also,
\[
\mu_1^B-\mu_1 = |G_\tout|\frac{n_1n_2}{N^2(N-1)},
\]
so
\[
\lim_{N\rightarrow\infty}\frac{\mu_1^B-\mu_1}{\sigma_1^B} = 0.
\]
Similarly, we have
\[
\lim_{N\rightarrow\infty}\frac{\mu_2^B-\mu_2}{\sigma_2^B} = 0.
\]
Since $\mu_3^B-\mu_3 = 0$,
\[
\lim_{N\rightarrow\infty}\frac{\mu_3^B-\mu_3}{\sigma_3^B} = 0.
\]

Last, we prove result (3). As $N\rightarrow\infty$,
\[
\frac{\sigma_{12}}{N}\longrightarrow p^2q^2\left\{\frac{1}{2}b_3-b_2\right\}.
\]
We have
\begin{align*}
\lim_{N\rightarrow\infty}\bcor(W_1,W_2) & = \lim_{N\rightarrow\infty}\frac{\sigma_{12}}{\sqrt{(\sigma_1)^2(\sigma_2)^2}} \\
& = \frac{b_3-2b_2}{\sqrt{(b_3-2b_2)^2+2b_2b_3/pq}},
\end{align*}
Strictly positive $2b_2b_3/pq$ implies $|\lim_{N\rightarrow\infty}\bcor(W_1,W_2)|<1$.
Note that
\[
|G_\tin|=\sum_uD_{uu}\leq\sum_uD_{uu}D_u\leq\sqrt{\sum_uD_{uu}^2}\sqrt{\sum_uD_u^2},
\]
$|G_\tin|=O(N)$, $\sum_uD_{uu}^2=O(N)$ and $\sum_uD_{u}^2=O(N)$.
We have $\sum_uD_{uu}D_u=O(N)$ and let $\lim_{N\rightarrow\infty}\sum_uD_{uu}D_u/N=b_6\in(0,\infty)$. Thus,
\[
\frac{\sigma_{13}}{N}\longrightarrow p^2q(b_6-2b_1b_4),
\]
and
\begin{align*}
\lim_{N\rightarrow\infty}\bcor(W_1,W_3) & = \lim_{N\rightarrow\infty}\frac{\sigma_{13}}{\sqrt{(\sigma_1)^2(\sigma_3)^2}} \\
& = \frac{p^2q(b_6-2b_1b_4)}{\sqrt{\frac{1}{2}p^3q^3b_3b_5+p^4q^2b_2b_5}},
\end{align*}
Note that
\[
\left|\frac{\sum_{u=1}^ND_{uu}D_u-\frac{2}{N}|G_{\tin}||G_{\tout}|}{\sqrt{(\sum_{u=1}^ND_u^2-\frac{4|G_{\tout}|^2}{N})(\sum_{u=1}^ND_{uu}^2-\frac{|G_{\tin}|^2}{N})}}\right| = |\bcor(Z_{\tout,d},Z_{\tin})| \leq 1,
\]
so 
\[
\left|\frac{b_6-2b_1b_4}{\sqrt{b_2b_5}}\right| = \left|\lim_{N\rightarrow\infty}\frac{\frac{\sum_{u=1}^ND_{uu}D_u}{N}-\frac{2}{N^2}|G_{\tin}||G_{\tout}|}{\sqrt{(\sum_{u=1}^N\frac{D_u^2}{N}-\frac{4|G_{\tout}|^2}{N^2})(\sum_{u=1}^N\frac{D_{uu}^2}{N}-\frac{|G_{\tin}|^2}{N^2})}}\right|\leq 1.
\]
We have $|\lim_{N\rightarrow\infty}\bcor(W_1,W_3)|< 1$, since $p^3q^3b_3b_5/2$ is strictly positive. Similarly, we obtain $|\lim_{N\rightarrow\infty}\bcor(W_2,W_3)|< 1$.

Since $(Z_{\tout,w},Z_{\tout,d},Z_\tin)^T$ is the linear transformation of $(R_{\tout,1},R_{\tout,2},R_{\tin,1})^T$, the proof above implies $(Z_{\tout,w},Z_{\tout,d},Z_\tin)^T$ converges in distribution to a trivariate Gaussian distribution with mean $\bzero$ and covariance matrix $\Gamma$. As shown in Proof \ref{pf:SR}, the linear transformation is nondegenerate. Thus $\Gamma$ is invertible and $\Gamma=\begin{psmallmatrix}
1 & 0 & 0 \\
0 & 1 & \rho_Z\\
0 & \rho_Z & 1
\end{psmallmatrix}$, where 
\[
\rho_Z = \lim_{N\rightarrow\infty}\frac{\sum_{u=1}^ND_{uu}D_u-\frac{2}{N}|G_{\tin}||G_{\tout}|}{\sqrt{(\sum_{u=1}^ND_u^2-\frac{4|G_{\tout}|^2}{N})(\sum_{u=1}^ND_{uu}^2-\frac{|G_{\tin}|^2}{N})}} = \frac{b_6-2b_1b_4}{\sqrt{b_2b_5}}
\]
and $|\rho_Z|<1$.

\clearpage
\section{Additional simulation results for data with an exponentially decayed within-subject correlation}
\begin{enumerate}[(1)]
\item Simulations for one-dimensional density, $p = 1$.
\begin{table}[!htbp]
	\caption{Parameter values for 5 different simulation settings for comparison one-dimensional density functions.} \label{simu1T2}
	\begin{tabular}{l}
		\hline
 A1':  null model.\\
  $\qquad\rho_1=0.6,~\beta_1=0,~\epsilon_1=1,~\nu_{11}=1,~\nu_{12}=2;$ \\
 $\qquad\rho_2=0.6,~\beta_2=0,~\epsilon_2=1,~\nu_{21}=1,~\nu_{22}=2;~\sigma=1$.\\
 A2': within-subject variability difference in $\rho$.
\\
   $\qquad\rho_1=0,~\beta_1=0,~\epsilon_1=1,~\nu_{11}=1,~\nu_{12}=1.2;$ \\
 $\qquad\rho_2=0.9,~\beta_2=0,~\epsilon_2=1,~\nu_{21}=1,~\nu_{22}=1.2;~\sigma=1$.
\\
A3': between-subject mean difference in $\beta$ and $\nu_{\cdot 1}+\nu_{\cdot 2}$.
\\
 $\qquad\rho_1=0,~\beta_1=0,~\epsilon_1=1,~\nu_{11}=1,~\nu_{12}=1.2;$ \\
 $\qquad\rho_2=0,~\beta_2=0.7,~\epsilon_2=1,~\nu_{21}=0.96,~\nu_{22}=1.16;~\sigma=1$.
\\
A4': between-subject variability difference in $\epsilon$ and $\nu_{\cdot 2}-\nu_{\cdot 1}$.
\\
  $\qquad\rho_1=0,~\beta_1=0,~\epsilon_1=1,~\nu_{11}=1,~\nu_{12}=1.3;$ \\
 $\qquad\rho_2=0,~\beta_2=0,~\epsilon_2=1.1,~\nu_{21}=0.97,~\nu_{22}=1.33;~\sigma=1$.
\\ 
A5': within-subject variability difference in $\rho$, between-subject mean difference in $\beta$ \\ and $\nu_{\cdot 1}+\nu_{\cdot 2}$, 
variance difference in $\epsilon$ and $\nu_{\cdot 2}-\nu_{\cdot 1}$.
\\
  $\qquad\rho_1=0,~\beta_1=0,~\epsilon_1=1,~\nu_{11}=1,~\nu_{12}=1.3;$ \\
 $\qquad\rho_2=0.35,~\beta_2=0.5,~\epsilon_2=1.1,~\nu_{21}=0.97,~\nu_{22}=1.36;~\sigma=1$.
 \\
 \hline
\end{tabular}
\end{table}

\begin{table}[!htp]
  \centering
  \caption{Empirical power of the proposed test statistics in the first 6 columns, generalized edge-count test ($S1$, $S2$) and Fr\'{e}chet test (Fretest$1$, Fretest$2$) at 0.05 significance level under the five scenarios denoted by A1'--A5'. Those above 95 percentage of the best power under A2'--A5' are in bold.}\label{tab:den1T2} 
   \renewcommand\arraystretch{1}
 \begin{tabular*}{1.0\textwidth}{@{\extracolsep{\fill}} lcccccccccc}
    \hline
 & $T_{\tin}$ & $Z_{\tout,w}$ & $T_{\tout,d}$ & $M_{\tout}$ & $S_R$ & $M$ & $S1$ & $S2$ & Fretest1 & Fretest2 \\
 & \multicolumn{10}{c}{Null model}\\
   A1' & 0.057 & 0.055 & 0.046 &  0.061 & 0.055 & 0.063 & 0.056 & 0.047 & 0.048 & 0.054 \\
\hline
 & \multicolumn{10}{c}{Alternative model}\\ 
  A2' & \textbf{0.966} & 0.050 & 0.149 & 0.111 & 0.851& 0.901 & 0.153 & \textbf{1.000} & 0.463 & 0.056 \\
A3' & 0.036 & \textbf{0.970} & 0.071& \textbf{0.960} & 0.914 & \textbf{0.949} & 0.640 & 0.555 & 0.293 & 0.288 \\
A4' & 0.053 & 0.186 & \textbf{0.908} & 0.859 & 0.790 & 0.814 & 0.140 & 0.062 & 0.304 & 0.277 \\
A5' & 0.067 & 0.697 & \textbf{0.993} & \textbf{0.995} & \textbf{0.986} & \textbf{0.991} & 0.398 & 0.391& 0.467 & 0.333 \\
    \hline
  \end{tabular*} 
 \end{table}
 \clearpage
 \item Simuations for moderate-dimensional density, $p = 30$.
 \begin{table}[!htbp]
	\caption{Parameter values for 5 different simulation settings for comparison 30-dimensional density functions.} \label{simu2}
	\begin{tabular}{l}
		\hline
 B1':  null model.\\
  $\qquad\rho_1=0.3,~\beta_1=\bzero_p,~\epsilon_1=1,~\nu_{11}=1,~\nu_{12}=2;$ \\
 $\qquad\rho_2=0.3,~\beta_2=\bzero_p,~\epsilon_2=1,~\nu_{21}=1,~\nu_{22}=2;~\sigma=1$.\\
 B2': within-subject variability difference in $\rho$.
\\
   $\qquad\rho_1=0,~\beta_1=\bzero_p,~\epsilon_1=1,~\nu_{11}=1,~\nu_{12}=1.3;$ \\
 $\qquad\rho_2=0.25,~\beta_2=\bzero_p,~\epsilon_2=1,~\nu_{21}=1,~\nu_{22}=1.3;~\sigma=1$.
\\
B3': between-subject mean difference in $\beta$ and $\nu_{\cdot 1}+\nu_{\cdot 2}$.
\\
 $\qquad\rho_1=0,~\beta_1=\bzero_p,~\epsilon_1=1,~\nu_{11}=1,~\nu_{12}=1.3;$ \\
 $\qquad\rho_2=0,~\beta_2=0.1\bone_p,~\epsilon_2=1,~\nu_{21}=1.2,~\nu_{22}=1.5;~\sigma=1$.
\\
B4': between-subject variability difference in $\epsilon$ and $\nu_{\cdot 2}-\nu_{\cdot 1}$.
\\
  $\qquad\rho_1=0,~\beta_1=\bzero_p,~\epsilon_1=1,~\nu_{11}=1,~\nu_{12}=1.3;$ \\
 $\qquad\rho_2=0,~\beta_2=\bzero_p,~\epsilon_2=1.1,~\nu_{21}=0.8,~\nu_{22}=1.5;~\sigma=1$.
\\ 
B5': within-subject variability difference in $\rho$, between-subject mean difference in $\beta$ \\ and $\nu_{\cdot 1}+\nu_{\cdot 2}$, 
variance difference in $\epsilon$ and $\nu_{\cdot 2}-\nu_{\cdot 1}$.
\\
  $\qquad\rho_1=0,~\beta_1=\bzero_p,~\epsilon_1=1,~\nu_{11}=1,~\nu_{12}=1.3;$ \\
 $\qquad\rho_2=0.22,~\beta_2=0.2\bone_p,~\epsilon_2=1.03,~\nu_{21}=1,~\nu_{22}=1.4;~\sigma=1$.
 \\
 \hline
\end{tabular}
\end{table}

 \begin{table}[!htp]
\linespread{1.2} 
  \centering
\caption{Empirical power of the proposed test statistics in the first 6 columns, generalized edge-count test ($S1$, $S2$) and Fr\'{e}chet test (Fretest$1$, Fretest$2$) at 0.05 significance level under the five scenarios denoted by B1'--B5'. Those above 95 percentage of the best power under B2'--B5' are in bold.}\label{tab:den2T2}
  \renewcommand\arraystretch{1}
 \begin{tabular*}{1.0\textwidth}{@{\extracolsep{\fill}} lcccccccccc}
    \hline
 & $T_{\tin}$ & $Z_{\tout,w}$ & $T_{\tout,d}$ & $M_{\tout}$ & $S_R$ & $M$ & $S1$ & $S2$ & Fretest1 & Fretest2 \\
  & \multicolumn{10}{c}{Null  model}\\
  B1' & 0.041 & 0.058 & 0.039 &  0.045 & 0.049 & 0.051 & 0.047 & 0.050 & 0.050 & 0.066  \\ \hline
 & \multicolumn{10}{c}{Alternative model}\\
  B2' & \textbf{0.934} & 0.054 & 0.046 & 0.043 & 0.838 & 0.852 & 0.212 & 0.054 & 0.485 & 0.079 \\
B3' & 0.057 & \textbf{0.945} & 0.049 & \textbf{0.916} & 0.812 & 0.893 & 0.772 & 0.717 & 0.118 & 0.181  \\
B4' & 0.144 & 0.241 & \textbf{0.882} &  0.835 & 0.774 & 0.806 & 0.720 & 0.800 & \textbf{0.859} & \textbf{0.883} \\
B5' & \textbf{0.892} & 0.668 & 0.131 & 0.596 & \textbf{0.911} & \textbf{0.906}& 0.756 & 0.483 & 0.853 & 0.461  \\
    \hline
  \end{tabular*} 
 \end{table}
\end{enumerate}

\clearpage
\section{Additional results on comparison of asymptotic $p$-values and permutation $p$-values}

We  next  examine whether the asymptotic $p$-values are close to the  $p$-values obtained based on 10,000 permutations. We consider the data generating procedure in Section \ref{sec:density} with the parameters set as those in model (A5). In particular, under the chosen sample sizes $n_1$=$n_2$=25, 50, 75 and 100, data were generated over 100 simulation runs. For each run, we estimate the asymptotic $p$-values and permutation $p$-values over 10,000 permutations. 

Figure \ref{fig:power1} shows the empirical powers of each test estimated by the asymptotic and permutation $p$-values over the 100 runs for $n_1$=$n_2$=$25,~50,~75,~100$.  The results show that the power obtained by the asymptotic $p$-value is very close to that based on the permutation $p$-value for all the proposed test statistics. As sample size increases, the results are almost identical as expected. 
We also provide the comparison of asymptotic $p$-value and permutation $p$-value over 100 simulation runs in the Supplementary Material. 

\begin{figure}[!h]
	\centering
	\includegraphics[width=0.95\textwidth]{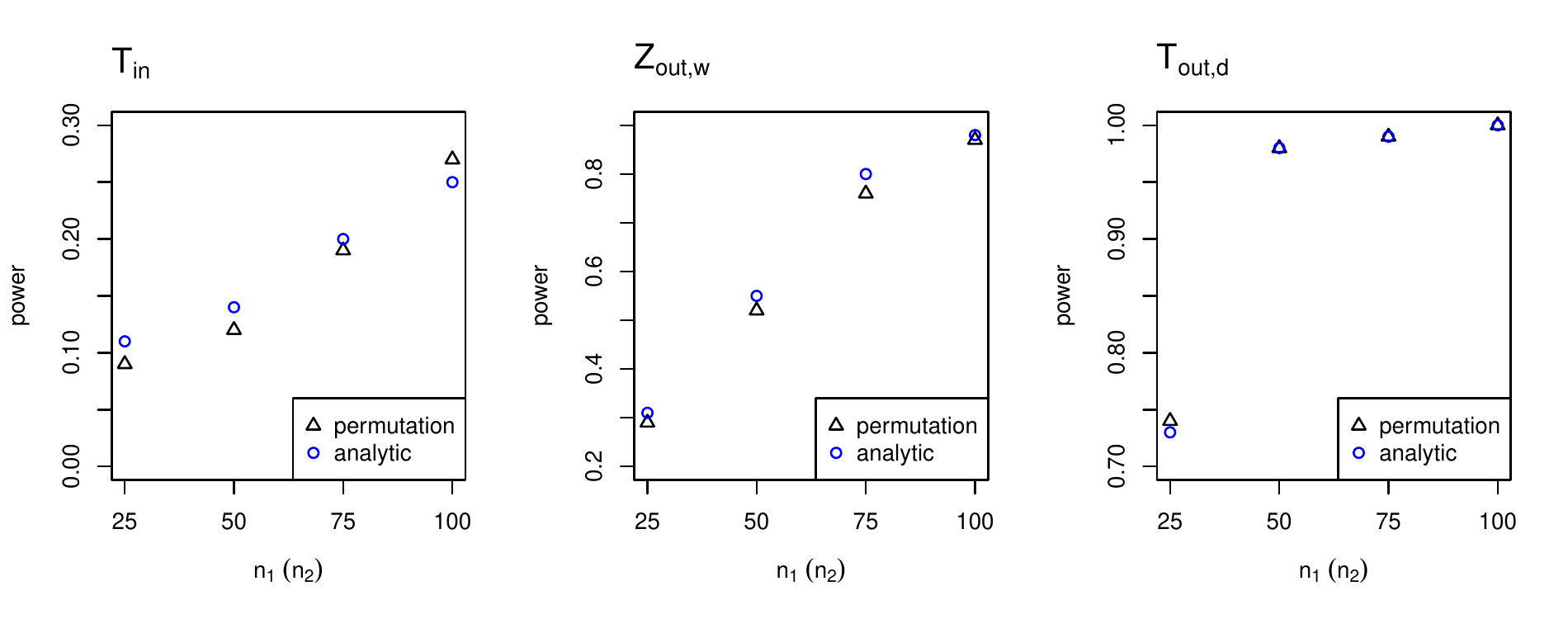}
	
	\vspace{-0.3cm}
	
	\includegraphics[width=0.99\textwidth]{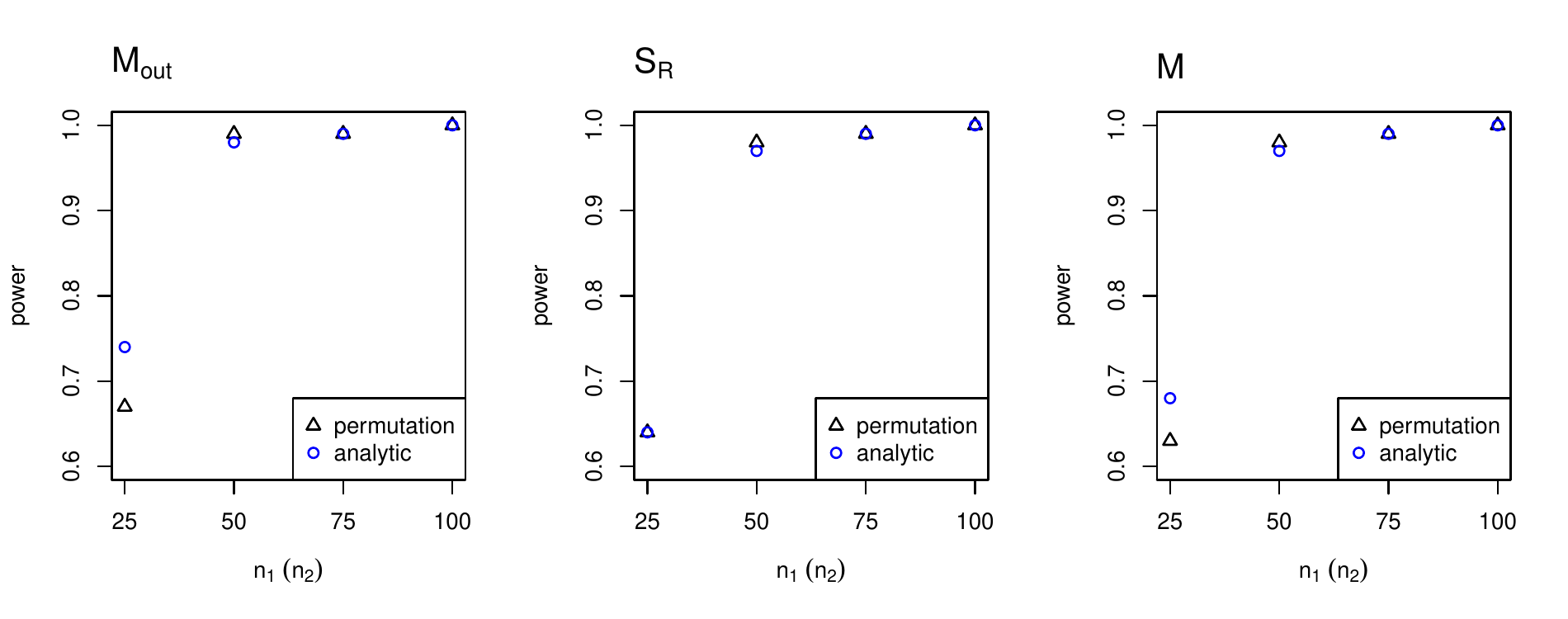}
	
	\vspace{-0.5cm}
	\caption{Comparison of the empirical power estimated by the asymptotic $p$-value and by the $p$-value calculated from 10,000 permutations for different proposed test statistics.}\label{fig:power1}
\end{figure}

\ignore{
	\begin{table}[h]
		\footnotesize
		\linespread{1} 
		\centering
		\caption{Average run time (in seconds) in obtaining the critical value through the asymptotic result (Asym), 10,000 random permutations (Perm1), and 100,000 random permutations (Perm2) under varying sample sizes based on 100 simulations.}\label{tab:time}
		\vspace{0.3cm}
		\begin{tabular*}{0.8\textwidth}{@{\extracolsep{\fill}} lrrr}
			\hline
			& Asym & Perm1 & Perm2 \\ \hline
			$n_1=n_2=25$  & 1.20 & 61.71 & 622.12  \\
			$n_1=n_2=50$ & 1.30  & 140.08 &  1420.04 \\ 
			$n_1=n_2=75$ & 2.70 & 269.76 & 2732.75  \\ 
			$n_1=n_2=100$ & 4.34 & 433.52  & 4259.72 \\ \hline
		\end{tabular*} 
	\end{table}
	
	The computational burden to implement the tests via the random permutations depends on the number of permutations drawn. Table \ref{tab:time} shows a typical running time in obtaining the critical value through the asymptotic result, 10,000 random permutations, and 100,000 random permutations.  All computations are done on a laptop with an Intel Core i5 2.3 GHz processor.
	
}

\begin{figure}[!h]
    \centering
    \includegraphics[width=0.9\textwidth]{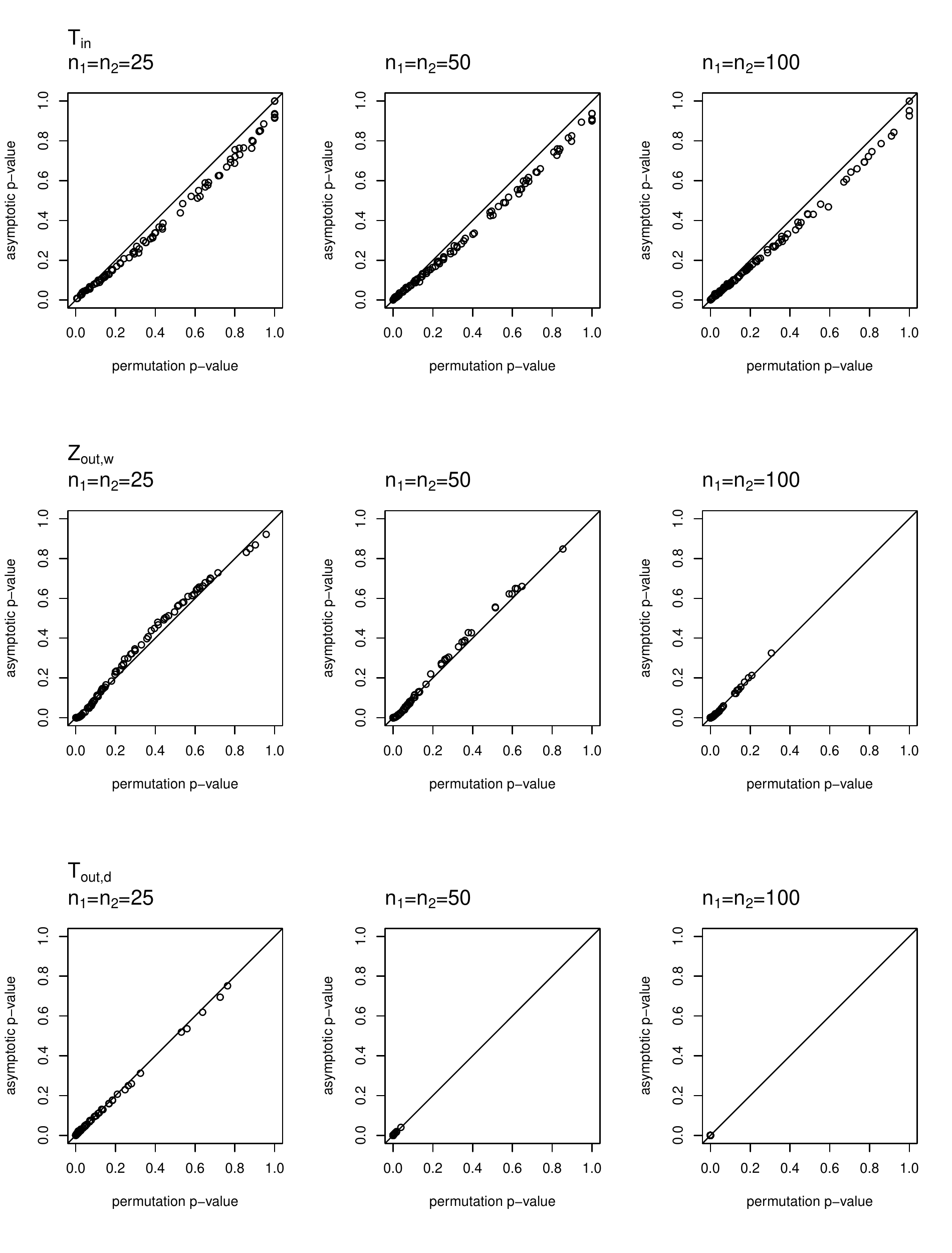}
        
      \vspace{-0.3cm}
       
     \caption{Compare the asymptotic $p$-value with the $p$-value calculated from 10,000 permutations with 100 simulations for test statistics $T_\tin,~Z_{\tout,w}$ and $T_{\tout,d}$.}\label{fig:pval1}
\end{figure}

\begin{figure}[!h]
    \centering
     \includegraphics[width=0.9\textwidth]{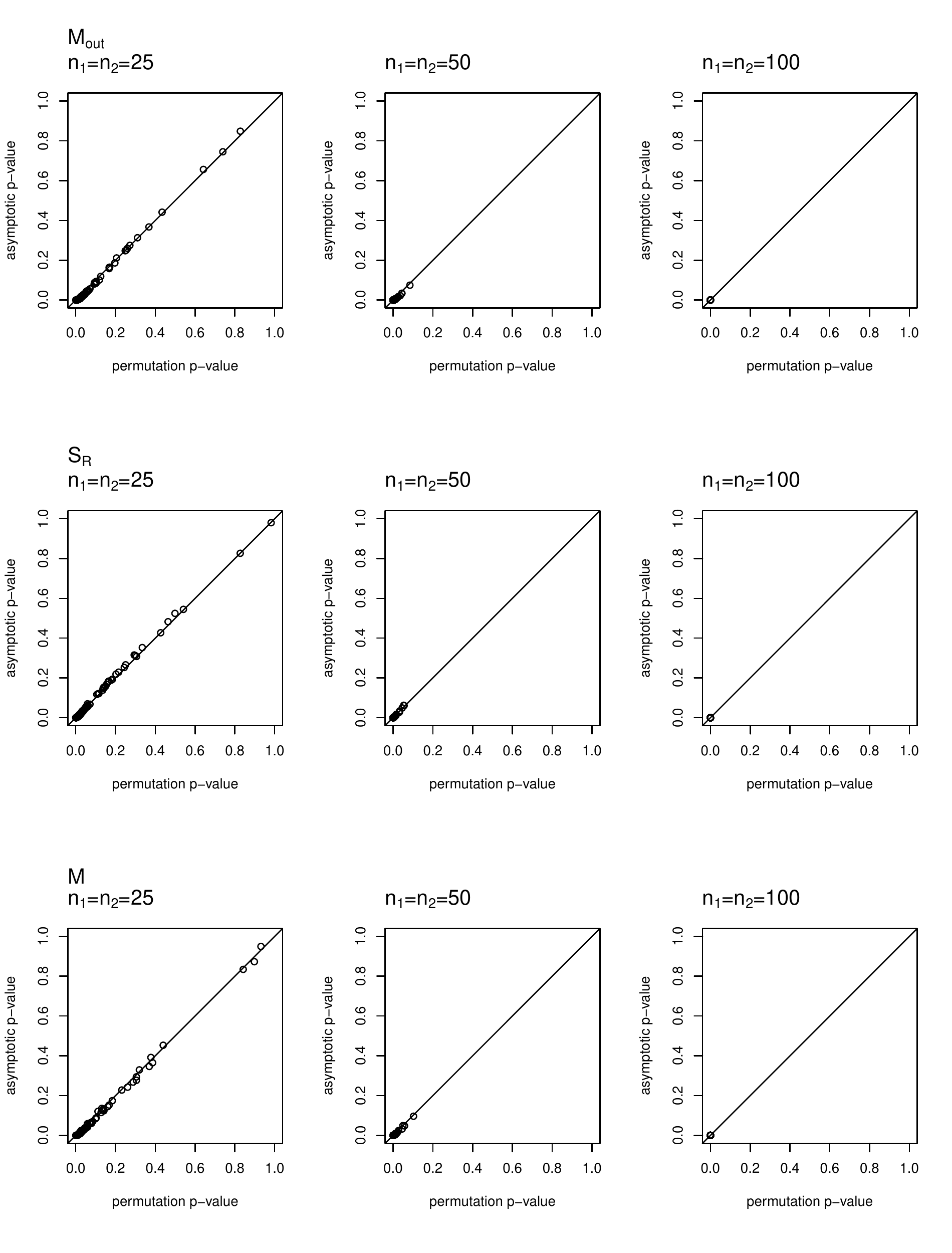}
      
     \vspace{-0.3cm}
 
    \caption{Compare the asymptotic $p$-value with the $p$-value calculated from 10,000 permutations with 100 simulations for test statistics  $M_\tout(1.14)$, $S_R$ and $M(1,1.14)$.}\label{fig:pval2}
\end{figure}

Figures \ref{fig:pval1} and \ref{fig:pval2} show the comparison of asymptotic $p$-value and permutation $p$-value over 100 simulation runs. The results show that the asymptotic $p$-values are very similar to that based on the permutation $p$-value for all the proposed test statistics. As sample size increases, the results are almost identical as expected. 

\clearpage

\section{Detailed $p$-values in real data analysis}

\begin{table}[!htp]
	\linespread{1.2} 
	\centering
	\caption{Comparisons of activity distributions among the controls, MDD, BPI and BPII patients. The $p$-values are presented for the proposed test statistics, the generalized edge-count tests ($S1$, $S2$) and Fr\'{e}chet tests (Fretest1, Fretest2) (bold for those $<$0.05).}\label{tab:real1}
	\renewcommand\arraystretch{1}
	\begin{tabular}{lcccccccccc}
		\hline
		& \multicolumn{10}{c}{weekdays, $n_1=117,~n_2=106,~n_3=26,~n_4=32$.}\\
		& $T_{\tin}$ & $Z_{\tout,w}$ & $T_{\tout,d}$ & $M_{\tout}$ & $S_R$ & $M$ & $S1$ & $S2$ & Fretest1 & Fretest2 \\
		HC vs MDD & 0.729 &  0.625 &  0.307 & 0.428 & 0.567 & 0.566   & 0.472     & 0.171 & 0.381 &  0.305 \\
		HC vs BPII &  0.777 &  0.442 & 0.238 &  0.345 & 0.550 & 0.488 & 0.305     & 0.053 & 0.661 & 0.231\\
		HC vs BPI &  0.770  &   0.061 &  0.867 &   0.140 & 0.405 & 0.205  & 0.734 & 0.734 &  0.262  & 0.105 \\
		MDD vs BPII &  0.708 & 0.464 & 0.067 &  0.116 & 0.155  &    0.169 &0.243 & 0.133 &  0.708 & 0.769 \\
		MDD vs BPI & 0.746 &  0.119 &   0.617 &  0.272 & 0.539 &  0.377 &0.693 & 0.680 & 0.703 & 0.491 \\
		BPII vs BPI & 0.718 & 0.055 & 0.119 &  0.107 & 0.088  & 0.156 & 0.286 &  0.151 & 0.123 &  0.127 \\
		& \multicolumn{10}{c}{weekends, $n_1=116,~n_2=98,~n_3=30,~n_4=33$.}\\
		& $T_{\tin}$ & $Z_{\tout,w}$ & $T_{\tout,d}$ & $M_{\tout}$ & $S_R$ & $M$ & $S1$ & $S2$ & Fretest1 & Fretest2 \\
		HC vs MDD & 0.404  &  0.623 &   0.700 & 0.780 & 0.695 & 0.632 &  0.471 & 0.710 & 0.061 & 0.060 \\
		HC vs BPII & 0.652 & 0.618 & 0.658 &  0.747 & 0.817 & 0.781 & 0.498 & 0.404 &  0.539 & 0.515 \\
		HC vs BPI & 0.287  &  \textbf{0.003} &  0.068 &  \textbf{0.005} & \textbf{0.005} & \textbf{0.006} & 0.063 & \textbf{0.018} & \textbf{0.014} & \textbf{0.002} \\
		MDD vs BPII & 0.714  &  0.451 &   0.848 &   0.813 &  0.885 &    0.828 & 0.482 & 0.728 & 0.274 & 0.317 \\
		MDD vs BPI &  0.744   & 0.085 &    0.317 & 0.161 & 0.297 & 0.237  & 0.154 &  0.133 & 0.480 & 0.288 \\
		BPII vs BPI & 0.699 & \textbf{0.002} &  0.250 &  \textbf{0.004} &  \textbf{0.017} &  \textbf{0.005} & \textbf{0.005} & 0.054 & \textbf{0.007} & \textbf{0.005} \\
		\hline
	\end{tabular} 
\end{table}

\begin{table}[!htp]
	\linespread{1.0} 
	\centering
	\caption{Comparisons of activity distrbutions in different age groups, where C1, C2 and C3 denoting young, middle-aged and older age groups. The $p$-values of the proposed test statistics, the generalized edge-count tests ($S1$, $S2$) and Fr\'{e}chet tests (Fretest1, Fretest2) are presented for different comparisons (bold for those $<$0.05). }\label{tab:real-age1}
	\renewcommand\arraystretch{1}
	\begin{tabular}{lcccccccccc}
		\hline
		& \multicolumn{10}{c}{weekdays}\\
		& $T_{\tin}$ & $Z_{\tout,w}$ & $T_{\tout,d}$ & $M_{\tout}$ & $S_R$ & $M$ & $S1$ & $S2$ & Fretest1 & Fretest2 \\
		HC 36, 55, 25\\
		C1 vs C2 & \textbf{0.020} & 0.139 & 0.513 & 0.314 & \textbf{0.028} & 0.054 & 0.812 & 0.060 & 0.718 & 0.547 \\
		C1 vs C3 & \textbf{0.004}  & \textbf{$<$1e-3} & 0.867 & \textbf{$<$1e-3} & \textbf{$<$1e-3} & \textbf{$<$1e-3} & \textbf{$<$1e-3} & \textbf{$<$1e-3} & 0.078 & 0.116 \\
		C2 vs C3 & 0.201 & \textbf{0.001} & 0.781 & \textbf{0.002} & \textbf{0.007} & \textbf{0.003} & 0.332 &  \textbf{0.008} & 0.272 & 0.369 \\
		MDD 21, 62, 20  \\
		C1 vs C2 & 0.383 & 0.295 & 0.879 & 0.680 & 0.639 & 0.596 & 0.543 & 0.617 &  0.053 &   0.060 \\
		C1 vs C3 & 0.160 & \textbf{0.012} & 0.841 & \textbf{0.024} & \textbf{0.044} & \textbf{0.033} &  \textbf{0.048} & 0.144 &  \textbf{0.011} &  \textbf{0.033} \\
		C2 vs C3 & 0.546 & 0.464 & 0.725 & 0.769 & 0.850 & 0.746 & 0.745 & 0.819 &  0.862 & 0.674\\
		&  \multicolumn{10}{c}{weekends}\\ 
		& $T_{\tin}$ & $Z_{\tout,w}$ & $T_{\tout,d}$ & $M_{\tout}$ & $S_R$ & $M$ & $S1$ & $S2$ & Fretest1 & Fretest2 \\
		HC 36, 53, 26 \\
		C1 vs C2 & 0.722 &  0.069 & 0.349 &  0.147 &  0.261 & 0.215 & 0.120 & 0.241 & 0.655 & 0.159  \\
		C1 vs C3 & 0.317 & \textbf{0.001} & \textbf{0.013} & \textbf{$<$1e-3}  & \textbf{$<$1e-3} &  \textbf{0.001} &  \textbf{0.047} &  \textbf{0.004} & 0.558 & 0.083 \\
		C2 vs C3 & 0.478 & \textbf{0.027} & \textbf{0.042} & \textbf{0.035} & \textbf{0.017} &  \textbf{0.050} & 0.671 & 0.090 & 0.770 & 0.483\\
		MDD 18, 55, 22 \\
		C1 vs C2 & 0.408 & 0.468 & 0.276 & 0.375 & 0.394 & 0.408 & 0.426 &  0.210 &  0.358 & 0.604 \\
		C1 vs C3 &  0.533 & \textbf{0.025} &  0.155 & \textbf{0.047} & 0.065 & 0.070 & 0.145 & 0.110 & 0.073 & 0.097\\
		C2 vs C3 & 0.633 &  0.526 & 0.684 & 0.726 & 0.767 & 0.726 & 0.868 &  0.691 & 0.543 & 0.355 \\
		\hline
	\end{tabular} 
\end{table}

\ignore{
	\begin{figure}[!h]
		\centering
		\includegraphics[width=0.48\textwidth]{realAgesepHeatin1_1.pdf}
		\includegraphics[width=0.48\textwidth]{realAgesepHeatout1_1.pdf}

		\includegraphics[width=0.48\textwidth]{realAgesepHeatin2_1.pdf}
		\includegraphics[width=0.48\textwidth]{realAgesepHeatout2_1.pdf}

		\includegraphics[width=0.48\textwidth]{realAgesepHeatin3_1.pdf}
		\includegraphics[width=0.48\textwidth]{realAgesepHeatout3_1.pdf}
		
		\caption{Left: within-subject distance of C1 vs C3 in HC; Right: heatmaps of distance matrix of C1 vs C3 in HC under the second problem (P2).}\label{fig:realAgesepX1}
	\end{figure}
}

\begin{table}[!htp]
	\linespread{1} 
	\centering
	\caption{Comparisons of activity distrbutions in different BMI groups, where D1 and D2 denoting lean and obese people, respectively. The $p$-values of the proposed test statistics, the generalized edge-count tests ($S1$, $S2$) and Fr\'{e}chet tests (Fretest1, Fretest2) are presented for different comparisons (bold for those $<$0.05).}\label{tab:real-bmi1}
	\renewcommand\arraystretch{1}
	\begin{tabular}{lcccccccccc}
		\hline
		&  \multicolumn{10}{c}{weekdays}\\
		& $T_{\tin}$ & $Z_{\tout,w}$ & $T_{\tout,d}$ & $M_{\tout}$ & $S_R$ &  $M$ & $S1$ & $S2$ & Fretest1 & Fretest2 \\
		HC 32, 48 \\
		D1 vs D2 & 0.815 &  0.837 &    0.487  &  0.621 &  0.645 & 0.776 & 0.809  &  0.788        & 0.529   &  0.390 \\   
		OTHER  43, 75 \\
		D1 vs D2 & 0.756  & \textbf{0.009}  &   0.163  &  \textbf{0.017} &  \textbf{0.040} & \textbf{0.022}& 0.384    &  0.614      &  0.826     &   0.858 \\
		& \multicolumn{10}{c}{weekends}\\
		& $T_{\tin}$ & $Z_{\tout,w}$ & $T_{\tout,d}$ & $M_{\tout}$ & $S_R$ &  $M$ & $S1$ & $S2$ & Fretest1 & Fretest2 \\
		HC 30, 49\\
		D1 vs D2 & 0.579 & 0.718 & 0.366  &   0.493  & 0.490  & 0.566 & 0.636    &  0.375        &0.262      &  0.168 \\
		OTHER 40, 74\\
		D1 vs D2 &  0.715 &    \textbf{0.030} &    0.223 &     0.063 &   0.133 &  0.094 & 0.372  &    0.803     &   0.533    &    0.673 \\
		\hline
	\end{tabular} 
\end{table}

\ignore{
	\begin{figure}[!h]
		\centering
		\includegraphics[width=0.48\textwidth]{realBmisepHeatin1_1.pdf}
		\includegraphics[width=0.48\textwidth]{realBmisepHeatout1_1.pdf}

		\includegraphics[width=0.48\textwidth]{realBmisepHeatin2_1.pdf}
		\includegraphics[width=0.48\textwidth]{realBmisepHeatout2_1.pdf}

		\includegraphics[width=0.48\textwidth]{realBmisepHeatin3_1.pdf}
		\includegraphics[width=0.48\textwidth]{realBmisepHeatout3_1.pdf}
		
		\caption{Left: within-subject distance of D1 vs D3 in HC; Right: heatmaps of distance matrix of D1 vs D3 in HC under the third problem (P3).}    
		\label{fig:realBmisepX1}
	\end{figure}
}

\clearpage

\section{Results of real application when using 5-MST and 15-MST as the similarity graph}
\subsection{5-MST with the 2-Wasserstein distance as the similarity graph}

\begin{table}[!h]
\linespread{1} 
  \centering
  \caption{Comparisons of activity distributions among the controls, MDD, BPI and BPII patients. The $p$-values are presented for the proposed test statistics, the generalized edge-count tests ($S1$, $S2$) and Fr\'{e}chet tests (Fretest1, Fretest2) (bold for those $<$0.05).}\label{tab:real1-5mst}
 \renewcommand\arraystretch{1}
 \begin{tabular}{lcccccccccc}
    \hline
 & \multicolumn{10}{c}{weekdays, $n_1=117,~n_2=106,~n_3=26,~n_4=32$.}\\
    & $T_{\tin}$ & $Z_{\tout,w}$ & $T_{\tout,d}$ & $M_{\tout}$ & $S_R$ & $M$ & $S1$ & $S2$ & Fretest1 & Fretest2 \\
    HC vs MDD & 0.701 & 0.803 & 0.658 & 0.769 & 0.674 & 0.807 & 0.426     & 0.212 & 0.384 & 0.307 \\
      HC vs BPII &  0.728 & 0.380 & 0.392 & 0.485 & 0.751 & 0.643 & 0.164     & 0.056 & 0.661 & 0.230\\
    HC vs BPI &  0.592 & 0.192 & 0.730 & 0.433  & 0.657 &  0.523 & 0.611     & 0.754 &  0.262  &  0.104 \\
     MDD vs BPII &  0.834 &  0.390 & 0.134 & 0.210  & 0.354 & 0.303 & 0.176 & 0.181 &  0.711 & 0.767 \\
  MDD vs BPI & 0.773 &  0.263 & 0.544 & 0.496 &  0.734 &  0.636 & 0.620     & 0.577 & 0.705 & 0.485 \\
  BPII vs BPI & 0.712 & 0.137 & 0.260 & 0.250  & 0.265 & 0.348 & 0.508     & 0.116 &  0.124  &  0.127 \\
  & \multicolumn{10}{c}{weekends, $n_1=116,~n_2=98,~n_3=30,~n_4=33$.}\\
    & $T_{\tin}$ & $Z_{\tout,w}$ & $T_{\tout,d}$ & $M_{\tout}$ & $S_R$ & $M$ & $S1$ & $S2$ & Fretest1 & Fretest2 \\
    HC vs MDD &  0.541 & 0.655 & 0.838 & 0.877 & 0.788 & 0.779 & 0.640 & 0.676 & 0.052 & 0.060 \\
      HC vs BPII &  0.405 & 0.491 & 0.725 & 0.723 & 0.691 &  0.592 & 0.596 & 0.403 & 0.542 & 0.515 \\
    HC vs BPI & 0.206 & \textbf{0.021} & 0.092 & \textbf{0.031} & \textbf{0.023} & \textbf{0.036} & 0.275 & 0.077 &  \textbf{0.014} &  \textbf{0.002} \\
     MDD vs BPII &0.759 & 0.431 &  0.830 &  0.782 & 0.905 &  0.845 &  0.530 &  0.718 &  0.279 & 0.322 \\
  MDD vs BPI &  0.654 & 0.176 &  0.276 & 0.237 &  0.378 & 0.334 & 0.354 & 0.196 & 0.477 &  0.288 \\
  BPII vs BPI & 0.720 & \textbf{0.030} &  0.670 &  0.065 & 0.209 &  0.095 &  \textbf{0.026} & 0.090 & \textbf{0.007} & \textbf{0.005} \\
  \hline
  \end{tabular} 
 \end{table}

 \begin{table}[!htp]
\linespread{1.0} 
  \centering
  \caption{Comparisons of activity distrbutions in different age groups, where C1, C2 and C3 denoting young, middle-aged and older age groups. The $p$-values of the proposed test statistics, the generalized edge-count tests ($S1$, $S2$) and Fr\'{e}chet tests (Fretest1, Fretest2)  are presented for different comparisons (bold for those $<$0.05). }\label{tab:real-age1-5mst}
 \renewcommand\arraystretch{1}
 \begin{tabular}{lcccccccccc}
    \hline
   & \multicolumn{10}{c}{weekdays}\\
  & $T_{\tin}$ & $Z_{\tout,w}$ & $T_{\tout,d}$ & $M_{\tout}$ & $S_R$ & $M$ & $S1$ & $S2$ & Fretest1 & Fretest2 \\
HC 36, 55, 25\\
C1 vs C2  &     0.109 &       0.074 &       0.663 &      0.169 &    \textbf{$<$1e-3} &        0.142 &       0.473 &       \textbf{0.044} &         0.709 &         0.559 \\
 C1 vs C3 &     \textbf{0.005} &    \textbf{$<$1e-3} &       0.723 &       \textbf{$<$1e-3} &     \textbf{$<$1e-3} &   \textbf{$<$1e-3} &       \textbf{0.005} &      \textbf{$<$1e-3} &         0.076 &         0.116 \\
 C2 vs C3  &     0.143 &       \textbf{$<$1e-3}  &       0.725 &      \textbf{0.002} &    \textbf{0.004} &       \textbf{ 0.002} &       0.395 &       \textbf{0.044} &         0.272 &         0.369 \\
MDD 21, 62, 20  \\
C1 vs C2 &     0.418 &       0.320 &       0.868 &      0.710 &    0.674 &        0.623 &       0.649 &       0.527 &         0.053 &         0.060 \\
C1 vs C3 &     0.182 &       \textbf{0.029} &       0.800 &      0.062 &    0.100 &        0.079 &       0.126 &       0.182 &         \textbf{0.011} &         \textbf{0.033} \\
C2 vs C3 &     0.611 &       0.368 &       0.746 &      0.697 &    0.842 &        0.730 &       0.746 &       0.778 &         0.861 &         0.675 \\
 &  \multicolumn{10}{c}{weekends}\\ 
& $T_{\tin}$ & $Z_{\tout,w}$ & $T_{\tout,d}$ & $M_{\tout}$ & $S_R$ & $M$ & $S1$ & $S2$ & Fretest1 & Fretest2 \\
HC 36, 53, 26 \\
 C1 vs C2  &     0.656 &       0.093 &       0.587 &      0.201 &    0.356 &        0.278 &       0.317 &       0.215 &         0.661 &         0.157 \\
C1 vs C3 &     0.597 &       \textbf{$<$1e-3}  &       \textbf{0.027} &       \textbf{$<$1e-3}  &     \textbf{$<$1e-3}  &       \textbf{ 0.002} &       \textbf{0.024} &       \textbf{0.003} &         0.553 &         0.082 \\
C2 vs C3 &     0.688 &       0.109 &       \textbf{0.048} &      0.068 &    0.078 &        0.100 &       0.471 &       0.106 &         0.767 &         0.480 \\
MDD 18, 55, 22 \\
C1 vs C2  &     0.558 &       0.542 &       0.308 &      0.416 &    0.471 &        0.484 &       0.340 &       0.197 &         0.348 &         0.603 \\
C1 vs C3 &     0.490 &       0.069 &       0.424 &      0.143 &    0.197 &        0.186 &       0.200 &       0.087 &         0.075 &         0.099 \\
C2 vs C3 &     0.716 &       0.580 &       0.447 &      0.553 &    0.656 &        0.640 &       0.801 &       0.588 &         0.553 &         0.359 \\
 \hline
  \end{tabular} 
 \end{table}
 
 \begin{table}[!htp]
\linespread{1} 
  \centering
  \caption{Comparisons of activity distrbutions in different BMI groups, where D1 and D2 denoting lean and obese people, respectively. The $p$-values of the proposed test statistics, the generalized edge-count tests ($S1$, $S2$) and Fr\'{e}chet tests (Fretest1, Fretest2) are presented for different comparisons (bold for those $<$0.05).}\label{tab:real-bmi1-5mst}
 \renewcommand\arraystretch{1}
 \begin{tabular}{lcccccccccc}
    \hline
  &  \multicolumn{10}{c}{weekdays}\\
 & $T_{\tin}$ & $Z_{\tout,w}$ & $T_{\tout,d}$ & $M_{\tout}$ & $S_R$ &  $M$ & $S1$ & $S2$ & Fretest1 & Fretest2 \\
HC 32, 48 \\
D1 vs D2 & 0.803   &   0.920  &    0.242  &  0.354 &  0.258 &    0.485 & 0.796      & 0.726    &    0.533    &    0.395 \\   
OTHER  43, 75 \\
 D1 vs D2 & 0.793  &    \textbf{0.011}  &    0.128  &   \textbf{0.020} &   \textbf{0.034} &    \textbf{0.029} &      0.439  &    0.451  &      0.826     &   0.857 \\
  & \multicolumn{10}{c}{weekends}\\
  & $T_{\tin}$ & $Z_{\tout,w}$ & $T_{\tout,d}$ & $M_{\tout}$ & $S_R$ &  $M$ & $S1$ & $S2$ & Fretest1 & Fretest2 \\
HC 30, 49\\
 D1 vs D2 & 0.631   &   0.332  &    0.401 &    0.444 &  0.542    &   0.544      & 0.595  &    0.424    &    0.257   &     0.163 \\
OTHER 40, 74\\
  D1 vs D2 &  0.692    &    0.105  &      0.121    &    0.142   &     0.204        & 0.208    &    0.464    &    0.680    &    0.504   &     0.657   \\
\hline
  \end{tabular} 
 \end{table}
\clearpage 
 
\subsection{15-MST with the 2-Wasserstein distance as the similarity graph}

\begin{table}[!h]
\linespread{1} 
  \centering
  \caption{Comparisons of activity distributions among the controls, MDD, BPI and BPII patients. The $p$-values are presented for  the proposed test statistics, the generalized edge-count tests ($S1$, $S2$) and Fr\'{e}chet tests (Fretest1, Fretest2) (bold for those $<$0.05).}\label{tab:real1-15mst}
 \renewcommand\arraystretch{1}
 \begin{tabular}{lcccccccccc}
    \hline
 & \multicolumn{10}{c}{weekdays, $n_1=117,~n_2=106,~n_3=26,~n_4=32$.}\\
    & $T_{\tin}$ & $Z_{\tout,w}$ & $T_{\tout,d}$ & $M_{\tout}$ & $S_R$ & $M$ & $S1$ & $S2$ & Fretest1 & Fretest2 \\
    HC vs MDD & 0.801 & 0.596 & 0.231 & 0.341 & 0.597 & 0.482 & 0.644 & 0.158 & 0.388 & 0.309 \\
      HC vs BPII &  0.773 & 0.248 & 0.089 & 0.148 &  0.237 & 0.219 & 0.706 & 0.057 &  0.657 & 0.226\\
    HC vs BPI &  0.759 & 0.056 & 0.886 & 0.055 & 0.222 & 0.081 & 0.781 & 0.623 & 0.257 & 0.103 \\
     MDD vs BPII &  0.540 & 0.451 & 0.055 & 0.064 & 0.059 & 0.094 & 0.574 & 0.089 & 0.717 & 0.762 \\
  MDD vs BPI &  0.787 & 0.058 & 0.755 & 0.107 & 0.333 & 0.154 & 0.746 & 0.778 & 0.696 &  0.482 \\
  BPII vs BPI & 0.354 & 0.051 & 0.147 & 0.051 & 0.058 & 0.059 & 0.113 & 0.189 & 0.122 & 0.126 \\
  & \multicolumn{10}{c}{weekends, $n_1=116,~n_2=98,~n_3=30,~n_4=33$.}\\
    & $T_{\tin}$ & $Z_{\tout,w}$ & $T_{\tout,d}$ & $M_{\tout}$ & $S_R$ & $M$ & $S1$ & $S2$ & Fretest1 & Fretest2 \\
    HC vs MDD &  0.329 & 0.390 & 0.392 & 0.468 & 0.513 & 0.443 & 0.459 & 0.699 & 0.051 & 0.059 \\
      HC vs BPII &  0.754 &  0.442 & 0.414 & 0.502 & 0.718 & 0.650 & 0.481 & 0.393 & 0.538 & 0.509 \\
    HC vs BPI & 0.508 & \textbf{0.001} & \textbf{0.031} & \textbf{0.001} & \textbf{0.001} & \textbf{0.002} & \textbf{0.007} & \textbf{0.005} & \textbf{0.013} &  \textbf{0.002} \\
     MDD vs BPII &0.319 & 0.348 & 0.848 & 0.711 & 0.609 & 0.499 & 0.545 & 0.748 & 0.276 & 0.317 \\
  MDD vs BPI &  0.620 & \textbf{0.038} & 0.218 & 0.073 & 0.122 & 0.107 & 0.072 & 0.078 & 0.482 & 0.288 \\
  BPII vs BPI & 0.768 & \textbf{$<1$e-3} & 0.168 & \textbf{0.001} & \textbf{0.003} & \textbf{0.001} & \textbf{0.003} & \textbf{0.006} & \textbf{0.007} & \textbf{0.005} \\
  \hline
  \end{tabular} 
 \end{table}
 
 \begin{table}[!htp]
\linespread{1.0} 
  \centering
  \caption{Comparisons of activity distrbutions in different age groups, where C1, C2 and C3 denoting young, middle-aged and older age groups. The $p$-values of the proposed test statistics, the generalized edge-count tests ($S1$, $S2$) and Fr\'{e}chet tests (Fretest1, Fretest2) are presented for different comparisons (bold for those $<$0.05). }\label{tab:real-age1-15mst}
 \renewcommand\arraystretch{1}
 \begin{tabular}{lcccccccccc}
    \hline
   & \multicolumn{10}{c}{weekdays}\\
  & $T_{\tin}$ & $Z_{\tout,w}$ & $T_{\tout,d}$ & $M_{\tout}$ & $S_R$ & $M$ & $S1$ & $S2$ & Fretest1 & Fretest2 \\
HC 36, 55, 25\\
C1 vs C2  &     \textbf{0.023} &       0.073 &       0.368 &      0.165 &    \textbf{0.020} &        0.055 &       0.859 &       0.083 &         0.717 &         0.551 \\
 C1 vs C3 &     \textbf{0.001} &       \textbf{$<1$e-3} &       0.577 &      \textbf{$<1$e-3} &    \textbf{$<1$e-3} &        \textbf{$<1$e-3} &       \textbf{$<1$e-3} &       \textbf{$<1$e-3} &         0.076 &         0.115 \\
 C2 vs C3  &     0.074 &       \textbf{$<1$e-3} &       0.726 &      \textbf{0.001} &    \textbf{0.001} &        \textbf{0.001} &       \textbf{0.011} &       \textbf{0.003} &         0.268 &         0.366 \\
MDD 21, 62, 20  \\
C1 vs C2 &     0.250 &       0.259 &       0.885 &      0.613 &    0.467 &        0.456 &       0.470 &       0.755 &         0.053 &         0.061 \\
C1 vs C3 &    0.059 &       \textbf{0.006} &       0.885 &      \textbf{0.010} &    \textbf{0.006} &       \textbf{ 0.013} &       \textbf{0.005} &       \textbf{0.020} &         \textbf{0.011 }&         \textbf{0.033} \\
C2 vs C3 &    0.399 &       0.401 &       0.815 &      0.781 &    0.772 &        0.644 &       0.705 &       0.863 &         0.858 &         0.677 \\
 &  \multicolumn{10}{c}{weekends}\\ 
& $T_{\tin}$ & $Z_{\tout,w}$ & $T_{\tout,d}$ & $M_{\tout}$ & $S_R$ & $M$ & $S1$ & $S2$ & Fretest1 & Fretest2 \\
HC 36, 53, 26 \\
 C1 vs C2  &     0.455 &       0.059 &       0.145 &      0.091 &    0.089 &        0.129 &       0.189 &       0.307 &         0.670 &         0.158 \\
C1 vs C3 &     \textbf{0.050} &       \textbf{0.002} &       \textbf{0.005} &      \textbf{0.002} &    \textbf{$<1$e-3} &        \textbf{0.003} &       0.187 &       \textbf{0.021} &         0.556 &         0.083 \\
C2 vs C3 &     0.311 &       \textbf{0.042} &       0.058 &      0.056 &    \textbf{0.026} &        0.076 &       0.722 &       0.096 &         0.769 &         0.481 \\
MDD 18, 55, 22 \\
C1 vs C2  &     0.181 &       0.395 &       0.295 &      0.395 &    0.285 &        0.289 &       0.837 &       0.452 &         0.350 &         0.607 \\
C1 vs C3 &    0.370 &       \textbf{0.039} &       0.116 &      0.068 &    0.057 &        0.096 &       \textbf{0.005} &       \textbf{0.021} &         0.071 &         0.095 \\
C2 vs C3 &    0.624 &       0.369 &       0.865 &      0.762 &    0.856 &        0.776 &       0.906 &       0.699 &         0.545 &         0.359 \\
 \hline
  \end{tabular} 
 \end{table}
 
 \begin{table}[!htp]
\linespread{1} 
  \centering
  \caption{Comparisons of activity distrbutions in different BMI groups, where D1 and D2 denoting lean and obese people, respectively. The $p$-values of the proposed test statistics, the generalized edge-count tests ($S1$, $S2$) and Fr\'{e}chet tests (Fretest1, Fretest2) are presented for different comparisons (bold for those $<$0.05).}\label{tab:real-bmi1-15mst}
 \renewcommand\arraystretch{1}
 \begin{tabular}{lcccccccccc}
    \hline
  &  \multicolumn{10}{c}{weekdays}\\
 & $T_{\tin}$ & $Z_{\tout,w}$ & $T_{\tout,d}$ & $M_{\tout}$ & $S_R$ &  $M$ & $S1$ & $S2$ & Fretest1 & Fretest2 \\
HC 32, 48 \\
D1 vs D2 & 0.843 &   0.717 &  0.649   & 0.767 & 0.874 &   0.891 & 0.886      & 0.826   &     0.531  &      0.394 \\   
OTHER  43, 75 \\
 D1 vs D2 & 0.793  &  \textbf{0.006} &    0.263  &  \textbf{0.011} & \textbf{0.043} &    \textbf{0.016} &   0.496      & 0.833    &    0.825   &     0.856 \\
  & \multicolumn{10}{c}{weekends}\\
  & $T_{\tin}$ & $Z_{\tout,w}$ & $T_{\tout,d}$ & $M_{\tout}$ & $S_R$ &  $M$ & $S1$ & $S2$ & Fretest1 & Fretest2 \\
HC 30, 49\\
 D1 vs D2 & 0.609   &  0.861 &   0.496  &   0.628  & 0.480  &     0.689      & 0.672   &   0.340      &  0.256    &    0.167 \\
OTHER 40, 74\\
  D1 vs D2 & 0.763 &  \textbf{0.031} &    0.422 &    0.065 &   0.195 &  0.098 &  0.690 &  0.673 &   0.482 &   0.643  \\
\hline
  \end{tabular} 
 \end{table}
 \clearpage
 
\subsection{5-MST with the maximum mean discrepancy as the similarity graph}

\begin{table}[!h]
\linespread{1} 
  \centering
  \caption{Comparisons of activity distributions among the controls, MDD, BPI and BPII patients. The $p$-values are presented for  the proposed test statistics, the generalized edge-count tests ($S1$, $S2$) and Fr\'{e}chet tests (Fretest1, Fretest2) (bold for those $<$0.05).}\label{tab:real1-5mstMMD}
 \renewcommand\arraystretch{1}
 \begin{tabular}{lcccccccccc}
    \hline
 & \multicolumn{10}{c}{weekdays, $n_1=117,~n_2=106,~n_3=26,~n_4=32$.}\\
    & $T_{\tin}$ & $Z_{\tout,w}$ & $T_{\tout,d}$ & $M_{\tout}$ & $S_R$ & $M$ & $S1$ & $S2$ & Fretest1 & Fretest2 \\
HC vs MDD &     0.730 &       0.810 &       0.603 &      0.714 &    0.629 &        0.774 &       0.718 &       0.504 &         0.382 &         0.307 \\
HC vs BPII &     0.735 &       0.062 &       0.744 &      0.133 &    0.302 &        0.182 &       0.466 &       0.322 &         0.656 &         0.230 \\
HC vs BPI &     0.644 &       0.205 &       0.774 &      0.435 &    0.634 &        0.507 &       0.768 &       0.758 &         0.259 &         0.103 \\
MDD vs BPII &     0.801 &       0.115 &       0.616 &      0.221 &    0.426 &        0.310 &       0.549 &       0.321 &         0.712 &         0.768 \\
MDD vs BPI &     0.706 &       0.343 &       0.738 &      0.608 &    0.782 &        0.684 &       0.761 &       0.759 &         0.700 &         0.483 \\
BPII vs BPI &     0.708 &       0.053 &       0.672 &      0.071 &    0.177 &        0.101 &       0.600 &       0.306 &         0.121 &         0.124 \\
  & \multicolumn{10}{c}{weekends, $n_1=116,~n_2=98,~n_3=30,~n_4=33$.}\\
    & $T_{\tin}$ & $Z_{\tout,w}$ & $T_{\tout,d}$ & $M_{\tout}$ & $S_R$ & $M$ & $S1$ & $S2$ & Fretest1 & Fretest2 \\
HC vs MDD &     0.511 &       0.912 &       0.733 &      0.824 &    0.447 &        0.696 &       0.614 &       0.743 &         0.062 &         0.060 \\
HC vs BPII &     0.593 &       0.593 &       0.716 &      0.755 &    0.745 &        0.726 &       0.636 &       0.316 &         0.553 &         0.523 \\
HC vs BPI &     0.456 &       0.066 &       0.190 &      0.087 &    0.119 &        0.112 &       0.144 &       0.093 &          \textbf{0.014} &          \textbf{0.002} \\
MDD vs BPII &     0.391 &       0.600 &       0.753 &      0.778 &    0.585 &        0.561 &       0.728 &       0.439 &         0.267 &         0.310 \\
MDD vs BPI &     0.689 &       0.243 &       0.257 &      0.236 &    0.378 &        0.329 &       0.296 &       0.143 &         0.485 &         0.284 \\
BPII vs BPI &     0.409 &        \textbf{0.021} &       0.359 &       \textbf{0.041} &    0.078 &        0.055 &       0.127 &       0.340 &          \textbf{0.007} &          \textbf{0.005} \\
  \hline
  \end{tabular} 
 \end{table}
 
 \begin{table}[!htp]
\linespread{1.0} 
  \centering
  \caption{Comparisons of activity distrbutions in different age groups, where C1, C2 and C3 denoting young, middle-aged and older age groups. The $p$-values of the proposed test statistics, the generalized edge-count tests ($S1$, $S2$) and Fr\'{e}chet tests (Fretest1, Fretest2) are presented for different comparisons (bold for those $<$0.05). }\label{tab:real-age1-5mstMMD}
 \renewcommand\arraystretch{1}
 \begin{tabular}{lcccccccccc}
    \hline
   & \multicolumn{10}{c}{weekdays}\\
  & $T_{\tin}$ & $Z_{\tout,w}$ & $T_{\tout,d}$ & $M_{\tout}$ & $S_R$ & $M$ & $S1$ & $S2$ & Fretest1 & Fretest2 \\
HC 36, 55, 25\\
C1 vs C2 &     0.409 &       \textbf{0.036} &       0.725 &      0.082 &    0.156 &        0.106 &       0.708 &       0.495 &         0.713 &         0.550 \\
C1 vs C3 &     \textbf{0.029} &       \textbf{$<1$e-3} &       0.270 &      \textbf{$<1$e-3} &    \textbf{$<1$e-3} &        \textbf{$<1$e-3} &       \textbf{0.002} &       \textbf{$<1$e-3} &         0.076 &         0.115 \\
C2 vs C3 &     0.206 &       \textbf{0.002} &       0.669 &      \textbf{0.005} &    \textbf{0.010} &        \textbf{0.006} &       0.154 &       0.062 &         0.272 &         0.369 \\
MDD 21, 62, 20  \\
C1 vs C2 &     0.396 &       0.121 &       0.801 &      0.281 &    0.313 &        0.278 &       0.425 &       0.332 &         0.053 &         0.060 \\
C1 vs C3 &     0.206 &       \textbf{0.002} &       0.742 &      \textbf{0.005} &    \textbf{0.012} &        \textbf{0.006} &       0.191 &       0.230 &         \textbf{0.011} &         \textbf{0.032} \\
C2 vs C3 &     0.744 &       0.183 &       0.624 &      0.380 &    0.602 &        0.482 &       0.779 &       0.793 &         0.863 &         0.681 \\
 &  \multicolumn{10}{c}{weekends}\\ 
& $T_{\tin}$ & $Z_{\tout,w}$ & $T_{\tout,d}$ & $M_{\tout}$ & $S_R$ & $M$ & $S1$ & $S2$ & Fretest1 & Fretest2 \\
HC 36, 53, 26 \\
C1 vs C2 &     0.696 &       \textbf{0.012} &       0.721 &      \textbf{0.027} &    0.079 &        \textbf{0.037} &       0.291 &       0.353 &         0.657 &         0.157 \\
C1 vs C3 &     0.612 &       \textbf{0.002} &       0.538 &      \textbf{0.004} &    \textbf{0.011} &        \textbf{0.005} &       \textbf{0.014} &       \textbf{0.009} &         0.554 &         0.081 \\
C2 vs C3 &     0.664 &       0.091 &       0.492 &      0.162 &    0.278 &        0.223 &       0.259 &       0.154 &         0.771 &         0.482 \\
MDD 18, 55, 22 \\
C1 vs C2 &     0.480 &       0.270 &       0.413 &      0.344 &    0.481 &        0.406 &       0.468 &       0.237 &         0.358 &         0.604 \\
C1 vs C3 &     0.646 &       \textbf{0.041} &       0.721 &      0.091 &    0.220 &        0.126 &       0.094 &       0.104 &         0.073 &         0.096 \\
C2 vs C3 &     0.715 &       0.379 &       0.236 &      0.281 &    0.418 &        0.377 &       0.754 &       0.700 &         0.546 &         0.353 \\
 \hline
  \end{tabular} 
 \end{table}
 
 \begin{table}[!htp]
\linespread{1} 
  \centering
  \caption{Comparisons of activity distrbutions in different BMI groups, where D1 and D2 denoting lean and obese people, respectively. The $p$-values of the proposed test statistics, the generalized edge-count tests ($S1$, $S2$) and Fr\'{e}chet tests (Fretest1, Fretest2) are presented for different comparisons (bold for those $<$0.05).}\label{tab:real-bmi1-5mstMMD}
 \renewcommand\arraystretch{1}
 \begin{tabular}{lcccccccccc}
    \hline
  &  \multicolumn{10}{c}{weekdays}\\
 & $T_{\tin}$ & $Z_{\tout,w}$ & $T_{\tout,d}$ & $M_{\tout}$ & $S_R$ &  $M$ & $S1$ & $S2$ & Fretest1 & Fretest2 \\
HC 32, 48 \\
D1 vs D2 &     0.770 &       0.798 &       0.501 &      0.614 &    0.624 &        0.716 &       0.774 &       0.713 &         0.533 &         0.400 \\
D1 vs D2 &     0.794 &       \textbf{0.008} &       0.284 &      \textbf{0.015} &    \textbf{0.038} &        \textbf{0.021} &       0.319 &       0.235 &         0.826 &         0.861 \\
  & \multicolumn{10}{c}{weekends}\\
  & $T_{\tin}$ & $Z_{\tout,w}$ & $T_{\tout,d}$ & $M_{\tout}$ & $S_R$ &  $M$ & $S1$ & $S2$ & Fretest1 & Fretest2 \\
HC 30, 49\\
 D1 vs D2 &     0.630 &       0.467 &       0.559 &      0.584 &    0.658 &        0.635 &       0.697 &       0.681 &         0.258 &         0.168 \\
D1 vs D2 &     0.697   &     0.134   &     0.592    &    0.293    &    0.490       & 0.404   &     0.551     &   0.455    &    0.531     &   0.697 \\
\hline
  \end{tabular} 
 \end{table}
 
\vspace{5cm}

\end{document}